%% file: main.tex
\documentclass[authorcolumns,numberwithinsect,cleveref,thm-restate]{no-lipics-v2022}
\pdfoutput=1
\usepackage{amssymb}
\usepackage{mathtools}
\usepackage{bm}
\usepackage[draft]{fixme}
\usepackage{afterpage}

\usepackage[T1]{fontenc}
\usepackage[utf8]{inputenc}

\usepackage[numbers,square,comma,longnamesfirst]{natbib}

\usetikzlibrary{calc}

\input{macros}

\title{Optimal Algorithms for Bounded Weighted Edit Distance}

\author{Alejandro Cassis}{Saarland University and\\Max Planck Institute for Informatics, SIC,\\Saarbrücken, Germany}{}{}{{This work is part of the project TIPEA that has received funding from the European Research Council (ERC) under the European Unions Horizon 2020 research and innovation programme (grant agreement No. 850979).}}
\author{Tomasz Kociumaka}{Max Planck Institute for Informatics, SIC,\\Saarbrücken, Germany}{}{https://orcid.org/0000-0002-2477-1702}{}
\author{Philip Wellnitz}{Max Planck Institute for Informatics, SIC,\\Saarbrücken, Germany}{}{https://orcid.org/0000-0002-6482-8478}{}

\authorrunning{A. Cassis, T. Kociumaka, and P. Wellnitz}
\acknowledgements{We thank Adam Karczmarz for his guide on the landscape of shortest paths algorithms in planar graphs.}

\nolinenumbers

\allowdisplaybreaks[2] %
\begin{document}
\maketitle
\begin{abstract}
The edit distance (also known as Levenshtein distance) of two strings is the minimum
number of insertions, deletions, and substitutions of characters needed to transform one string into the other.
The textbook dynamic-programming algorithm computes the edit distance of two length-\(n\) strings in $\Oh(n^2)$ time,
which is optimal up to subpolynomial factors and conditioned on the Strong Exponential Time Hypothesis ({\sc seth}).
An established way of circumventing this hardness is to consider the \emph{bounded} setting,
where the running time is parameterized by the edit distance \(k\).
A celebrated algorithm by Landau and Vishkin ({\sc jcss}'\oldstylenums{88})
achieves a running time of \(\Oh(n + k^2)\),
which is optimal as a function of $n$ and $k$ (again, up to subpolynmial factors and conditioned on {\sc seth}).

While the theory community thoroughly studied the Levenshtein distance, most practical
applications rely on a more general \emph{weighted} edit distance, where each edit has a
weight depending on its type and the involved characters from the alphabet $\Sigma$.
This is formalized through a weight function
\(w : \Sigma \cup \{\emptystring\} \times \Sigma \cup \{\emptystring\} \to \mathbb{R}\) normalized so that \(\w{a}{a} = 0\) for \(a \in \Sigma\cup \{\emptystring\}\)
and \(\w{a}{b} \geq 1\) for \(a,b \in \Sigma\cup\{\varepsilon\}\) with \(a \neq b\); the goal is to find an alignment of the two
strings minimizing the total weight of edits.
The classic $\Oh(n^2)$-time algorithm supports this setting seamlessly, but for many
decades just a straightforward $\Oh(nk)$-time solution was known for the bounded version of
the weighted edit distance problem.
Only very recently, Das, Gilbert, Hajiaghayi, Kociumaka, and Saha ({\sc
stoc}'\oldstylenums{23}) gave the first
non-trivial algorithm, achieving a time complexity of \(\Oh(n + k^5)\).
While this running time is linear for $k\le n^{1/5}$, it is still very
far from \(\Oh(n + k^2)\)---the bound achievable in the unweighted setting.
This is unsatisfactory, especially given the lack
of any compelling evidence that the weighted version is inherently harder.

In this paper, we essentially close this gap by showing both an improved \(\Ohtilde(n +
\sqrt{nk^3})\)-time algorithm and, more surprisingly, a matching lower bound:
Conditioned on the All-Pairs Shortest Paths (\apsp) hypothesis, the running time of our
solution is optimal for $\sqrt{n}\le k\le n$ (up to subpolynomial factors).
In particular, this is the first separation between the complexity of the weighted and
unweighted edit distance problems.

Just like the Landau--Vishkin algorithm, our algorithm can be adapted to a wide variety of
settings, such as when the input is given in a compressed representation.
This is because, independently of the string length $n$, our procedure takes
$\Ohtilde(k^3)$ time assuming that the equality of any two substrings can be tested in
$\Ohtilde(1)$ time.

Consistently with the previous work, our algorithm relies on the observation that strings
with a rich structure of low-weight alignments must contain highly repetitive substrings.
Nevertheless, achieving the optimal running time requires multiple new insights.
We capture the right notion of repetitiveness using a tailor-made compressibility measure
that we call \emph{self-edit distance}.
Our divide-and-conquer algorithm reduces the computation of weighted edit distance to
several subproblems involving substrings of small self-edit distance and, at the same
time, distributes the budget for edit weights among these subproblems.
We then exploit the repetitive structure of the underlying substrings using
state-of-the-art results for multiple-source shortest paths in planar graphs (Klein, {\sc
soda}'\oldstylenums{05}).

As a stepping stone for our conditional lower bound, we study a dynamic problem of
maintaining two strings subject to updates (substitutions of characters) and weighted edit
distance queries.
We significantly extend the construction of Abboud and Dahlgaard ({\sc
focs}'\oldstylenums{16}), originally for dynamic shortest paths in planar graphs, to show
that a sequence of $n$ updates and $q\le n$ queries cannot be handled much faster than in
$\Oh(n^2 \sqrt{q})$ time.
We then compose the snapshots of the dynamic strings to derive hardness of the static
problem in the bounded setting.
\end{abstract}

\clearpage

\input{sections/introduction}
\input{sections/preliminaries}
\input{sections/upper_bounds}

\input{sections/lower_bounds}

\input{sections/static_lower_bound}

\footnotesize
\bibliographystyle{alphaurl}
\bibliography{main}

\end{document}

%% file: macros.tex
\fxsetup{mode=multiuser,theme=color, layout=inline}
\FXRegisterAuthor{tk}{atk}{\color{brown} {\bf Tomasz}}
\FXRegisterAuthor{ac}{aac}{\color{blue} {\bf Alejandro}}

\usetikzlibrary{arrows,arrows.meta,shapes.misc, decorations.pathreplacing,decorations.pathmorphing,calligraphy, patterns.meta}
\tikzset{dotmark/.style={circle,fill,inner sep=1.5pt}}
\tikzset{emptymark/.style={circle,draw,fill=white,inner sep=1.5pt}}
\tikzset{crossmark/.style={thick,inner sep=1.5pt}}

\SetKwBlock{Begin}{}{end}

\newcommand{\A}{\mathcal{A}}
\def\alm{\A}
\newcommand{\B}{\mathcal{B}}
\newcommand{\C}{\mathcal{C}}

\newcommand{\Oh}{\mathcal{O}}
\newcommand{\Ohtilde}{\widetilde{\Oh}}

\newcommand{\sm}{\setminus}
\newcommand{\sub}{\subseteq}

\def\twoheadleadsto{\tikz[baseline=(a.base)]{\draw[%
    decorate,decoration={zigzag,segment length=4, amplitude=.9},%
    ] (0,0) -- (.25, 0);%
    \draw[%
    -{Classical TikZ Rightarrow}.{Classical TikZ Rightarrow},%
    ] (.25, 0) -- (.4, 0);%
    \node (a) at (.4/2,-.03) {\phantom{\(\leadsto\)}};%
}}
\def\Twoheadleadsto{\tikz[baseline=(a.base)]{\draw[%
    decorate,decoration={zigzag,segment length=4, amplitude=.9},%
    ] (0,-.02) -- (.27, -.02);%
    \draw[%
    -{Classical TikZ Rightarrow}.{Classical TikZ Rightarrow},%
    ] (.35, -.0) -- (.4, -.0);%
    \draw[%
    decorate,decoration={zigzag,segment length=4, amplitude=.9},%
    ] (0,0.02) -- (.27, 0.02);%
    \node (a) at (.4/2,-.03) {\phantom{\(\leadsto\)}};%
}}

\newcommand{\onto}{\twoheadleadsto}%
\newcommand{\eqonto}{\Twoheadleadsto}%
\def\aonto#1{\onto}
\def\aontoeq#1{\eqonto}

\newcommand\Int{\mathbb Z}
\newcommand\Real{\mathbb R}

\newcommand{\floor}[1]{\lfloor #1 \rfloor}
\newcommand{\ceil}[1]{\lceil #1 \rceil}

\def\fragmentco#1#2{\bm{[}\,#1\,\bm{.\,.}\,#2\,\bm{)}}
\def\fragmentoc#1#2{\bm{(}\,#1\,\bm{.\,.}\,#2\,\bm{]}}
\def\fragmentoo#1#2{\bm{(}\,#1\,\bm{.\,.}\,#2\,\bm{)}}
\def\fragment#1#2{\bm{[}\,#1\,\bm{.\,.}\,#2\,\bm{]}}
\def\interval#1#2{\bm{[}\,#1\bm{,}#2\,\bm{]}}
\def\intervalco#1#2{\bm{[}\,#1\bm{,}#2\,\bm{)}}

\def\position#1{\bm{[}\,#1\,\bm{]}}

\def\problembox#1{%
    \vspace{2mm}%
    \noindent\fbox{%
    \begin{minipage}{.985\linewidth}%
        #1
    \end{minipage}%
    }%
    \vspace{2mm}%
}
\newcommand{\defproblem}[4]{%
    \problembox{%
        \textsc{#1}\\
        {\bf{Input:}} #2  \\
        {\bf{Update:}} #3\\
        {\bf{Query:}} #4
    }%
}

\newcommand{\defproblemalg}[3]{%
    \problembox{%
        \textsc{#1}\\
        {\bf{Input:}} #2  \\
        {\bf{Task:}} #3
    }%
}

\makeatletter
\renewenvironment{cases}{%
  \matrix@check\cases\env@cases
}{%
  \endarray\right.%
}
\def\env@cases{%
  \let\@ifnextchar\new@ifnextchar
  \left\lbrace
  \def\arraystretch{1.1}%
  \array{@{\;}c@{\quad}l@{}}%
}
\makeatother

\def\mid{\ensuremath :}

\def\emptyset{\varnothing}

\def\w#1#2{\ensuremath w(#1 \mapsto #2)}
\def\hw#1#2{\ensuremath \hat{w}(#1 \mapsto #2)}
\def\hwed{\ed^{\hat{w}}}

\def\wabctn{\ensuremath w_{A,B,C,\tau}}
\def\wabcttn{\ensuremath \hat{w}_{A,B,C,\tau}}
\def\wabn{\ensuremath w_{A,B}}
\def\wab#1#2{\ensuremath \wabn(#1 \mapsto #2)}

\def\wabct#1#2{\ensuremath \wabctn(#1 \mapsto #2)}
\def\wabctt#1#2{\ensuremath \wabcttn(#1 \mapsto #2)}
\def\sy{\ensuremath\Sigma_{Y}}
\def\sab{\ensuremath\Sigma_{A,B}}
\def\sabx{\ensuremath\Sigma_{A,B,X}}
\def\saibx#1{\ensuremath\Sigma_{A_{#1},B,X}}

\def\sabct{\ensuremath\Sigma_{A,B,C,\tau}}
\def\sabctx{\ensuremath\Sigma_{X,A,B,C,\tau}}
\def\sabcty{\ensuremath\Sigma_{Y,A,B,C,\tau}}
\def\aijl{\ensuremath\mathcal{A}_{i,j,\ell}}
\def\aial{\ensuremath\mathcal{A}^{\alpha}_{i,\ell}}
\def\aijla{\ensuremath\mathcal{A}^{\alpha,\bullet}_{i,\ell}}

\newcommand{\Xb}{X^{\bot}}
\newcommand{\hX}{\hat{X}}
\newcommand{\hY}{\hat{Y}}
\newcommand{\hk}{\hat{k}}

\newcommand\ed{\mathsf{ed}}
\newcommand\hd{\mathsf{hd}}
\def\weds#1{\ed^{w_{#1}}}
\newcommand\wed{\weds{\!}}

\newcommand\selfed{\mathsf{self}\text{-}\ed}
\newcommand\LZ{\mathrm{LZ}}

\def\emptystring{\ensuremath\varepsilon}
\def\wg{w_G}
\def\sca{\ensuremath\alpha}
\def\scb{\ensuremath\beta}
\def\scc{\ensuremath\gamma}
\def\smallconst{\ensuremath\delta}

\def\ptau{\ensuremath p_{\tau}}

\def\alg{\mathbb{A}}

\newcommand{\dist}{\mathsf{dist}}

\newcommand{\Als}{\mathbf{A}} %

\colorlet{highcol}{blue!80!green}
\colorlet{lowcol}{red!30!yellow!40!black}

\newcommand{\modelname}{\texttt{PILLAR}\xspace}
\def\modelname{{\tt PILLAR}\xspace}
\def\lceOp#1#2{{\tt LCP}(#1, #2)}
\def\lcbOp#1#2{{\tt LCP}^R(#1, #2)}

\def\accOp#1#2{#1\position{#2}}
\def\ipmOpName{{\tt IPM}\xspace}

\def\accOpName{{\tt Access}\xspace}
\def\extractOpName{{\tt Extract}\xspace}
\def\lenOpName{{\tt Length}\xspace}

\def\lceOpName{{\tt LCP}\xspace}

\newcommand{\AG}{\mathsf{AG}}

\newcommand{\apsp}{\textsc{apsp}\xspace}

\makeatletter
\newcommand{\pushright}[1]{\ifmeasuring@#1\else\omit\hfill$\displaystyle#1$\fi\ignorespaces}
\newcommand{\pushleft}[1]{\ifmeasuring@#1\else\omit$\displaystyle#1$\hfill\fi\ignorespaces}
\makeatother

\allowdisplaybreaks

%% file: sections/introduction.tex
\section{Introduction}

The edit distance (also known as the \emph{Levenshtein distance}~\cite{Levenshtein66}) of two strings
is the minimum number of insertions, deletions, and substitutions of characters needed to transform one string into the other.
Computational problems that involve edit distance have been extensively investigated for decades.
The textbook dynamic-programming algorithm~\cite{Vintsyuk1968,NeedlemanW70,Sellers74,WagnerF74} computes the edit distance between two strings of lengths at most $n$ in time $\Oh(n^2)$. This is
optimal up to subpolynomial factors under the Orthogonal Vectors Hypothesis~(\textsc{ovh})~\cite{AbboudBW15,BringmannK15,AbboudHWW16,bi18}, which, in turn, follows from the Strong Exponential Time Hypothesis~(\textsc{seth})~\cite{Impagliazzo2001,Impagliazzo2001a}.
An established way to circumvent this hardness is to study the \emph{bounded setting}, where we measure the running time both by the length $n$ of the input strings and by their edit distance $k$.
For this setting, the celebrated algorithm of Landau and Vishkin~\cite{LandauV88}, building upon earlier results of Ukkonen~\cite{Ukkonen85} and Myers~\cite{Myers86}, runs in time $\Oh(n+k^2)$, which is linear for $k \leq \sqrt{n}$.
This algorithm was instrumental in many recent advancements in approximation and sublinear-time algorithms for edit distance~\cite{GoldenbergKS19,KociumakaS20},
and its running time is optimal as a function of $n$ and $k$ (again up to subpolynomial
factors under {\sc ovh}); see for instance~\cite{GKKS23}.

While the theoretical computer science community has extensively studied the Levenshtein distance, most real-world applications require the more general \emph{weighted} edit distance, where each edit is associated with a weight depending on its type and the involved characters from the alphabet $\Sigma$~\cite{Waterman95,KleinbergT05,Skiena08,FontanFFPA16};
see also \cite[Section 11.5.2]{Gusfield97} for a discussion of weight functions used in bioinformatics.
Consequently, the weighted edit distance has been studied since the earliest papers concerned with
edit distance~\cite{Sellers74,WagnerF74,Sellers80,Waterman95}.
The weighted edit distance is formalized through a weight function
$w : \Sigma \cup\{\emptystring\} \times \Sigma \cup \{\emptystring\} \to \Real$, which can be normalized so that $\w{a}{a}
= 0$ for $a \in \Sigma\cup \{\emptystring\}$ and $\w{a}{b} \geq 1$ for $a,b \in \Sigma\cup\{\varepsilon\}$ with $a \neq b$.
The underlying problem asks to align the two input strings while minimizing the total weight of the edits performed.
The standard dynamic-programming $\Oh(n^2)$-time algorithm supports this setting seamlessly, but for many decades just a straightforward $\Oh(nk)$-time solution was known for the bounded variant of weighted edit distance.
In particular, the algorithm by Landau and Vishkin does not extend to the weighted case.
Perhaps surprisingly, only recently Das, Gilbert, Hajiaghayi, Kociumaka, and
Saha~\cite{DGHKS23} obtained the first $\Oh(n + \mathrm{poly}(k))$-time
algorithm for weighted edit distance.
More precisely, their deterministic algorithm runs in time $\Oh(n + k^5)$---leaving open a wide gap between the unweighted and
weighted settings. In particular, the most fundamental separation question remained unsettled:

\begin{center}
    \emph{Is there an $\Oh(n + k^2)$-time algorithm for the Bounded Weighted Edit Distance problem?}
\end{center}

\subsection{Our Results}

In this work, we develop a significantly faster algorithm for Bounded Weighted Edit Distance and, on the way to proving its conditional optimality,
we answer the aforementioned question negatively:
\begin{itemize}
    \item first, we obtain a significantly improved deterministic algorithm
        that computes the  weighted edit distance in
        time $\Oh(n + \sqrt{nk^3}\log^3n)$; and
    \item second, assuming the All-Pairs Shortest Paths ({\apsp}) Hypothesis~\cite{Williams18}, we obtain
        an almost matching conditional lower bound.
\end{itemize}
Taken together, we thus settle the complexity of computing the weighted edit distance in the
bounded setting.
Our result also proves that the weighted variant of the edit distance problem is strictly harder than the unweighted one.

Formally, our main results read as follows:
\begin{restatable*}{mtheorem}{mthmalg}\label{mthm:algorithm}
   Given strings \(X\), \(Y\), each of length at most $n$,
   and oracle access to a normalized weight function~$w$,
   we can compute the value $k=\wed(X, Y)$ in time $\Oh(n + \sqrt{nk^3}\log^3 n)$.
\end{restatable*}

\begin{restatable*}{mtheorem}{mthmlb}\label{mthm:lb}
    Let $\smallconst > 0$ and $0.5 \le \kappa \le 1$ denote real constants.
    Assuming the \apsp Hypothesis, there is no algorithm that,
    given strings $X,Y$ of length at most $n$, a threshold $k\in \Real_{>0}$
    satisfying $k\le n^\kappa$, and oracle access to a normalized weight function $w$,
    decides $\wed(X,Y)\le k$ in time $\Oh(n^{0.5+1.5\kappa - \smallconst})$.

    No such algorithm exists even if the alphabet $\Sigma$ is of size at most $n^{\kappa}$ and the weights are rationals in $\interval{0}{2}$ with a common $\Oh(\log n)$-bit denominator.
\end{restatable*}

\subparagraph*{The \modelname{} Model and Improved Algorithms for Other Settings.}

Our algorithm can be adapted to work in a variety of settings, such as when the input is
given in a compressed representation.
For this purpose, we use the \modelname{} model, as introduced by Charalampopoulos,
Kociumaka, and Wellnitz~\cite{Charalampopoulos20}.
The \modelname{} model provides an interface to a set of primitive operations (the
\modelname{} operations), such as computing the length of the longest common prefix of two
substrings, that can be efficiently implemented in different settings.
Thus, by bounding the running time of an algorithm in terms of the number of \modelname{}
operations, we can obtain algorithms in any setting which implements the \modelname{}
operations efficiently.

We obtain the following result.

\begin{restatable*}{theorem}{mthmpillar}\label{mthm:pillar}
    Given distinct strings $X, Y$ of length at most $n$ and oracle access to a normalized weight function~$w$,
    we can compute the value $k=\wed(X, Y)$ in time $\Oh(k^3\log^2 n)$ in the \modelname{} model.
\end{restatable*}

\subsection{Related Work}

Edit distance has been studied extensively since the 1960s~\cite{Levenshtein66,Vintsyuk1968,NeedlemanW70,WagnerF74,MP80},
with the history of the bounded version reaching mid-1980s~\cite{Ukkonen85,Myers86,LandauV88}.
The quadratic hardness result conditioned on \textsc{ovh}~\cite{AbboudBW15,BringmannK15,AbboudHWW16,bi18}
naturally motivates the study of subquadratic approximation algorithms~\cite{AndoniKO10,AndoniO12,ChakrabortyDGKS18,KouckyS20,BrakensiekR20,GoldenbergRS20,BEGHS21}.
The current best approximation algorithm in almost linear time is due to Andoni and Nosatzki~\cite{AndoniN20}, achieving a constant-factor approximation in time $\Oh(n^{1+\smallconst})$ for any constant $\smallconst > 0$.
Sublinear-time approximations have also flourished in the last few years~\cite{Batu2003,AndoniO12,GoldenbergKS19,KociumakaS20,GKKS22,BCFN22}.
None of these approximation algorithms extend to weighted edit distance, where the current state of the art, due to Kuszmaul~\cite{Kuszmaul19}, is
an $\Oh(n^{\delta})$-approximation in time $\Oh(n^{2-\delta})$ for any $0 \le \delta \le 1$.

The (unweighted) bounded edit distance problem has also been studied in multiple other settings: there are efficient sketching \& streaming algorithms~\cite{BZ16,JNW21,KociumakaPS21,BK23},  algorithms for compressed strings~\cite{GaneshKLS22}, and algorithms for preprocessed data~\cite{GoldenbergRS20,BCFN22b}.
Recent works have also focused on the bounded version of two related problems: Dyck edit distance~\cite{BackursO16,FriedGKKPS22,D22} and tree edit
distance~\cite{AkmalJ21,DasGHKSS22}. The $\Oh(n+\operatorname{poly}(k))$-time weighted-edit-distance algorithm of~\cite{DGHKS23} generalizes to these two settings (with mild additional assumptions on the weight function, such as the triangle inequality), but the dependency on $k$ degrades to $\Oh(k^{12})$ and $\Oh(k^{15})$ respectively.

Most of the earlier theoretical results for weighted edit distance are limited to \emph{uniform} weight functions that assign one weight to all substitutions and another weight to all indels (insertions and deletions). When both weights are small integers, such a weighted edit distance can be embedded to the \emph{deletion distance}~\cite{abs-0707-3619} with cost-1 indels and cost-2 substitutions, and essentially all positive results can be lifted from the unweighted setting. Moreover, for any fixed weight function satisfying this assumption, the quadratic lower bound still applies~\cite{BringmannK15}. More recently~\cite{GKKS23}, the version with cost-$1$ substitutions and cost-$a$ indels (for \emph{superconstant} $a$) received some attention, resulting in conditionally optimal exact algorithms (also for the bounded setting) and sublinear-time approximation algorithms.
Edit distance with uniform weights has also been studied in compressed~\cite{HermelinLLW13,abs-0707-3619}
and dynamic~\cite{Charalampopoulos20a} settings.

\section{Technical Overview}\label{sec:technical-overview}
In this section, we give a high-level overview of our results.
\subsection{Improved Algorithms for Bounded Weighted Edit Distance}

Let us start by briefly reviewing the existing
algorithms.
First, recall the textbook dynamic-programming algorithm for edit distance that runs in
quadratic time and, in particular, its natural extension to the weighted edit distance.
It is instructive to view said algorithm in terms of the
\emph{alignment graph} of the input strings $X$, $Y$, and the normalized weight function
$w$ (consult~\cref{def:alignment-graph} for a formal definition).
The alignment graph is a directed acyclic graph formed by grid on the vertices
$\fragment{0}{|X|} \times \fragment{0}{|Y|}$ that is augmented with diagonal edges.

For intuition, visualize the vertex $(0, 0)$ in the top left corner of the grid
and a vertex $(x,y)$ in the $x$-th column and $y$-th row of the grid.
Now, we interpret
\begin{itemize}
    \item a horizontal edge $(x, y) \to (x+1, y)$
        as a deletion of $X\position{x}$
        and, accordingly, assign to it a weight of
        $\w{X\position{x}}{\varepsilon}$;
    \item a vertical edge $(x, y) \to (x,y+1)$
        as an insertion of $Y\position{y}$
        and, accordingly, assign to it a weight of
        $\w{\varepsilon}{Y\position{y}}$;
    \item a diagonal edge $(x, y) \to (x+1, y+1)$
        as a match or a substitution of $X\position{x}$ for $Y\position{y}$
        and, accordingly, assign to it a weight of
        $\w{X\position{x}}{Y\position{y}}$.
\end{itemize}
A \emph{diagonal} in the alignment graph is a path that starts at some vertex \((x,0)\) or
\((0,y)\) and proceeds along diagonal edges as long as possible.
The diagonal starting at \((0,0)\) is the \emph{main diagonal}.

Observe that the weighted edit distance between $X$ and $Y$ corresponds to distance from
$(0,0)$ to $(|X|, |Y|)$ in the alignment graph. Since the graph is acyclic, we can compute
this distance in time proportional to the graph size, which is $\Oh(n^2)$.

If we have an upper bound $k$ on the distance, we can easily improve the running
time to $\Oh(nk)$: for a normalized weight function, the weight of every horizontal and vertical
edge in the alignment graph is at least 1.
Hence, all vertices $(x,y)$ that are at a distance of at most $k$ from $(0,0)$ satisfy
$|x-y| \le k$.
Thus, we can restrict ourselves to the subgraph induced by the $k$ diagonals above
and below the main diagonal.
This yields the classic \(\Oh(nk)\) algorithm (see~\cref{prop:baseline-wed}).

Next, we give a brief description of the $\Oh(n + k^5)$-time algorithm by Das, Gilbert,
Hajiaghayi, Kociumaka, and Saha~\cite{DGHKS23}, which we use as a base for our algorithm.
To that end, fix two strings $X$ and $Y$ of length at most $n$ and a bound~$k$ on their
weighted edit distance.
We say that a string \(P\) has \emph{synchronized occurrences} at \((x,y)\) if
\(X\fragmentco{x}{x + |P|} = P = Y\fragmentco{y}{y + |P|}\) and \(|x-y| \le k\).
We then call \(X\fragmentco{x}{x + |P|}\) and \( Y\fragmentco{y}{y + |P|}\) \emph{synchronized fragments}.
The algorithm of~\cite{DGHKS23} proceeds in three steps:

\begin{itemize}
    \item \textbf{Step 1: Identify \boldmath $\Oh(k)$ maximal strings with disjoint synchronized occurrences.}
    We first use the algorithm by Landau and Vishkin~\cite{LandauV88} to check if the
    unweighted edit distance between \(X\) and \(Y\) is at most \(k\).
    If we find $\ed(X, Y) > k$, then certainly also $\wed(X, Y) > k$ too, and we
    are done.

    If we find $\ed(X, Y) \leq k$, then the Landau--Vishkin algorithm also
    returns an alignment of \(X\) to \(Y\) of \emph{unweighted} cost at most $k$.
    We use this alignment to decompose $X$ into at most $k$ individual characters
    (that are deleted or substituted) and at most $k+1$ (maximal) fragments that are matched
    perfectly.
    Since the alignment makes at most $k$ edits, each pair of fragments that is matched perfectly
    forms a synchronized occurrence.

    \item \textbf{Step 2: Shorten each synchronized fragment to  \(\bm{\Oh(k^3)}\).}
    Suppose that in the previous step, we found a substring \(P = X_1 = Y_1\) and its synchronized occurrences.
    We shorten \(P\) (and the corresponding fragments of \(X\) and \(Y\))
    to a length of \(\Oh(k^3)\) in two steps.
    \begin{enumerate}[a]
        \item \emph{Shorten repetitive regions to \(4k\) repetitions.}
            Suppose \(P\) contains a string \(Q^{\alpha}\)
            for some string \(Q\) of length at most \(2k\) and an integer \(\alpha \ge 4k\).
            Then, any alignment of cost at most \(k\) has to match perfectly at least one full
            \(Q\) in \(X_1\) to one full \(Q\) in \(Y_1\).
            Hence, without altering the weighted edit distance, we can remove one full \(Q\)
            from $P$ as well as \(X_1\) and \(Y_1\).

            We repeat the above reduction until every run of a string \(Q\) with \(|Q| \le
            2k\) has an exponent of at most \(4k\) and write \(P'\) for the resulting substring.
        \item \emph{Remove irrelevant middle part.}
            Suppose that $|P'| > 42k^3$ (otherwise, we skip this step),
            consider the length-$21k^3$ prefix of $P'$, and write \(X_2\) and \(Y_2\) for
            the corresponding synchronized fragments of \(X\) and~\(Y\), respectively.
            Observe that every alignment of cost of at most \(k\) has to match perfectly
            some fragment of \(X_2\) of length $10k^2$ to a fragment of~\(Y_2\); call the underlying
            substring \(R = X_2\fragmentco{x}{x + |R|} = Y_2\fragmentco{y}{y + |R|}\).
            Moreover, we note that \(x = y\), that is, the canonical occurrence of \(R\) in $X$ is
            matched perfectly to the canonical occurrence of
            $R$ in $Y$ (both induced by the synchronized occurrences of $P'$).
            This is because a perfect match with another occurrence would imply that $R$ has a period of length at most $2k$,
            contradicting the assumption that $P'$ has been processed by the previous sub-step 2a.
            A symmetric argument shows that the length-$21k^3$ suffix of $P'$ also contains a
            fragment $R'$ whose occurrences in $X$ and $Y$ (corresponding to the synchronized
            occurrences of $P'$) are matched perfectly.
            We conclude that every optimal \emph{weighted} alignment between $X$ and $Y$ must
            match perfectly everything between the occurrences of $R$ and $R'$.
            Consequently, we can trim $P'$ to
            $P'\fragmentco{0}{21k^3}\cdot P'\fragmentco{|P'|-21k^3}{|P'|}$.
    \end{enumerate}
    \item \textbf{Step 3: Run the \boldmath $\Oh(nk)$-time algorithm.}
    From the previous step, we obtain strings \(X'\) and \(Y'\) that have lengths of at most
    \(\Oh(k^4)\): both strings consist in
    at most \(\Oh(k)\) (maximal) synchronized fragments and each of those synchronized
    fragments has a length of at most \(\Oh(k^3)\).
    Further, as our transformations preserve the weighted edit distance, we can use the
    \(\Oh(nk)\) algorithm on \(X'\) and \(Y'\) to obtain the weighted edit distance between
    \(X\) and \(Y\)---this step takes time \(\Oh(k^5)\) and dominates the running time in
    terms of~\(k\), yielding a total running time of \(\Oh(n + k^5)\) (step 2 can be implemented in linear time).
\end{itemize}

Next, we describe the three new ideas that let us gradually improve upon the running time of~\cite{DGHKS23}.

\subparagraph*{Improvement 1: Processing Repetitive Regions Efficiently.}
Our first observation is that  Step~3 of~\cite{DGHKS23} does not utilize
any structural properties of the strings $X'$ and $Y'$.
Indeed, a closer look at Step~2b reveals that we can shrink a fragment \(P\) with synchronized
occurrences in \(X'\) and \(Y'\) even further
until it has the following structure: \(P\) consists in \(\Oh(k)\)
\emph{periodic pieces} of period \(\Oh(k)\) each.

Observe that Step~3 spends $\Oh(k^3)$ time on every such periodic piece:
Indeed, any periodic piece involves $\Oh(k^2)$ rows (by the reduction in Step~2a),
which induce $\Oh(k^3)$ vertices in the restricted alignment graph $G$
(that is, \(G\) restricted to vertices $(x,y)$ with $|x-y|\le k$).

We reduce the time
per periodic piece to $\Ohtilde(k^2)$, thereby shaving a factor of
$\widetilde{\Omega}(k)$.%
\footnote{We write $\Ohtilde(\cdot)$ and \(\widetilde\Omega(\cdot)\)
and \(\widetilde\Theta(\cdot)\) to hide polylogarithmic factors $(\log n)^{\Oh(1)}$.}
Fix a periodic piece $P$ with period $q =\Oh(k)$
and suppose that it has synchronized occurrences in $X$ and $Y$ at $(x^*, y^*)$.
Consider the subgraph $H$ of the (restricted) alignment graph that is
induced by the vertices $(x,y)$ with $x\in \fragment{x^*}{x^*+|P|}$.
For $i\in \fragment{0}{|P|}$, let $V_i$ denote the sequence of vertices in the $i$-th column
of $H$ (listed from top to bottom).
We may assume that we already know the distances from $(0,0)$ to the vertices in $V_{0}$.
Now, the task is to compute the distances from $(0,0)$ to the vertices in $V_{|P|}$.

Our crucial insight is as follows: the periodicity of $P$ implies that many bands of $H$ are isomorphic.
More precisely, if $i\in \fragmentco{2k}{|P|-2q-2k}$, then the graph between $V_{i}$ and
$V_{i+q}$ is isomorphic to the graph between $V_{i+q}$ and $V_{i+2q}$. This is because $q$
is a period of $P=X\fragmentco{x^*}{x^*+|P|}=Y\fragmentco{y^*}{y^*+|P|}$,
and thus also of both $X\fragmentco{x^*+i}{x^*+i+2q}$ and $Y\fragmentco{x^*+i-k}{x^*+i+2q+k}$.

Let $D$ denote a matrix consisting of the pairwise distances from $V_{2k}$ to $V_{2k+q}$.
By the above observation, for every integer $e$ satisfying $qe \in \fragment{0}{|P|-4k}$, the $e$-th power $D^e$
with respect to the min-plus product\footnote{The min-plus product of two $n\times n$
    matrices $A$ and $B$ is the $n\times n$ matrix $C$ defined as $C_{i,j} = \min_k
A_{i,k} + B_{k,j}$.} represents the pairwise distances from $V_{2k}$ to $V_{2k+qe}$.

At this point, we exploit that \(H\) is a directed planar graph and
apply known tools from the planar graph literature.
We can compute the matrix $D$ in $\Ohtilde(k^2)$ time using the multiple-source shortest
paths algorithm by Klein~\cite{Klein2005} (see~\cref{thm:klein}).
Moreover, we use the (well-known) fact that the matrix $D$ satisfies the Monge property
(see~\cref{fct:monge}) to conclude
that we can compute any power of $D$ in time $\Ohtilde(k^2)$ using the SMAWK algorithm~\cite{SMAWK87} (see~\cref{thm:smawk})
and the exponentiation by squaring method.

Specifically, we compute $D^e$ for $e=\lceil(|P|-4k)/q\rceil$.
In order to exploit the pairwise distances from $V_{2k}$ to $V_{2k+qe}$, we use the naive algorithm to compute
the vector of distances from $(0,0)$ to $V_{2k}$, then derive the distances from $(0,0)$ to $V_{2k+qe}$ as the min-plus product of that vector with $D^e$,
and finally use the naive algorithm again to determine the distances from $(0,0)$ to $V_{|P|}$.

In total, the above allows to process a periodic piece in time $\Ohtilde(k^2)$.
As each of the \(\Oh(k)\) synchronized occurrences has $\Oh(k)$ periodic pieces,
we obtain an $\Ohtilde(k^4)$-time algorithm.

\subparagraph*{Improvement 2: Divide and Conquer.}

Next, we effectively reduce the total number of periodic pieces to $\Oh(k)$, deriving a running time of
$\Ohtilde(k^3)$.
Intuitively, this is possible due to \emph{locality} of edit distance: we can optimally
align any position $x^*$ in $X$
based on the neighborhood of $x^*$ spanning $\Theta(k)$ periodic pieces in each direction.

Although aligning a single position $x^*$ in itself does not give us the weighted edit distance,
we can use this aligned position to partition the input into two independent subproblems.
The remaining challenge is to partition the budget $k$ between these two subproblems.
A natural attempt would be to use exponential search so that $d=\wed(X,Y)$ is always
computed in $\Ohtilde(d^3)$ time.
Unfortunately, a constant-factor overhead at each level of recursion would accumulate to a
polynomial overhead globally.
Thus, our goal is more modest: we wish to compute $d=\wed_{\le k}(X,Y)$ in $\Ohtilde(k^2 d)$ time.

Our divide-and-conquer approach hinges on a procedure that solves the following subproblem:
given a position $x^*\in \fragment{0}{|X|}$, determine
a position $y^*\in \fragment{0}{|Y|}$ such that
\[
    \wed(X,Y)=\wed(X\fragmentco{0}{x^*},Y\fragmentco{0}{y^*})+\wed(X\fragmentco{x^*}{|X|},Y\fragmentco{y^*}{|Y|}).
\]
We apply this procedure for $x^* \approx \frac12|X|$ in order to achieve
a logarithmic recursion depth.
In the algorithm, we focus on the maximal fragment $X^* \coloneqq X\fragmentco{\ell}{r}$
such that  both $X\fragmentco{\ell}{x^*}$
and $X\fragmentco{x^*}{r}$ can be partitioned into $\Theta(k)$ substrings of period $\Theta(k)$.

As alluded earlier, assuming that $\wed(X,Y)\le k$, we can find $y^*$ just by looking at $X^*$ and $Y^* := Y\fragmentco{\ell-k}{r+k}$.
Specifically, we consider the restricted alignment graph $G$
and compute the shortest path $P$ from the $\ell$-th column to the $r$-th column.

Our algorithm selects any $y^*$ such that  $(x^*,y^*)\in P$.
Its correctness relies on the fact that the globally optimal alignment (a path from
$(0,0)$ to $(|X|,|Y|)$ of cost at most~$k$)
must meet $P$ at least twice: once before $(x^*,y^*)$ and once after $(x^*,y^*)$;
otherwise, the existence of two disjoint paths would
imply that $X\fragmentco{\ell}{x^*}$ or $X\fragmentco{x^*}{r}$ can be partitioned into few
periodic substrings, contradicting the maximality in their definitions.
Since $P$ is optimal between the meeting points, we can reroute the global alignment to
follow $P$ there and, in particular, to contain $(x^*,y^*)$.

Thus, it remains to compute $P$.
Since $\wed(X,Y)\le \ed(X,Y)\le k$, not only can we partition $X^*$ into $\Oh(k)$ periodic pieces,
but if we disregard $\Oh(k)$ positions,
we can also partition $X^*$ into $\Oh(k)$ periodic pieces that have synchronized
occurrences in $Y$ (included in $Y^*$).
Using the techniques of Improvement~1, each of these pieces can be processed in
$\Ohtilde(k^2)$ time for a total of $\Ohtilde(k^3)$.

The remaining challenge is to improve this running time to $\Ohtilde(k^2 d)$. For this, we
do use exponential search, but, unlike in the infeasible approach sketched above, the
recursive calls are outside this exponential search, so the overheads do not accumulate
across the levels of recursion.
Specifically, we observe that our correctness argument only requires that the cost of $P$
is at most $k$ and that $P$ is the shortest path between its endpoints.
If $\wed(X,Y)\le d$, we are guaranteed to find a path of cost at most $d$ starting at one
of the vertices $(\ell,\ell-d),\ldots,(\ell,\ell+d)$
and ending at one of the vertices $(r,r-d),\ldots,(r,r+d)$. All the internal vertices
$(x,y)$ of such a path satisfy $|x-y|\le 2d$.
Hence, while computing $P$, we can restrict the graph $G$ further so that it contains
$\Oh(d)$ diagonals.
This lets us improve the processing time to $\Ohtilde(kd)$ per periodic piece: the naive
algorithm propagates distances over a single column in $\Oh(d)$ time, the subgraph for
which we run Klein's algorithm is of size $\Oh(kd)$, and resulting matrix is of size
$\Oh(d)\times \Oh(d)$.

\subparagraph*{Improvement 3: Self-Edit Distance and Decomposition into Boxes.}

In Improvement 2, we chose $X^*$ as a maximal fragment consisting of $\Oh(k)$ pieces with
period $\Oh(k)$,
and then we exploited the structural properties of $X^*$ to efficiently align it with the
corresponding fragment~$Y^*$.
While this was very helpful if $n = \omega(k^2)$, any string of length $n = \Oh(k^2)$ can
be decomposed into $\Oh(k)$ pieces of with periods $\Oh(k)$.
In that case, our sequence of improvements brings us back to square one---the task of
beating the classic $\Oh(nk)$-time algorithm on arbitrary strings. Our final idea is to
exploit stronger notions of repetitiveness.
As already observed in~\cite{KociumakaPS21}, if two strings have disjoint alignments of
cost at most $k$, then their \emph{Lempel--Ziv}~\cite{lz77} factorization consists of
$\Oh(k)$ phrases. Few strings of length $\Oh(k^2)$ satisfy this property, so, with $X^*$
redefined accordingly, we have fresh hope for further improvements.

In particular, since the edit distance problem has already been considered in the
compressed setting~\cite{HermelinLLW13,Gawrychowski2012, Tiskin2015,GaneshKLS22}, we can
leverage some existing ideas, such as the notion of \emph{box decompositions} of the
alignment graph, which proved useful for bounded edit distance~\cite{GaneshKLS22}.
Given a decomposition of $X^*$ and $Y^*$ into disjoint \emph{phrases} $X^* = X^*_1 \cdots
X^*_{p_X}$ and $Y^* = Y^*_1 \cdots Y^*_{p_Y}$, the box decomposition of the alignment
graph is its partition into $p_X \cdot p_Y$ boxes induced by the Cartesian products of the
phrases.
Thus, if we decompose $X^*$ (and $Y^*$) into $\Oh(n/\ell)$ phrases of length
$\Theta(\ell)$, this yields $\Oh(nd/\ell^2)$ boxes within the scope of the $\Oh(d)$
diagonals of our interest. The bound $\Oh(k)$ on the Lempel--Ziv factorization size
implies that we can construct partitions with $\Ohtilde(k)$ distinct phrases,
and thus a box decomposition with $\Ohtilde(k^2)$ distinct boxes~\cite{GaneshKLS22}.
With the box decomposition in hand, we can compute the shortest path $P$ from the leftmost
to the rightmost column using the techniques from Improvement 1.
Namely, if we compute the boundary-to-boundary distance matrix for each box, then we can
use one min-plus matrix-vector product per box: given a vector of distances to the
top-left boundary, we obtain distances to the bottom-right boundary.
Since the distance matrix satisfies the Monge property, the product can be computed in
$\Oh(\ell)$ time using the SMAWK algorithm even though the matrix takes $\Theta(\ell^2)$
space. On the other hand, the distance matrix construction using Klein's algorithm takes
$\Ohtilde(\ell^2)$ time, but it can be shared by all copies of every distinct box.
Overall, this yields  $\Ohtilde(nd/\ell + k^2 \ell^2)$ time; for our testbed of
$n=\Oh(k^2)$ and $d=\Oh(k)$, this is $\Ohtilde(k^{8/3})$---still slower than the promised
$\Ohtilde(k^{5/2})$.

Intuitively, the issue is that we have too many distinct boxes because the $\LZ$
factorization does not limit the distance between each phrase and its source: in the
resulting decomposition of $X^*$, subsequent copies of the same phrase might be
arbitrarily far apart.
To improve this, in~\cref{sec:alg:sec:self-alignments} we define a tailor-made
compressibility measure that we call \emph{self-edit distance}.
This measure, denoted by $\selfed(\cdot)$, is defined as the minimum cost of an unweighted
alignment between the string and itself, with the restriction that the alignment does not
contain any edge of the main diagonal.
In particular, we choose $X^*$ to be the maximal fragment $X^* \coloneqq
X\fragmentco{\ell}{r}$ such both $X\fragmentco{\ell}{x^*}$ and $X\fragmentco{x^*}{r}$
have $\selfed(\cdot) = \Oh(k)$.

As disjoint alignments imply small self-edit
distance (see~\cref{lem:compose-disj-alignments}),
the previous correctness argument remains valid.
Then, in~\cref{lem:decomp} we show that the updated definition allows us to partition $X^*$ (and $Y^*$) into
$\Oh(n/\ell)$ phrases of length $\Theta(\ell)$ where there are $\Oh(k)$ \emph{fresh
phrases}, which are distinct from $\Theta(k/\ell)$ previous phrases.
This is because matching fragments in the self-alignments yield phrases whose sources are
just $\Oh(k)$ positions behind.
Constructing the box decomposition based on this, we can obtain $\Oh(k^2/\ell)$ distinct
boxes, an improvement over the $\Oh(k^2)$ using $\LZ$ factorization.
Thus, following the approach outlined earlier, we obtain an algorithm with overall running time
$\Ohtilde(nd/\ell + k^2 \ell)$. Choosing $\ell = \widetilde\Theta((\sqrt{nd})/k)$ yields
our result.

\subsection{Conditional Lower Bounds Assuming the APSP Hypothesis}

\subparagraph*{Fine-Grained Lower Bounds and the All-Pairs Shortest Paths Hypothesis.}
Whereas the lower bounds for unweighted edit distance~\cite{bi18} follow from the Orthogonal Vectors Hypothesis~\cite{Williams2005}
(implied by the Strong Exponential Time Hypothesis~\cite{Impagliazzo2001,Impagliazzo2001a}),
the starting point of our negative results is the widely-believed
hypothesis on the All-Pairs Shortest Paths
(\apsp) problem:

\begin{restatable*}[{\sc apsp} Hypothesis \cite{Williams18}]{hypothesis}{apspconj}\label{apsp-conj}
    For every constant \(\smallconst > 0\), no \(\Oh(n^{3-\smallconst})\)-time algorithm solves the \apsp problem on graphs with \(n\) vertices and edge weights in \(\fragment{-n^{\Oh(1)}}{n^{\Oh(1)}}\).
\end{restatable*}

In this paper, we give a reduction from \apsp to (variants of) the Bounded Weighted Edit
Distance problem. Assuming~\cref{apsp-conj}, this reduction justifies calling our improved
algorithms optimal.

For our reduction, we in fact start with an unbalanced version of the \textsc{Negative Triangle} problem.
Its \apsp-hardness is a folklore corollary of the results of~\cite{Williams18}; see \cref{sec:finegrained} for more details.

\begin{restatable*}[{folklore; see~\cite{Williams18}}]{corollary}{cornegtriangle}\label{cor:negtriangle}
    Let \(0 < \sca, \scb, \scc \le 1\) and \(\smallconst > 0\) denote real constants.
    Unless the \apsp Hypothesis fails, there is no \(\Oh(N^{(\sca + \scb + \scc) - \smallconst})\)-time
    algorithm that, given a complete tripartite graph \(G = (P \cup Q \cup R, E,\wg)\) with \(|P|
    \le N^{\sca}\), \(|Q| \le N^{\scb}\), \(|R| \le N^{\scc}\), and weights in \(\fragment{-N^{\Oh(1)}}{N^{\Oh(1)}}\), decides
    if there exist vertices $a\in P$, $b\in Q$, and $c\in R$ such that $\wg(a,b)+\wg(b,c)+\wg(c,a)<0$.
\end{restatable*}

\subparagraph*{The Construction of~\cite{Abboud2016} for Dynamic Planar All-Pairs Shortest
Path.}

For our lower bounds, we proceed similarly to the lower bound for dynamic planar graph
algorithms by Abboud and Dahlgaard~\cite{Abboud2016}; similar ideas were also used to prove lower bounds for distance labeling in planar graphs~\cite{GavoillePPR04} and for dynamic variants of the Longest Increasing Subsequence
problem~\cite{Gawrychowski2021}. While the existing lower bounds and our new
lower bound all reduce from the All-Pairs Shortest Paths problem---and use a similar
general framework in doing so---each reduction requires a substantial amount of new ideas.
We start by reviewing the existing lower bounds of~\cite{Abboud2016} and then highlight
where and how our construction differs.

Among the results of~\cite{Abboud2016}, let us focus on that for Dynamic
Planar All-Pairs Shortest Paths.

\defproblem{Dynamic Planar All-Pairs Shortest Paths}%
{Planar graph \(G = (V, E, \wg)\)}%
{Add or remove a single edge or change the weight of a single edge}%
{Compute the shortest path distance between any pair of vertices of \(V\)}

\begin{figure}[t]
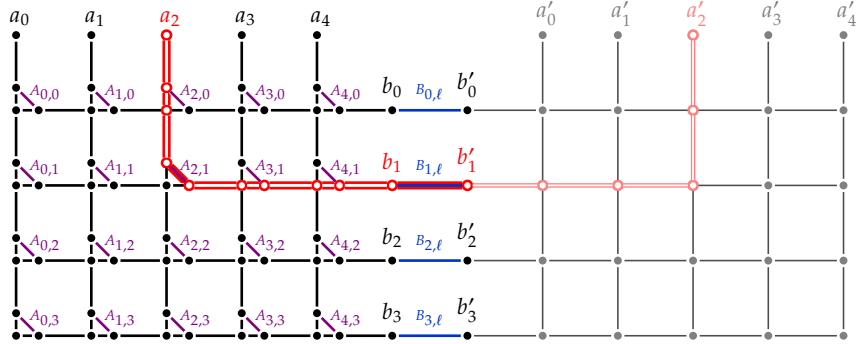

    \centering
    \include{figures/fig-0}
    \caption{The key construction of~\cite{Abboud2016} for an \(p\times q\) matrix \(A\)
        and the \(\ell\)-th column of a \(q\times r\) matrix \(B\).
        Suitably chosen auxiliary costs
        ensure that the shortest path from \(a_i\) to \(a'_i\) via \(b_j\) has a cost of
        \((A_{i,j} + B_{j, \ell}) + C\) for a constant \(C\) that is independent of
        \(i\) and---crucially---independent of \(j\).
        The part to the right of the \(b'_{\star}\)'s ensures that the total weight is
        independent of \(j\) and
        carries only auxiliary weights.\\
        Using edge weight updates, we change the costs of the edges \((b_{j}, b'_{j})\) to
        those of a different column of \(B\).
        Using queries, we query the shortest path between \(a_i\) and \(a'_i\).
        In total, we thus need \(\Oh(qr)\) updates and \(\Oh(pr)\) queries to compute the
        min-plus product of \(A\) and \(B\), which yields the lower bound
        of~\cite{Abboud2016}.}\label{fig:4-2-1}
\end{figure}

Central to the reduction of~\cite{Abboud2016} is a clever grid
construction that is used to compute the min-plus matrix-vector product of a \(p \times q\)
matrix \(A\) and a column \(B_{\ell}\) of another \(q\times r\) matrix \(B\);
see \cref{fig:4-2-1}.

Intuitively, the authors of~\cite{Abboud2016} create a \(p\) by \(q\) grid
and add a ``shortcut'' with weight \(A_{i,j}\) to every grid vertex
\(v_{i,j}\). Auxiliary costs ensure that the shortest path from any vertex \(a_i\) of the topmost
row to any vertex \(b_j\) in the rightmost column takes exactly one shortcut, namely the
shortcut at vertex \(v_{i,j}\). Then, they add an extra edge of weight
\(B_{j,\ell}\) from each of the vertices \(b_j\) to a fresh vertex \(b'_j\).

This construction ensures that the shortest path from \(a_i\) to \(b'_j\) has a
weight of \(A_{i,j} + B_{j, \ell}\) plus auxiliary costs.
To stabilize the auxiliary costs, the authors add another (mirrored) copy of the original
grid but without any shortcuts (introducing vertices \(a'_i\) in the process).

Altogether, their construction ensures that the distance from \(a_i\) to \(a'_i\)
is \(\min_{j}(A_{i,j} + B_{j,\ell})\). Hence, using \(p\) queries,
\cite{Abboud2016} can obtain the min-plus matrix-vector product of \(A\) and \(B_{\ell}\).
Now, using \(q\) of the edge weight update operations (alternatively, by deleting and inserting $q$
edges), the authors replace the column vector \(B_{\ell}\) by the next column vector
\(B_{\ell + 1}\).
In total, the reduction uses \(pr\) queries and \(qr\) updates to compute the min-plus
product of \(A\) and \(B\);
choosing \(p\), \(q\), and \(r\) appropriately then yields the desired lower bound (since Min-Plus Product is equivalent to \apsp~\cite{Williams18}).

\subparagraph*{Batched Weighted Edit Distance Problem and its Hardness.}

As in~\cite{Abboud2016}, we wish to use a dynamic variant of the (bounded) weighted
edit distance to compute the min-plus product of two matrices.
Formally, we introduce the following problem, which we call Batched Weighted Edit
Distance.

\defproblemalg{Batched Weighted Edit Distance}{
    An alphabet $\Sigma$,
    a weight function $w$ on $\Sigma\cup\{\emptystring\}$,\\
    a~sequence~of~\(m\)~strings $X_0,\ldots,X_{m-1}\in \Sigma^x$ (a~\emph{batch}),
    a~string~$Y\in \Sigma^y$,
    and a threshold $k \in \Real_{\ge 0}$.
}{
    Decide if $\min\{ \wed(X_i,Y) \mid i \in \fragmentco{0}{m} \}\le k$.
}

Given matrices \(A \in \fragment{-E}{E}^{p\times q}\) and \(B \in \fragment{-E}{E}^{q
\times r}\) (for \(n \coloneqq \max\{p,q, r\}\) and \(E = n^{\Oh(1)}\)),
we wish to construct a string \(Y\), a string \(X_{\ell}\) for each column
\(B_{\ell}\) of \(B\), as well as a weight function $w$ such that
\begin{enumerate}
    \item the strings are of length \(\Oh(p + q)\);
    \item the alphabet is of size \(\Oh(p + q)\) (and hence the weight function has \(\Oh((p+q)^2)\) values);
    \item the distance \(\wed(X_{\ell},Y)\) corresponds to (an entry)
        of the matrix-vector min-plus product of \(A\) and \(B_{\ell}\);
    \item few substitutions in $X_\ell$ are sufficient to swap one column of \(B\) for a different
        column of \(B\).
\end{enumerate}
Observe that, in order to encode the matrix \(A\), it makes sense to rely on substitution costs:
for two strings of length \(n\), our weight function may have at most \(n\) different
values for deletion or insertions of characters, but up to \(n^2\) different values for
substitution costs. In fact, we may even assume that deleting characters in either string
costs the same for every character---for simplicity, let us assume that deletions in
\(X_{\ell}\) (which we construct to be shorter than \(Y\)) cost \(\infty\), meaning that
they are not possible, and that deletions in \(Y\) cost \(0\), meaning that they are free.
Further, we may assume that we never match a character
of any \(X_\ell\) to a character of \(Y\) by choosing \(X_\ell\) and \(Y\) to contain disjoint
characters.

\begin{figure}[p]
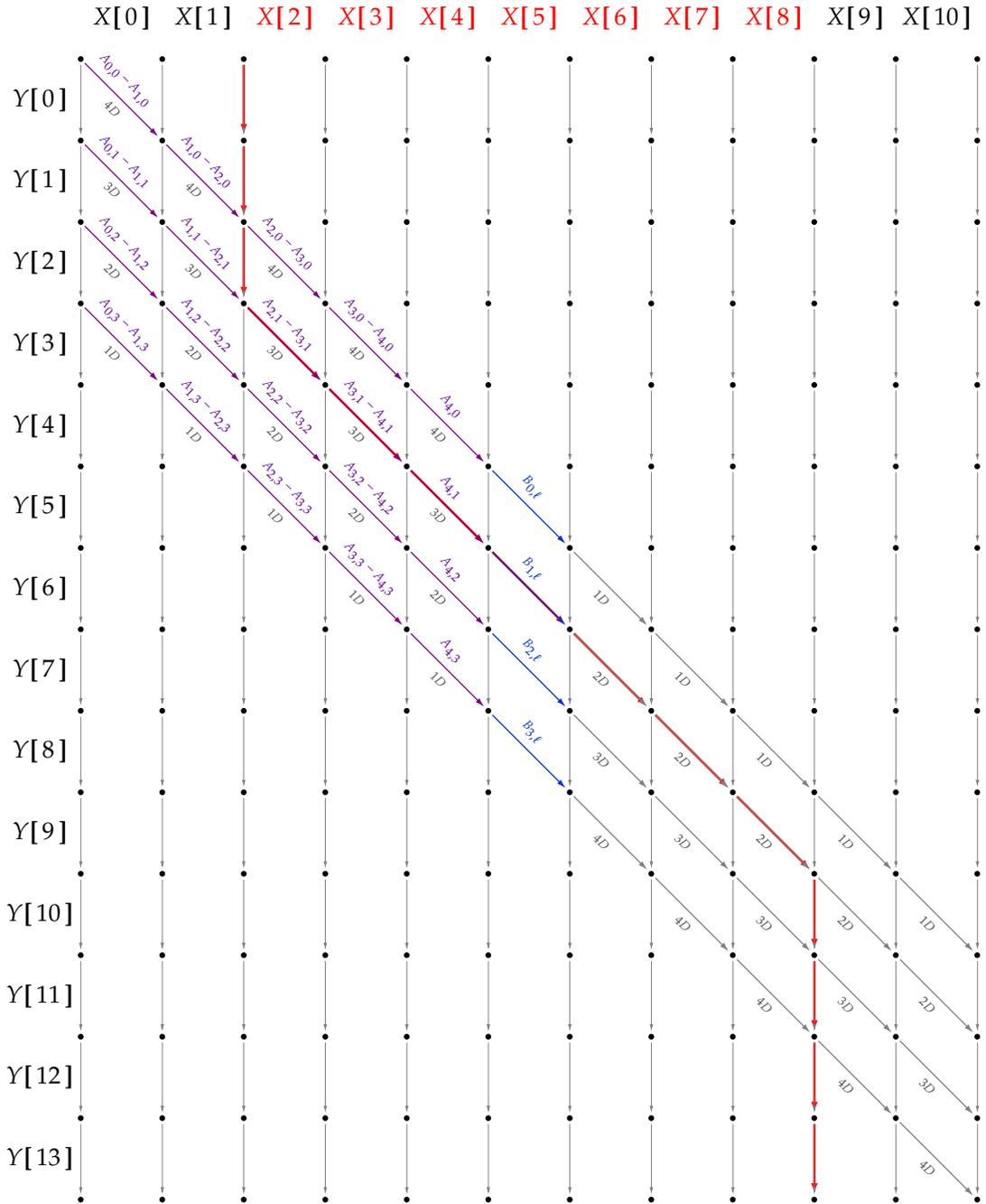

    \centering
    \include{figures/fig-1}
    \caption{The basic min-plus matrix-vector multiplication gadget for a \(p\times q\)
        matrix \(A\) and a column vector \(B_{\ell}\) of a \(q\times r\) matrix \(B\).
        Deletions in \(Y\) are assumed to be free; we omit transitions of infinite costs.
        The strings \(X\) and \(Y\) consist of disjoint, pairwise distinct characters and
        have lengths \(|X| = 2p + 1\) and \(2p + q\).
        We have \(\wed(X\fragmentco{i}{|X|-i}, Y) = \min_j(A_{i,j} + B_{j,\ell}) +
        (p-i)(q + 1)D\),
        where \(D\) is a constant much larger than any potential cost incurred by
        \(A_{\star}\) or \(B_{\star}\) costs.
        By replacing a single character in \(X\) (namely \(X\position{p}\)), we can switch
        out the vector \(B_{\ell}\) for a different vector.
    }\label{fig:4-2-2}
\end{figure}

Now, in a first attempt to implement the idea from~\cite{Abboud2016}, we proceed as
follows. We start with a string \(X\) of length \(2p + 1\) and a string \(Y\) of length
\(2p + q\). With the intuition that the alignment \(X\fragmentco{i}{|X|-i} \onto Y\)
that ``takes the \(j\)-th diagonal'' should have a cost of \(A_{i,j} + B_{j,\ell}\), we
assign weights for substitutions as follows (also consult~\cref{fig:4-2-2}), where we set \(A_{p,j} \coloneqq 0\) to simplify our exposition:
\begin{alignat*}{3}
    &\w{X\position{i}}{Y\position{i + j}} &&\coloneqq A_{i, j} - A_{i + 1,j},\quad
    && \forall i \in \fragmentco{0}{p}, \forall j \in \fragmentco{0}{q}\\
    &\w{X\position{p}}{Y\position{p + j}} &&\coloneqq B_{j, \ell},
    && \forall j \in \fragmentco{0}{q}\\
    &\w{X\position{p + i}}{Y\position{p + i + j}} &&\coloneqq 0,
    && \forall i \in \fragment{1}{p}, \forall j \in \fragmentco{0}{q}\\[1.5ex]
    &\w{a}{b} &&\coloneqq \infty,
    && \text{otherwise}.
\end{alignat*}

Observe that we have to use a telescoping sum for the weights on the diagonal because we cannot
insert ``shortcuts'' as in~\cite{Abboud2016}. Using telescoping sums complicates things:
the above weights do not guarantee yet that we would follow just one diagonal.
Hence, we add auxiliary costs to the diagonals as follows (again consult \cref{fig:4-2-2})
\begin{alignat*}{3}
    &\w{X\position{i}}{Y\position{i + j}}  &&\coloneqq A_{i, j} - A_{i + 1,j} + (q - j) D,\quad
    && \forall i \in \fragmentco{0}{p}, \forall j \in \fragmentco{0}{q}\\
    &\w{X\position{p}}{Y\position{p + j}} &&\coloneqq B_{j, \ell},
    && \forall j \in \fragmentco{0}{q}\\
    &\w{X\position{p + i}}{Y\position{p + i + j}} &&\coloneqq (j + 1) D,
    && \forall i \in \fragment{1}{p}, \forall j \in \fragmentco{0}{q}\\[1.5ex]
    &\w{a}{b} &&\coloneqq \infty,
    && \text{otherwise}.
\end{alignat*}
Observe that by following exactly one diagonal, we pay a cost of \((p-i)(q + 1)D\); using
more than one diagonal incurs a cost of at least \(((p-i)(q + 1) + 1)D\).
Hence, we indeed obtain that
\[\wed(X\fragmentco{i}{|X|-i}, Y) = \min_{j}(A_{i,j} + B_{j,
    \ell}) + (p-i)(q + 1)D.
\]
Now, while the previous construction is indeed useful, we are not done yet for several
reasons.
\begin{itemize}
    \item Ultimately, we need to create instances that ask to compute the weighted edit
        distance between two full strings, and not fragments of them. At the cost of
        increasing the number of strings \(X\) by a factor of \(\Oh(p)\),
        this problem can be
        fixed by adding a ``selector'' character to the beginning and end of \(X\) that
        ``disables'' suitable prefixes and suffixes of (the original) \(X\).
    \item Currently, our construction results in \(\Oh(pr)\) strings \(X\) (we replace
        the character \(X\position{p}\) once for each of the \(r\) columns of \(B\); for
        each possible \(X\position{p}\), we have to enable each possible row of \(A\)
        once). Hence, our instances have a total size of \(\Oh(pqr)\)---however,
        multiplying \(A\) and \(B\) naively also takes time \(\Oh(pqr)\).
\end{itemize}
To circumvent the second problem, we need to do more work. First, we observe that we can
extend the above construction to yield the min-plus product of \emph{three} matrices,
or---to be precise---the min-plus product of \(A\), a column vector of \(B\), and a row
vector of a third matrix \(C \in \fragment{-E}{E}^{r\times p}\).

While computing the full result would be too costly, we can get away with just computing the minimum entry of the
resulting matrix: said value is exactly the minimum weight of a triangle in the
complete tripartite graph \(G = (P\cup Q\cup R, E, \wg)\), where \(A\), \(B\), and \(C\)
describe the weights of the edges between \(P\) and \(Q\), between \(Q\) and \(R\), and
between \(R\) and \(P\), respectively. Seen through the lens of Negative Triangle, our
(improved) gadget computes the minimum weight of a triangle that uses
fixed vertices in \(P\) and \(R\).

Now, the threshold and the task of the Batched Weighted Edit Distance problem make
sense: as we are interested only in checking whether the minimum weight triangle has a
negative weight (recall that already Negative Triangle is \apsp-hard),
we may restrict ourselves to checking just that.
This translates to checking if the weighted edit distance is less than a certain threshold that comes out of our construction.

To finally obtain a useful lower bound, we further observe that we can split the matrix
\(A\) into several parts.
In particular, we split \(A\) into \(\tau\) parts of roughly the same size. Next, for each
part, we create a separate copy of the above gadget (for the same column of \(B\) and the
same row of \(C\)). We combine the \(\tau\) gadgets created in this way and can now use a
common ``selector'' for all \(\tau\) gadgets at once.

Splitting \(A\) helps: instead of \(\Oh(rp)\) strings, we now have to create only
\(\Oh(rp/\tau)\) strings.
In total, choosing \(p, q = \Theta(n)\) and \(r = \Theta(n^{\scb/2})\) and
setting \(\tau \coloneqq n^{1-\scb/2}\), we prove the following lower bound for the Batched Weighted Edit Distance
problem (see \cref{sec:batched-wed-lowerbound}).

\begin{restatable*}{theorem}{thmbatched}\label{thm:batched}\label{thm:dyn-lb-advanced}\label{thm:dyn-main}
    Let \(\smallconst > 0\) and
    \(\scb \in \interval{0}{1}\)
    denote real numbers. Assuming
    the \apsp~Hypothesis, there is no algorithm that, for every $n\in \Int_{\ge 1}$,
    solves every instance of the Batched Weighted Edit Distance problem
    with \(m \le n^{\scb}\) strings of lengths $x\le y \le n$
    in time $\Oh(n^{2+\beta/2-\smallconst})$.

    No such algorithm exists even if we restrict the instances to additionally satisfy the following:
    \begin{itemize}
        \item The values of the weight function $w$ are rationals with a
            common $\Oh(\log n)$-bit integer denominator.
        \item We have $\w{a}{b}=\w{b}{a}\in \interval{1}{2}$ for all distinct $a,b\in
            \Sigma\cup\{\emptystring\}$.
        \item We have $\w{a}{\emptystring}=1$ if $a\in \Sigma$ occurs in any of the
            strings $X_i$.
        \item We have $\w{\emptystring}{b}=2$ if $b\in \Sigma$ occurs in $Y$.
        \item The threshold satisfies $k \in \intervalco{2y-x}{2y-x+1}$.
        \item Any two consecutive strings have a small Hamming distance,
            that is, $\max_{i=1}^{m-1} \hd(X_i,X_{i+1})\le n^{1-\beta}$.\qedhere
    \end{itemize}
\end{restatable*}

\subparagraph*{From Batched Weighted Edit Distance to Bounded Weighted Edit Distance.}

We may view Batched Weighted Edit Distance as the problem
that crystallizes out the hardness that our overall reduction yields---it is easily
imaginable that Batched Weighted Edit Distance can be used as the start of reductions to
many other problems in the area (apart from Bounded Weighted Edit Distance).
As a simple and illustrative example, consider the following dynamic variant of the
Weighted Edit Distance problem.

\defproblem{Dynamic Weighted Edit Distance}%
{Alphabet \(\Sigma\), strings \(X,Y \in \Sigma^{*}\), weight function
\(w: (\Sigma \cup \emptystring)^2 \to \Real_{\ge 0}\)}%
{Substitute a single character from \(X\) with another character from \(\Sigma\)}%
{Compute \(\wed(X, Y)\)}

Observe that any hardness for the above version of the dynamic problem readily extends to
more complicated variations, where edits (including insertions or deletions) to both
strings are possible.
Our hardness for Batched Weighted Edit Distance readily yields the following lower bound.

\begin{corollary}\label{cor:dyn-main-cor}
    Let $\scb \ge 0$ and $\smallconst > 0$ denote real numbers.
    Assuming the {\sc apsp} Hypothesis, there is no algorithm that
    for every $n\in \Int_{\ge 1}$
    solves every \emph{offline instance} of the Dynamic Weighted Edit Distance
    problem on length-$\Theta(n)$ strings \(X\) and \(Y\)
    that consists in $\Theta(n)$ updates and $\Theta(n^{\beta})$
    queries in \(\Oh(n^{2+\scb/2-\smallconst})\) total time.
\end{corollary}
\begin{proof}
    Suppose that we are given an instance of Batched Weighted Edit Distance. We initialize
    the Dynamic Weighted Edit Distance data structure with the first string of the batch. Then, we
    issue a query and use updates to transform the current string into the next string of
    the batch. In the end, we check if the minimum of the query results is below the
    threshold and return the corresponding answer.

    Using that there are at most \(n^{\scb}\) strings in the batch and that we have to
    update at most \(n^{1-\scb}\) characters between each query, we conclude the total
    number of updates is bounded by \(\Oh(n)\). Further, the number of strings in the
    batch directly corresponds to the number of queries.
\end{proof}
We remark that the approach of \cite[Section~4]{Charalampopoulos20a}, although originally
described in terms of a \emph{uniform} weight function
that assigns the same weight to all substitutions and the same weight to all insertions
and deletions,
seamlessly generalizes to arbitrary weight functions and provides a matching upper bound to \cref{cor:dyn-main-cor}
(in the online version maintaining $X$ and $Y$ subject to edits).

As a more substantial application of \cref{thm:batched}, in~\cref{sec:main-wed-lowerbound} we show how to obtain \cref{mthm:lb}.

\mthmlb*

The key technical idea is to define suitable ``intermediate'' strings $\Xb_i$ that can be aligned
to \(X_i\) and \(X_{i+1}\) for the same cost $h \le n^{1-\beta}$ (independent of $i$).
For this, starting with \(X_i\), we construct a string \(X_i^{\bot}\) by replacing
with a fresh symbol \(\bot\) exactly \(h\) characters of \(X_i\), including all characters that differ in \(X_i\) and
\(X_{i+1}\). If $\w{a}{\bot}=\w{\bot}{a}=1$ for every $a\in \Sigma$, this construction guarantees that $\wed(X_i,\Xb_i)=\wed(\Xb_i,X_{i+1})=h$ holds for every $i$.

Now, consider (simplified) strings
\[
    \hX \coloneqq X_1 \cdot \bigodot_{i=2}^{m} \left(U  Y  V X_i \right),\quad
    \hY \coloneqq \Xb_{0}\cdot \bigodot_{i=1}^{m} \left(U Y V \Xb_i\right),
\] where \(U\) and \(V\) are suitable separators that essentially can only be
either deleted or aligned to themselves.
Observe the string \(\hY\) contains one block \(U Y  V \Xb_i\) more
compared to the number of \(U Y V  X_i \) blocks in \(\hX\).
This means that an optimal alignment has to align exactly one \(X_i\) to \(Y\). Every
\(X_j\) before \(X_i\) is then aligned to \(\Xb_{j-1}\), and every \(X_j\) after \(X_i\) is
then aligned to \(\Xb_j\)---here it is critical that a string \(\Xb_i\) can be aligned to
both \(X_i\) and \(X_{i+1}\) for the same cost $h$, so that the total cost of the alignment
depends only on the cost of aligning \(X_i\) to~\(Y\). Similarly, since both $\Xb_{i-1}$
and $\Xb_{i}$ are inserted, it is important the cost $x$ of these insertions is independent of $i$.

In total, we show that our construction works essentially as a minimum operator, selecting the
minimum-weight pair \((X_i, Y)\) as desired.

The constraints $h\le n^{1-\beta}$ and $m\le n^{\beta}$ guarantee that $\wed(\hX,\hY)=\Theta(hm+x+k)=\Theta(n)$
and $|\hX|,|\hY|=\Theta(n^{1+\beta})$. Thus, to derive \cref{mthm:lb}, it suffices to use \cref{thm:batched}
for $\beta = (1/\kappa)-1$. This explains why \cref{mthm:lb} requires $\kappa \ge 0.5$; otherwise,
$\beta > 1$ would imply $X_1=\ldots=X_{m}$ due to $\hd(X_i,X_{i+1})\le n^{1-\beta}$.

\subsection{Open Problems}
The main problem that stems from our work is to settle the complexity of Bounded Weighted
Edit Distance for $n^{1/3} \le k \le n^{1/2}$. In that case, our algorithm works in
$\Ohtilde(\sqrt{nk^3})$ time, but the lower-bound construction
only shows that an $\Oh(n+k^{2.5-\smallconst})$-time solution would refute the \apsp hypothesis.

Another fascinating problem is whether our lower bound can be circumvented for some simple
classes of weight functions (other than those assigning uniform weights). In particular,
our lower-bound construction requires a relatively big alphabet (of size $\Theta(k)$). Its
size can likely be reduced to $\Theta(k^\smallconst)$ if we assume the stronger
Negative-Weight $d$-Clique Hypothesis~\cite{BGMW20}, but we are not aware of any faster
algorithm if the alphabet size is a small constant. A specific weight function of
potential interest is the $\ell_1$ distance on $\mathbb{Z}_{\ge 0}$, with $\emptystring$
mapped to $0$; it can be used to model the tree edit distance of subdivided stars.

\subsection{Outline}
The paper is organized as follows.
In \cref{sec:preliminaries} we give formal preliminaries and establish some notation.
In \cref{sec:alg} we give our algorithms for weighted edit distance, proving \cref{mthm:algorithm,mthm:pillar}.
In \cref{sec:finegrained} we give some preliminaries and preparations for our fine-grained lower bounds.
In \cref{sec:batched-wed-lowerbound} we give our lower bound for Batched Weighted Edit Distance, proving \cref{thm:batched}.
Finally, in \cref{sec:main-wed-lowerbound} we give our lower bound for weighted edit distance, proving \cref{mthm:lb}.

%% file: figures/fig-0.tex
\begin{tikzpicture}
    \tikzset{vertex/.style = {fill,circle, minimum size=#1,
        inner sep=0pt, outer sep=1pt,line width=1pt},
        vertex/.default = 3pt %
    }
    \tikzset{svertex/.style = {fill=white,draw=red,circle, minimum size=#1,
        inner sep=0pt, outer sep=0.3pt,line width=1pt},
        svertex/.default = 3.4pt %
    }

    \tikzset{edge/.style={white, double=black, line width=.3pt, double distance=#1},
    edge/.default = 1pt}
    \tikzset{sedge/.style={red, double=white, line width=1pt, double distance=1.1pt}}
    \tikzset{hedge/.style={white, double=black!70, line width=.15pt, double distance=#1},
    hedge/.default = .6pt}
    \tikzset{shedge/.style={red!50, double=white, line width=.5pt, double distance=.7pt}}

    \pgfmathsetmacro{\dimp}{5}
    \pgfmathsetmacro{\dimq}{4}

    \pgfmathsetmacro{\selj}{2}
    \pgfmathsetmacro{\seli}{3}

    \pgfmathsetmacro{\dimppo}{int(\dimp+1)}
    \pgfmathsetmacro{\dimppt}{int(\dimp+2)}

    \foreach\j in {1,...,\dimq}{
        \pgfmathsetmacro{\jmo}{int(\j-1)}
        \ifthenelse{\j=\selj}{
            \node[svertex] (\dimppo;\j) at (\dimppo,-\j) {};
            \node[red,anchor=south] at (\dimppo,-\j+.05) {\footnotesize \(b_\jmo\)};
            \node[svertex] (m0;\j) at (\dimppt,-\j) {};
            \node[red,anchor=south] at (\dimppt,-\j+.05) {\footnotesize \(b'_\jmo\)};

            \draw[sedge,double=blue!80!green] (\dimppo;\j) -- (m0;\j);
        }{
            \node[vertex] (\dimppo;\j) at (\dimppo,-\j) {};
            \node[anchor=south] at (\dimppo,-\j+.05) {\footnotesize \(b_\jmo\)};
            \node[vertex] (m0;\j) at (\dimppt,-\j) {};
            \node[anchor=south] at (\dimppt,-\j+.05) {\footnotesize \(b'_\jmo\)};

            \draw[edge,double=blue!80!green] (\dimppo;\j) -- (m0;\j);
        }
        \node[anchor=south,outer sep = 0] at
        ($(\dimppo;\j)!.5!(m0;\j)+(-.0,.0)$)
        {\color{blue!80!green} \tiny \(B_{\jmo,\ell}\)};
    }
    \foreach\i in {1,...,\dimp}{
        \pgfmathsetmacro{\imo}{int(\i-1)}
        \ifthenelse{\i=\seli}{
            \node[svertex] (\i;0) at (\i,0) {};
            \node[red,anchor=south] at (\i+.05,0) {\footnotesize \(a_\imo\)};

            \node[svertex,opacity=.5] (m\i;0) at (\dimppt + \i,0) {};
            \node[red,anchor=south,opacity=.5] at (\dimppt +\i+.05,0) {\footnotesize
            \(a'_\imo\)};
        }{
            \node[vertex] (\i;0) at (\i,0) {};
            \node[anchor=south] at (\i+.05,0) {\footnotesize \(a_\imo\)};

            \node[vertex,opacity=.5] (m\i;0) at (\dimppt + \i,0) {};
            \node[anchor=south,opacity=.5] at (\dimppt + \i+.05,0) {\footnotesize
            \(a'_\imo\)};
        }

        \foreach\j in {1,...,\dimq}{
            \pgfmathsetmacro{\ijj}{int(((\i==\seli) && (\j<\selj)) || ((\i>\seli) &&
            (\j==\selj))))}
            \ifthenelse{\ijj=1}{
                \node[svertex] (\i;\j) at (\i,-\j) {};
            }{
                \node[vertex] (\i;\j) at (\i,-\j) {};
            }
            \pgfmathsetmacro{\ijj}{int(((\i==\seli) && (\j<=\selj)) || ((\i<=\seli) &&
            (\j==\selj))))}
            \ifthenelse{\ijj=1}{
                \node[svertex,opacity=.5] (m\i;\j) at (\dimppt + \i,-\j) {};
            }{
                \node[vertex,opacity=.5] (m\i;\j) at (\dimppt + \i,-\j) {};
            }
            \pgfmathsetmacro{\ijj}{int(((\i>=\seli) &&
            (\j==\selj))))}
            \ifthenelse{\ijj=1}{
                \node[svertex] (\i;\j;t) at (\i+.3,-\j) {};
            }{
                \node[vertex] (\i;\j;t) at (\i+.3,-\j) {};
            }

            \pgfmathsetmacro{\ijj}{int((\i==\seli) && (\j<=\selj))}
            \ifthenelse{\ijj=1}{
                \node[svertex] (\i;\j;s) at (\i,-\j+.3) {};
                \pgfmathsetmacro{\jmo}{int(\j-1)}
                \draw[sedge] (\i;\j;s) -- (\i;\jmo);

                \draw[shedge] (m\i;\j) -- (m\i;\jmo);
            }{
                \node[vertex] (\i;\j;s) at (\i,-\j+.3) {};
                \pgfmathsetmacro{\jmo}{int(\j-1)}
                \draw[edge] (\i;\j;s) -- (\i;\jmo);

                \draw[hedge] (m\i;\j) -- (m\i;\jmo);
            }
            \pgfmathsetmacro{\ijj}{int((\i==\seli) && (\j<\selj))}
            \ifthenelse{\ijj=1}{
                \draw[sedge] (\i;\j;s) -- (\i;\j);
            }{
                \draw[edge] (\i;\j;s) -- (\i;\j);
            }
            \pgfmathsetmacro{\ijj}{int((\i==\seli) && (\j==\selj))}
            \ifthenelse{\ijj=1}{
                \draw[sedge,double=blue!50!red] (\i;\j;s) -- (\i;\j;t);
            }{
                \draw[edge,double=blue!50!red] (\i;\j;s) -- (\i;\j;t);
            }

            \pgfmathsetmacro{\ijj}{int((\i>\seli) && (\j==\selj))}
            \ifthenelse{\ijj=1}{
                \draw[sedge] (\i;\j) -- (\i;\j;t);
            }{
                \draw[edge] (\i;\j) -- (\i;\j;t);
            }
            \pgfmathsetmacro{\imo}{int(\i-1)}
            \pgfmathsetmacro{\jmo}{int(\j-1)}

            \node[anchor=south west,outer sep = 0] at
            ($(\i;\j;s)!.5!(\i;\j;t)+(-.09,-.15)$)
            {\color{blue!50!red} \tiny \(A_{\imo, \jmo}\)};
        }
    }
    \foreach\i in {1,...,\dimp}{
        \foreach\j in {1,...,\dimq}{
            \pgfmathsetmacro{\ijj}{int((\i>=\seli) && (\j==\selj))}
            \ifthenelse{\ijj=1}{
                \pgfmathsetmacro{\ipo}{int(\i+1)}
                \draw[sedge] (\i;\j;t) -- (\ipo;\j);
            }{
                \pgfmathsetmacro{\ipo}{int(\i+1)}
                \draw[edge] (\i;\j;t) -- (\ipo;\j);
            }
            \pgfmathsetmacro{\ijj}{int((\i<=\seli) && (\j==\selj))}
            \ifthenelse{\ijj=1}{
                \pgfmathsetmacro{\imo}{int(\i-1)}
                \draw[shedge] (m\i;\j) -- (m\imo;\j);
            }{
                \pgfmathsetmacro{\imo}{int(\i-1)}
                \draw[hedge] (m\i;\j) -- (m\imo;\j);
            }
        }
    }
\end{tikzpicture}

%% file: figures/fig-1.tex
\begin{tikzpicture}[scale=1.2,transform shape]
    \tikzset{vertex/.style = {fill,circle, minimum size=#1,
        inner sep=0pt, outer sep=1pt,line width=1pt},
        vertex/.default = 2pt %
    }
    \tikzset{svertex/.style = {fill=white,draw=red,circle, minimum size=#1,
        inner sep=0pt, outer sep=0.3pt,line width=1pt},
        svertex/.default = 2.4pt %
    }

    \tikzset{htedge/.style={black!50, line width=.3pt, -{Latex[length=1mm,
        width=.5mm]}}}
    \tikzset{tedge/.style={black!50, line width=.5pt, -{Latex[length=1.4mm,
        width=.8mm]}}}
    \tikzset{pedge/.style={blue!50!red, line width=.5pt, -{Latex[length=1.4mm,
        width=.8mm]}}}
    \tikzset{bedge/.style={blue!80!green, line width=.5pt, -{Latex[length=1.4mm,
        width=.8mm]}}}

    \tikzset{redge/.style={red, line width=1pt, -{Latex[length=1.6mm,
        width=1mm]}}}

    \pgfmathsetmacro{\dimp}{5}
    \pgfmathsetmacro{\dimq}{4}

    \pgfmathsetmacro{\seli}{2}
    \pgfmathsetmacro{\selj}{1}

    \pgfmathsetmacro{\lenx}{int(2 * \dimp + 1)}
    \pgfmathsetmacro{\leny}{int(2 * \dimp+ \dimq)}

    \pgfmathsetmacro{\selipj}{int(\seli + \selj)}
    \pgfmathsetmacro{\selqmipj}{int(2 *\dimp-\seli +\selj + 2)}

    \pgfmathsetmacro{\dimqmo}{int(\dimq-1)}
    \pgfmathsetmacro{\dimpmo}{int(\dimp-1)}
    \pgfmathsetmacro{\dimppo}{int(\dimp+1)}
    \pgfmathsetmacro{\dimppt}{int(\dimp+2)}

    \foreach\i in {1,...,\lenx}{
        \pgfmathsetmacro{\imo}{int(\i-1)}
        \pgfmathsetmacro{\sele}{int((\seli<\i) && (\lenx-\seli>=\i)}
        \ifthenelse{\sele=1}{
            \node[red] at (\i - .5, 0.5) {\footnotesize \(X\position{\imo}\)};
        }{
            \node at (\i - .5, 0.5) {\footnotesize \(X\position{\imo}\)};
        }
    }

    \foreach\j in {1,...,\leny}{
        \pgfmathsetmacro{\jmo}{int(\j-1)}
        \node at (-0.5, -\j + .5) {\footnotesize \(Y\position{\jmo}\)};
    }

    \foreach\i in {0,...,\lenx}{
        \foreach\j in {0,...,\leny}{
            \node[vertex] (x\i-y\j) at (\i,-\j) {};
        }
    }
    \ifthenelse{\selipj>0}{
    \foreach\j in {1,...,\selipj}{
        \pgfmathsetmacro{\jmo}{int(\j-1)}
        \draw[redge] (x\seli-y\jmo) -- (x\seli-y\j);
    }
    }{}
    \ifthenelse{\selqmipj<\leny}{
    \foreach\j in {\selqmipj,...,\leny}{
        \pgfmathsetmacro{\jmo}{int(\j-1)}
        \pgfmathsetmacro{\selit}{int(\lenx-\seli)}
        \draw[redge] (x\selit-y\jmo) -- (x\selit-y\j);
    }}{}
    \foreach\i in {0,...,\lenx}{
        \foreach\j in {1,...,\leny}{
            \pgfmathsetmacro{\jmo}{int(\j-1)}
            \draw[htedge] (x\i-y\jmo) -- (x\i-y\j);
        }
    }
    \foreach\i in {1,...,\lenx}{
        \pgfmathsetmacro{\ipdimq}{int(min(\leny,\i+\dimq - 1))}
        \foreach\j in {\i,...,\ipdimq}{
            \pgfmathsetmacro{\imo}{int(\i-1)}
            \pgfmathsetmacro{\imt}{int(\i-2)}
            \pgfmathsetmacro{\ipo}{int(\i+1)}
            \pgfmathsetmacro{\jmi}{int(\j-\i)}
            \pgfmathsetmacro{\jmo}{int(\j-1)}
            \pgfmathsetmacro{\jpo}{int(\j+1)}
            \pgfmathsetmacro{\jpi}{int(\j-\i+1)}
            \pgfmathsetmacro{\qmj}{int(\dimq + \i-\j)}

            \pgfmathsetmacro{\sele}{int((\selj==\jmi) && (\seli<\i) && (\lenx-\seli>=\i)}
            \ifthenelse{\sele=1}{
                \draw[redge] (x\imo-y\jmo)
                --
                (x\i-y\j);
            }{}

            \ifthenelse{\i<\dimp}{
                \draw[pedge] (x\imo-y\jmo)
                to
                node[above, sloped,scale=.8,transform shape]
                {\tiny\color{blue!50!purple}\(A_{\imo,\jmi} - A_{\i,\jmi}\qquad\)}
                node[below, sloped,scale=.8,transform shape]
                {\tiny\color{black!70}\(\qmj D\)}
                (x\i-y\j);
            }{}
            \ifthenelse{\i=\dimp}{
                \draw[pedge] (x\imo-y\jmo)
                to
                node[above, sloped,scale=.8,transform shape]
                {\tiny\color{blue!50!purple}\(A_{\imo,\jmi}\qquad\)}
                node[below, sloped,scale=.8,transform shape]
                {\tiny\color{black!70}\(\qmj D\)}
                (x\i-y\j);
            }{}
            \ifthenelse{\imo=\dimp}{
                \draw[bedge] (x\imo-y\jmo)
                to node[
                above,sloped,scale=.8,transform shape]
                {\tiny\color{blue!80!green}\(B_{\jmi,\ell}\qquad\)}
                (x\i-y\j);
            }{}
            \ifthenelse{\imo>\dimp}{
                \draw[tedge] (x\imo-y\jmo) to
                node[below, sloped,scale=.8,transform shape]
                {\tiny\color{black!70}\(\jpi D\)}
                (x\i-y\j);
            }{}
        }
    }
\end{tikzpicture}

%% file: sections/preliminaries.tex
\section{Preliminaries}\label{sec:preliminaries}

For integers $i, j \in \Int$, we write $\fragmentco{i}{j} \coloneqq \{i,i+1,\dots,j-1\}$ and
$\fragment{i}{j} = \{i,i+1,\dots,j\}$;
we define the sets $\fragmentoo{i}{j}$ and $\fragmentoc{i}{j}$ analogously.

\subparagraph*{Strings.} A string $X = X\position{0}\dots X\position{n-1} \in \Sigma^n$
is a sequence of $|X| = n$ characters over an
alphabet $\Sigma$; \(|X|\) is the \emph{length} of \(|X|\).
For a \emph{position} $i \in \fragmentco{0}{n}$, we
say that $X\position{i}$ is the $i$-th character of $X$.
We denote the empty string over $\Sigma$ by $\emptystring$.
Given indices $0 \leq i \leq j \leq |X|$, we say that $X\fragmentco{i}{j} \coloneqq X\position{i}\cdots Y\position{j-1}$ is a
\emph{fragment} of $X$.
We may also write $X\fragment{i}{j-1}, X\fragmentoc{i-1}{j-1}$,
or $X\fragmentoo{i-1}{j-1}$
for the fragment $X\fragmentco{i}{j}$,

We say that $X$ occurs as a substring of a string $Y$, if there are
$0 \leq i \leq j \leq |Y|$ such that $X = Y\fragmentco{i}{j}$.

\subparagraph*{Alignments and (Weighted) Edit Distances.}
We start with the crucial notion of an \emph{alignment}, which give us a formal way to
describe a sequence of edits to transform a string into another.

\begin{definition}[Alignment~{\cite[Definition 2.3]{DGHKS23}}]
    A sequence $\A = (x_t, y_t)_{t=0}^m$ is an \emph{alignment} of $X\fragmentco{x}{x'}$ onto
    $Y\fragmentco{y}{y'}$, denoted by $\A: X\fragmentco{x}{x'} \onto Y\fragmentco{y}{y'}$ if
    $(x_0, y_0) = (x, y)$, $(x_m, y_m) = (x', y')$, and $(x_{t+1}, y_{t+1}) \in \{(x_t+1, y_t), (x_t, y_t+1), (x_t+1, y_t+1)\}$.

    We write
    $\Als(X\fragmentco{x}{x'}, Y\fragmentco{y}{y'})$
    for the set of all alignments of $X\fragmentco{x}{x'}$ onto $Y\fragmentco{y}{y'}$.
\end{definition}

For an alignment $\A = (x_t, y_t)_{t=0}^m\in \Als(X\fragmentco{x}{x'},
Y\fragmentco{y}{y'})$ and an index $t \in \fragmentco{0}{m}$, we say that
\begin{itemize}
    \item $\A$ \emph{deletes} $X\position{x_t}$ if $(x_{t+1}, y_{t+1}) = (x_t+1, y_t)$.
    \item $\A$ \emph{inserts} $Y\position{y_t}$ if $(x_{t+1}, y_{t+1}) = (x_t, y_t+1)$.
    \item $\A$ \emph{aligns} $X\position{x_t}$ to $Y\position{y_t}$, denoted by
        $X\position{x_t} \aonto{\A} Y\position{y_t}$
    if $(x_{t+1}, y_{t+1}) = (x_t+1, y_t+1)$.
    \item $\A$ \emph{matches} $X\position{x_t}$ with $Y\position{y_t}$,
        if $X\position{x_t}
        \aonto{\A} Y\position{y_t}$ and
    $X\position{x_t} = Y\position{y_t}$.
    \item $\A$ \emph{substitutes} $X\position{x_t}$ for $Y\position{y_t}$ if
        $X\position{x_t} \aonto{\A} Y\position{y_t}$ but
    $X\position{x_t} \neq Y\position{y_t}$.
\end{itemize}
Insertions, deletions, and substitutions are jointly called (character) \emph{edits}.

Given $\A = (x_t, y_t)_{t=0}^m \in \Als(X, Y)$, we define the \emph{inverse alignment}
as $\A^{-1} \coloneqq (y_t, x_t)_{t=0}^m \in \Als(Y, X)$.

Given an alphabet $\Sigma$, we set $\bar\Sigma \coloneqq \Sigma \cup \{\emptystring\}$.
We call $w$ a \emph{weight function} if $w : \bar\Sigma \times \bar\Sigma \to \Real_{\geq 0} \cup \{\infty\}$
and we have $\w{a}{a} = 0$ for all $a \in \bar\Sigma$.
Note that \(w\) does not need to satisfy the triangle inequality nor does \(w\) need to be
symmetric.

We write $\wed_\A(X\fragmentco{x}{x'}, Y\fragmentco{y}{y'})$ for
the \emph{cost} of an alignment $\A \in \Als(X\fragmentco{x}{x'}, Y\fragmentco{y}{y'})$ with
respect to a weight function $w$, that is, for the total cost of edits made by $\A$, where
\begin{itemize}
    \item the cost of deleting $X\position{x}$ is $\w{X\position{x}}{\emptystring}$,
\item the cost of inserting $Y\position{y}$ is $\w{\emptystring}{Y\position{y}}$,
    \item the cost of aligning $X\position{x}$ with $Y\position{y}$ is
        $\w{X\position{x}}{Y\position{y}}$.
\end{itemize}

We define the \emph{weighted edit distance} of strings $X, Y \in \Sigma^*$
with respect to a weight function $w$ as
$\wed(X, Y) \coloneqq \min_{\A \in \Als(X, Y)} \wed_\A(X, Y)$.
For an integer $k \geq 0$, we also define a capped version
\[
    \wed_{\leq k}(X, Y) \coloneqq
        \begin{cases}
            \wed(X, Y) & \text{if } \wed(X, Y) \leq k,\\
            \infty & \text{otherwise}.
        \end{cases}
\]

\begin{definition}[Alignment Graph]\label{def:alignment-graph}
For strings $X,Y\in \Sigma^*$ and a weight function $w:\bar\Sigma\to \Real_{\ge 0}$,
we define the \emph{alignment graph} $\AG^w(X,Y)$ as follows.
\(\AG^{w}(X, Y)\) has vertices $\fragment{0}{|X|}\times \fragment{0}{|Y|}$,
\begin{itemize}
    \item horizontal edges $(x,y)\to (x+1,y)$ of cost $\w{X\position{x}}{\emptystring}$
        for $(x,y)\in \fragmentco{0}{|X|}\times \fragment{0}{|Y|}$,
    \item vertical edges $(x,y)\to (x,y+1)$ of cost $\w{\emptystring}{Y\position{y}}$
        for $(x,y)\in \fragment{0}{|X|}\times \fragmentco{0}{|Y|}$, and
    \item diagonal edges $(x,y)\to (x+1,y+1)$ of cost $\w{X\position{x}}{Y\position{y}}$
        for $(x,y)\in \fragment{0}{|X|}\times \fragmentco{0}{|Y|}$.
        \qedhere
\end{itemize}
\end{definition}

We visualize the alignment graph $\AG^w(X, Y)$ as a grid graph with $|X|+1$ columns
and $|Y|+1$ rows. We think of the vertex $(0,0)$ as the top left vertex of the grid,
and a vertex $(x, y)$ in the $x$-th column and $y$-th row.

Observe that we can interpret $\Als(X\fragmentco{x}{x'},Y\fragmentco{y}{y'})$
as the set of $(x,y)\leadsto (x',y')$ paths in $G:=\AG^w(X,Y)$.
Moreover, $\wed_\A(X\fragmentco{x}{x'},Y\fragmentco{y}{y'})$
is the cost of $\A$ interpreted as a path in $G$ and
thus $\ed^w(X\fragmentco{x}{x'},Y\fragmentco{y}{y'})=\dist_G((x,y), (x',y'))$.

If for every $a, b \in \bar\Sigma$ we have that $\w{a}{b} = 1$ if $a \neq b$
and $\w{a}{b} = 0$ otherwise, then $\wed(X, Y)$ corresponds
to the standard \emph{unweighted} edit distance (also known as Levenshtein distance~\cite{Levenshtein66}).
For this case, we drop the subscript $w$ in $\wed$ and $\wed_\A$.

We say that an alignment $\A \in \Als(X, Y)$ is optimal if $\wed_\A(X, Y) = \wed(X, Y)$.

In this work, we consider the weight function $w$ to be \emph{normalized} so that
$w(a,b)\ge 1$ holds for all $a,b\in \bar\Sigma$ with $a\ne b$.
As a result, $\wed_\A(X,Y)\ge \ed_\A(X,Y)$ holds for every alignment $\A\in \Als(X,Y)$.
This lower bound on the costs for each edit allows us to compute $\wed_{\leq k}(X, Y)$
faster than the standard $\Oh(n^2)$-time dynamic program.

\begin{fact}[{\cite[Proposition 2.16]{DGHKS23}}]\label{prop:baseline-wed}
    Given strings $X, Y \in \Sigma^{\leq n}$, an integer $k \geq 1$ and (oracle access to)
    a normalized weight function
    $w: \Sigma^2 \to \Real_{\geq 0}$, the value $\wed_{\leq k}(X, Y)$ can be computed in
    $\Oh(nk)$ time.
    \lipicsEnd
\end{fact}

The \emph{breakpoint representation} of an alignment $\A=(x_t,y_t)_{t=0}^m\in \Als(X,Y)$
is the subsequence of $\A$ consisting of pairs $(x_t,y_t)$ such that $t\in \{0,m\}$ or
$\A$ does not match $X\position{x_t}$ with $Y\position{y_t}$.
Note that the size of the breakpoint representation is at most $2+\ed_\A(X,Y)$ and that it
can be used to retrieve the entire alignment:
for any two consecutive elements $(x',y'),(x,y)$ of the breakpoint representation, it
suffices to add $(x-\delta),(y-\delta)$ for $\delta \in \fragmentco{0}{\max(x-x',y-y')}$.

Given an alignment $\A = (x_t, y_t)_{t=0}^m \in \Als(X, Y)$, for every $\ell, r \in
\fragment{0}{m}$ with $\ell \leq r$
we say that $\A$ aligns $X\fragmentco{x_\ell}{x_r}$ to $Y\fragmentco{y_\ell}{y_r}$ and
denote it by
$X\fragmentco{x_\ell}{x_r} \aonto{\A} Y\fragmentco{y_\ell}{y_r}$.
We denote the cost of the induced alignment
of $X\fragmentco{x_\ell}{x_r}$ onto $Y\fragmentco{y_\ell}{y_r}$
by $\wed_\A(X\fragmentco{x_\ell}{x_r}, Y\fragmentco{y_\ell}{y_r})$.

Given $\A: X \onto Y$ and a fragment $X\fragmentco{x_\ell}{x_r}$ of $X$, we write $\Als(X\fragmentco{x_\ell}{x_r})$ for the fragment
$Y\fragmentco{y_\ell}{y_r}$ of $Y$ where
\[
    y_\ell \coloneqq \min\{y \mid (x_\ell, y) \in \A\} \quad \text{ and } \quad
        y_r \coloneqq \begin{cases} |Y| & \text{if } x_r = |X|, \\ \min\{y \mid (x_r, y) \in \A\} &\text{otherwise.}\end{cases}
\]
Intuitively, $\Als(X\fragmentco{x_\ell}{x_r})$ is the fragment that
$\A$ aligns $X\fragmentco{x_\ell}{x_r}$ onto.

\begin{fact}[{Triangle Inequality \cite[Fact 2.5]{DGHKS23}}]\label{fct:triangle}
    Consider strings $X,Y,Z\in \Sigma^*$ as well as alignments $\A : X\onto Y$ and $\B : Y\onto Z$.
    There exists a \emph{composition alignment} $\B \circ \A : X \onto Z$ satisfying the following properties for all $x\in \fragmentco{0}{|X|}$ and $z\in\fragmentco{0}{|Z|}$:
    \begin{itemize}
        \item $\B\circ \A$ aligns $X\position{x}$ to $Z\position{z}$ if and only if there exists $y\in \fragmentco{0}{|Y|}$ such that $\A$ aligns $X\position{x}$ to $Y\position{y}$ and $\B$ aligns $Y\position{y}$ to $Z\position{z}$.
        \item $\B\circ \A$ deletes $X\position{x}$ if and only if $\A$ deletes $X\position{x}$ or there exists $y\in \fragmentco{0}{|Y|}$ such that $\A$ aligns $X\position{x}$ to $Y\position{y}$ and $\B$ deletes $Y\position{y}$.
        \item $\B\circ \A$ inserts $Z\position{z}$ if and only if $\B$ inserts $Z\position{z}$ or there exists $y\in \fragmentco{0}{|Y|}$ such that $\A$ inserts $Y\position{y}$ and $\B$ aligns $Y\position{y}$ to $Z\position{z}$.
    \end{itemize}

    If a weight function $w$ satisfies the \emph{triangle inequality}, that is, $\w{a}{b}\le \w{a}{c}+\w{c}{b}$ holds for all $a,b,c\in \bar\Sigma$,  then $\wed_{\B \circ \A}(X, Z) \leq \wed_{\A}(X, Y) + \wed_{\B}(Y, Z)$.
\end{fact}

\begin{fact}[{\cite[Fact 2.6]{DGHKS23}}]\label{fct:substring}
    If a weight function $w$ satisfies the triangle inequality,
    then $\wed(X,X\fragmentco{i}{j})=\wed(X\fragmentco{0}{i}\cdot X\fragmentco{j}{|X|},\emptystring)$
    holds for every string $X$ and its fragment $X\fragmentco{i}{j}$.
\end{fact}

\begin{corollary}\label{cor:greedy}
    If a weight function $w$ satisfies the triangle inequality,
    then $\wed(aX,aY)=\wed(X,Y)=\wed(Xa,Ya)$ holds for every character $a\in \Sigma$ and strings $X,Y\in \Sigma^*$.
\end{corollary}
\begin{proof}
    By symmetry, it suffices to prove $\wed(X,Y)=\wed(X',Y')$ for $X'=Xa$ and $Y'=Ya$.
    The inequality $\wed(X',Y')\le \wed(X,Y)$ holds trivially: an alignment $\A\in \Als(X,Y)$
    can be extended to an alignment $\A'\in \Als(Xa,Ya)$ of the same cost by appending $(|X|+1,|Y|+1)$,
    which corresponds to matching the trailing characters of $X'$ and $Y'$.
    As for the converse implication, consider an optimal alignment $\A'\in \Als(X',Y')$.
    Consider the last pair $(x,y)\in \A'$ such that $x\le |X|$ and $y\le |Y|$.
    Either of these inequalities must be an equality. By symmetry, assume without loss of generality that $y=|Y|$.
    This means that $\A'$ aligns $X\fragmentco{0}{x}$ with $Y$
    and $X'\fragmentco{x}{|X'|}$ with $Y'\position{|Y|}=a$.
    Since $a$ is a suffix of $X'\fragmentco{x}{|X'|}$,
    \cref{fct:substring} implies $\wed(X'\fragmentco{x}{|X'|},Y'\position{|Y|})= \wed(X\fragmentco{x}{|X|},\emptystring)$.
    Consequently, an alignment $\A \in \Als(X,Y)$ that proceeds as $\A$ until $(x,y)$ but then deletes $X\fragmentco{x}{|X|}$
    has the same cost as $\A'$.
    This implies $\wed(X,Y)\le \wed(X',Y')$.
\end{proof}

\begin{lemma}\label{fct:split-alignment}
    Let $X, Y \in \Sigma^*$ denote strings and let $\A \in \Als(X,Y)$ denote an alignment.
    For every $(x,y)\in \A$, we have
    $\wed(X,Y)\le \wed(X\fragmentco{0}{x}, Y\fragmentco{0}{y}) +
    \wed(X\fragmentco{x}{|X|}, Y\fragmentco{y}{|Y|})\le \wed_{\A}(X,Y)$.
\end{lemma}
\begin{proof}
    Consider the alignment graph $G \coloneqq \AG(X,Y)$.
    Since $\dist_G$ satisfies the triangle inequality, we have
    \begin{align*}
        \wed(X,Y) &= \dist_G((0,0),(|X|,|Y|))\\
                  &\le \dist_G((0,0),(x,y)) + \dist_G((x,y),(|X|,|Y|))\\
                  &= \wed(X\fragmentco{0}{x}, Y\fragmentco{0}{y})
                  + \wed(X\fragmentco{x}{|X|}, Y\fragmentco{y}{|Y|}).
    \end{align*}
    Observe that $\A$ corresponds to a $(0,0)\leadsto (|X|,|Y|)$ path in
    $G \coloneqq \AG^w(X,Y)$ containing $(x,y)$.
    The cost of this path, which is $\wed_{\A}(X,Y)$ must be at least
    \begin{align*}
        &\dist_G((0,0),(x,y)) + \dist_G((x,y),(|X|,|Y|))\\
        &\quad =
        \wed(X\fragmentco{0}{x}, Y\fragmentco{0}{y})
        + \wed(X\fragmentco{x}{|X|}, Y\fragmentco{y}{|Y|});
    \end{align*}
    completing the proof.
\end{proof}

\subparagraph*{The PILLAR Model.}
Charalampopoulos, Kociumaka, and Wellnitz~\cite{Charalampopoulos20} introduced the
\modelname{} model.
The \modelname{} model provides an abstract interface to a set of primitive operations on
strings which can be efficiently implemented in different settings. Thus, an algorithm
developed using the \modelname{} interface does not only yield algorithms in the standard
setting, but also directly yields algorithms in diverse other settings,
for instance, fully compressed, dynamic, etc.

\newcommand{\calX}{\mathcal{X}}
In the~\modelname{} model we are given a family $\calX$ of strings to preprocess. The
elementary objects are fragments $X\fragmentco{\ell}{r}$ of strings $X \in \calX$.
Initially, the model gives access to each
$X \in \calX$ interpreted as $X\fragmentco{0}{|X|}$. Other fragments can be retrieved via
an \extractOpName{} operation:
\begin{itemize}
    \item $\extractOpName(S, \ell, r)$: Given a fragment $S$ and positions $0 \leq \ell \leq r \leq |S|$, extract
    the fragment $S\fragmentco{\ell}{r}$, which is defined as $X\fragmentco{\ell' + \ell}{r' + r}$ if $S = X\fragmentco{\ell'}{r'}$ for
    $X \in \calX$.
\end{itemize}
Moreover, the \modelname provides the following primitive operations~\cite{Charalampopoulos20}:
\begin{itemize}
    \item $\lceOp{S}{T}$: Compute the length of~the longest common prefix of~$S$ and $T$.
    \item $\lcbOp{S}{T}$: Compute the length of~the longest common suffix of~$S$ and $T$.
    \item $\accOpName(S,i)$: Assuming $i\in \fragmentco{0}{|S|}$, retrieve the character $\accOp{S}{i}$.
    \item $\lenOpName(S)$: Retrieve the length $|S|$ of~the string $S$.
\end{itemize}
Observe that in the original definition~\cite{Charalampopoulos20}, the \modelname{} model
also includes an \(\ipmOpName\) operation to find all (internal) exact occurrences of one
fragment in another. We do not need the \(\ipmOpName\) operation in this
work.\footnote{We still use the name \modelname, and not {\tt PLLAR}, though.}

\subparagraph*{Planar Graph Toolbox and Monge Matrices}\label{sec:alg:sec:alg-periodic:sec:planar-toolbox}

We use tools originating from planar graph algorithms to process (induced subgraphs of)
the alignment graph $\AG^w(X,Y)$.

\begin{theorem}[Klein~\cite{Klein2005}]\label{thm:klein}
    We can preprocess a directed planar graph with non-negative edge weights in
    $\Oh(|V(G)|\log |V(G)|)$ time so that,
    for any two vertices $u,v$ on the outer face, we can compute the distance
    $\dist_G(u,v)$ in $\Oh(\log |V(G)|)$ time.

    Moreover, we can report the shortest $u\leadsto v$ path $P$
    in $\Oh(|P|\log \log\Delta(G))$ time, where $\Delta(G)$ is the maximum degree of $G$.
    \lipicsEnd
\end{theorem}

\begin{fact}[{\cite[Section 2.3]{FR06}}]\label{fct:monge}
    Consider a directed planar graph $G$ with non-negative edge weights. For vertices
    $u_0,\ldots,u_{p-1},v_{q-1},\ldots,v_0$ lying (in this cyclic order) on the outer face
    of $G$, define a $p\times q$ matrix $D$ with $D_{i,j}=\dist_G(u_i,v_j)$.
    Then, $D$ satisfies the \emph{Monge property}, that is,
    $D_{i,j}+D_{i',j'} \le D_{i,j'}+D_{i',j}$
    holds for $0\le i\le i' <p$ and $0\le j\le j'<q$.
\end{fact}

Given an $p \times q$ matrix $A$ and an $q \times r$ matrix $B$, their min-plus product is
the $p \times r$ matrix $C$ defined as
$C_{i,j} \coloneqq \min_k A_{i,k} + B_{k, j}$.

\begin{theorem}[SMAWK algorithm~\cite{SMAWK87}]\label{thm:smawk}
    We can compute the min-plus product of an $p\times q$ Monge matrix
    and an $q\times r$ Monge matrix in $\Oh(p\cdot r\cdot (1+\log(q/\max(p,r))))$ time
    assuming $\Oh(1)$-time random access to the input matrices.
    \lipicsEnd
\end{theorem}

%% file: sections/upper_bounds.tex
\section{Faster Algorithms for Bounded Weighted Edit Distance}\label{sec:alg}

In this section we give our main algorithm. Formally, we prove the following result.

\SetKwFunction{WeightedED}{WeightedED}
\begin{restatable*}{lemma}{lemwed}\label{lem:wed}
    Given two strings $X, Y \in \Sigma^{\leq n}$, and a threshold parameter $k \in \Int_{\geq 0}$,
    there is an algorithm $\WeightedED(X, Y, k)$ that computes $c \coloneqq \wed(X, Y)$
    in time $\Oh(n\log^2n + k\sqrt{nc}\log^{1.5}n)$ provided that $c \leq k$.
    The \modelname implementation of this algorithm runs in time $\Oh(k^2 c \log^2 n)$.
\end{restatable*}

Observe that \cref{lem:wed} directly implies \cref{mthm:algorithm};
it suffices to start with the threshold $k=\Oh(k)$ and double that threshold as long as $\wed(X,Y)>k$.

\mthmalg
\begin{proof}
    We first set $k=\Theta((n/\log^2 n)^{1/3})$ and compute $\wed_{\le k}(X,Y)$ using the \modelname{} algorithm of \cref{lem:wed};
    this takes time $\Oh(n)$.
    Otherwise, we run the standard implementation of \cref{lem:wed}, doubling the threshold $k$ in each execution.
    This takes time $\Oh(\sqrt{nk^3}\log^3 n)$.
\end{proof}

Further, \cref{lem:wed} directly implies \cref{mthm:pillar}.

\mthmpillar
\begin{proof}
    We run the \modelname{} implementation of~\cref{lem:wed}, doubling the threshold $k$
    in each execution; yielding the result
\end{proof}

Our presentation is structured as follows.
In \cref{sec:alg:sec:self-alignments} we introduce our crucial ingredient: alignments of a string to
itself---\emph{self-alignments}---which then gives rise to the new notion of \emph{self-edit distance}.
In \cref{sec:alg:sec:alg-periodic} we give algorithms to compute the weighted edit distance of strings with small self-edit distance.
Finally, in \cref{sec:alg:sec:alg-main} we give our main algorithm to compute the weighted edit distance or two arbitrary strings
by making use of the divide and conquer scheme outlined in \cref{sec:technical-overview}, which proves \cref{lem:wed}.

\subsection{Self-Alignments and the Self-Edit Distance of a String}\label{sec:alg:sec:self-alignments}

We say that an alignment $\A : X \onto X$ is a \emph{self-alignment}
if $\A$ does not align any character $X\position{x}$ to itself.
We define the \emph{self edit distance} of $X$ as $\selfed(X) \coloneqq \min_\A \ed_{\A}(X, X)$,
where the minimization ranges over self-alignments $\A : X \onto X$. In words,
$\selfed(X)$ is the minimum (unweighted) cost of a self-alignment.
We can interpret a self-alignment as a
$(0, 0) \leadsto (|X|, |X|)$ path in the alignment graph $\AG(X, X)$
that does not contain any edge of the main diagonal.

\begin{lemma}[Properties of $\selfed{}$]\label{fct:selfed-properties}
    Let $X \in \Sigma^*$ denote a string. Then, all of the following hold.
    \begin{description}
        \item[Monotonicity.] For any $\ell' \le \ell \le r \le r'\in \fragment{0}{|X|}$, we have
            \(
                \selfed(X\fragmentco{\ell}{r}) \leq \selfed(X\fragmentco{\ell'}{r'}).
            \)
        \item[Sub-additivity.] For any $m \in \fragment{0}{|X|}$, we have
            \(
                \selfed(X) \leq \selfed(X\fragmentco{0}{m}) + \selfed(X\fragmentco{m}{|X|}).
            \)
        \item[Triangle inequality.] For any $Y\in \Sigma^*$, we have
            \(
                \selfed(Y) \le \selfed(X)+2\ed(X,Y).
            \)
    \end{description}
\end{lemma}
\begin{proof}
    \begin{description}
        \item[Monotonicity:] Consider the alignment graph $G \coloneqq \AG(X, X)$,
            and write $P$ for the $(\ell', \ell') \leadsto (r', r')$
            path corresponding to an optimal alignment for $\selfed(X\fragmentco{\ell'}{r'})$.
            Without loss of generality, we can assume that $P$ lies above the main diagonal
            (the parts below the main diagonal can be mirrored along the main diagonal without affecting $\mathrm{cost}(P)$).

            Since $\ell' \leq \ell \leq r'$, the path $P$
            contains a point $(j, \ell)$ for some $j\in \fragment{\ell}{r'}$.
            Similarly, since $\ell' \leq r \le r'$, the path $P$
            contains a point $(r, i)$ for some $i\in \fragment{\ell'}{r}$.
            We consider two cases depending on the order of points $(j,\ell)$ and $(r,i)$ on $P$.

            If $(j,\ell)$ lies before $(r,i)$, then we split $P$ into three subpaths:
            $P_1 : (\ell', \ell') \leadsto(j, \ell)$,
            $P_2 : (j, \ell) \leadsto (r, i)$, and $P_3 : (r, i) \leadsto (r', r')$.
            Now, we construct a self-alignment for $X\fragmentco{\ell}{r}$ by
            choosing $P'$ to be the $(\ell,\ell) \leadsto (r,r)$ path formed by
            beginning with $(\ell, \ell) \leadsto (j, \ell)$,
            then continuing with $P_2$, and then ending with $(r, i) \leadsto (r, r)$.
            Observe that this path does not contain any edge of the main diagonal,
            so it is a valid self-alignment.
            Moreover, we have $\mathrm{cost}(P') = j - \ell + \mathrm{cost}(P_2) + r - i$.
            Further, observe that
            \begin{align*}
                \mathrm{cost}(P_1) &\geq
                \ed(X\fragmentco{\ell'}{\ell}, X\fragmentco{\ell'}{j}) \geq j - \ell
                \intertext{and}
                \mathrm{cost}(P_3) &\geq \ed(X\fragmentco{i}{r'}, X\fragmentco{r}{r'}) \geq
                r - i.
            \end{align*}
            Thus, we obtain that
            \[
                \mathrm{cost}(P')
                = j - \ell + \mathrm{cost}(P_2) + r - i
                \leq \mathrm{cost}(P_1) + \mathrm{cost}(P_2) + \mathrm{cost}(P_3)
                = \mathrm{cost}(P)
                = \selfed(X\fragmentco{\ell'}{r'}).
            \]
            Since $P'$ is a self-alignment, it follows that
            $\selfed(X\fragmentco{\ell}{r}) \leq \mathrm{cost}(P') \leq
            \selfed(X\fragmentco{\ell'}{r'})$,
            as desired.

            It remains to consider the case when $(r,i)$ lies before or coincides with $(j,\ell)$,
            that is, $r\le j$ and $i\le \ell$.
            In this case,
            \[
                \mathrm{cost}(P)
                \geq \ed(X\fragmentco{\ell'}{r}, X\fragmentco{\ell'}{i})+\ed(X\fragmentco{j}{r'}, X\fragmentco{\ell}{r'})
                \geq (r-i)+(j-\ell)
                \geq 2(r-\ell).
            \]
            A self-alignment for $X\fragmentco{\ell}{r}$ of cost $2(r-\ell)$
            can be obtained trivially as $(\ell,\ell) \leadsto (r,\ell) \leadsto (r,r)$.

        \item[Sub-additivity:] Let $\A_1$ and $\A_2$ be optimal self-alignments
            for $\selfed(X\fragmentco{0}{m})$ and $\selfed(X\fragmentco{m}{|X|})$, respectively.
            The concatenation of these alignments, interpreted as paths in $G \coloneqq \AG(X, X)$,
            is a valid self-alignment for $X$. Its cost does not exceed  $\selfed(X\fragmentco{0}{m}) + \selfed(X\fragmentco{m}{|X|})$.

        \item[Triangle inequality:] Let $\A : X \onto X$ denote
            an optimal self-alignment for $X$
            and let $\B : X \onto Y$ denote an optimal alignment.
            We build an self-alignment $\C : Y \onto Y$ by composing
            as a composition $\C = \B \circ \A \circ \B^{-1}$.
            By \cref{fct:triangle}, $\ed_\C(Y,Y) \le \ed_{\B^{-1}}(Y,X)+\ed_{\A}(X,X)+\ed_{\B}(X,Y)
            = \selfed(X) + 2\ed(X,Y)$. It remains to prove that $\C$ is a valid self-alignment,
            that is, $\C$ does not align $Y\position{y}$ itself for any $y\in \fragment{0}{|Y|}$.
            If $\C$ aligns $Y\position{y}$ to $Y\position{y'}$, then, by the characterization of \cref{fct:triangle},
            there exist $x,x'\in \fragmentco{0}{|X|}$ such that $\B^{-1}$ aligns $Y\position{y}$ to $X\position{x}$,
            $\A$ aligns $X\position{x}$ to $X\position{x'}$, and $\B$ aligns $X\position{x'}$ to $Y\position{y'}$.
            If $y=y'$, then $\B$ aligns both $X\position{x}$ and $X\position{x'}$ to $Y\position{y}=Y\position{y'}$,
            which implies $x=x'$. However, $\A$ is a self-alignment, so it cannot align $X\position{x}$ to itself.
            Thus, $\C$ is a self-alignment for $Y$, which means that $\selfed(Y)\le \ed_\C(Y,Y) \le \selfed(X) + 2\ed(X,Y)$.
            \qedhere
    \end{description}
\end{proof}

\newcommand{\iti}{\tilde\imath}
\newcommand{\jti}{\tilde\jmath}

Next, we discuss a useful variation of the triangle inequality. If there are two
disjoint alignments from a string \(X\) to another string \(Y\), then for every (internal)
position of \(Y\) one of the alignment aligns always an earlier position of \(X\) to this
position compared to the other alignment. In particular, we can construct a self-alignment
of \(X\) whose cost is bounded by the sum of the costs of the original alignments.

\begin{lemma}\label{lem:compose-disj-alignments}
    Consider two strings $X, Y \in \Sigma^*$ and two alignments $\A, \A' \in \Als(X, Y)$.
    If the alignments $\A, \A'$, interpreted as paths in $\AG(X, Y)$,
    do not contain any common diagonal edge, then $\selfed(X) \leq \ed_\A(X, Y) + \ed_{\A'}(X, Y)$.
\end{lemma}
\begin{proof}
    The proof is similar to the triangle inequality of~\cref{fct:selfed-properties}.
    Consider a composition alignment $\B = \A^{-1}\circ \A' \in \Als(X,X)$.
    Its cost does not exceed $\ed_{\A'}(X, Y)+\ed_{\A^{-1}}(Y,X) =\ed_\A(X, Y) + \ed_{\A'}(X, Y)$,
    so it suffices to prove that $\B$ is a valid self-alignment,
    that is, $\B$ does not align $X\position{x}$ itself for any $x\in \fragment{0}{|X|}$.
    If $\B$ aligns $X\position{x}$ to $X\position{x'}$, then, by the characterization of \cref{fct:triangle},
    there exists $y\in \fragmentco{0}{|Y|}$ such that $\A'$ aligns $X\position{x}$ to $Y\position{y}$
    and $\A^{-1}$ aligns $Y\position{y}$ to $X\position{x'}$.
    If $x=x'$, then both $\A$ and $\A'$ would contain an edge $(x,y)\to (x+1,y+1)$, contradicting their edge-disjointness.
    Thus, $\B$ is a self-alignment for $X$, which means that $\selfed(X)\le \ed_\B(X,X) \le \ed_\A(X, Y) + \ed_{\A'}(X, Y)$.
\end{proof}

As it turns out, we may strengthen \cref{lem:compose-disj-alignments} to hold (slightly
adapted) also for alignments mapping a string $X$ onto different fragments of another string $Y$.

\begin{lemma}\label{prop:disjoint-alignments-bound-sed}
    For two strings $X, Y \in \Sigma^*$ and positions
    $i \le j, i' \le j' \in \fragment{0}{|Y|}$,
    write $\A\in \Als(X,Y\fragmentco{i}{j})$ and $\A'\in \Als(X,Y\fragmentco{i'}{j'})$.
    If the alignments $\A$ and $\A'$, interpreted as paths in $\AG(X, Y)$, do not contain any common diagonal edge,
    then \[
        \selfed(X)\le
        |i-i'|+\ed_\A(X,Y\fragmentco{i}{j})
        +\ed_{\A'}(X,Y\fragmentco{i'}{j'})+|j-j'|.
    \]
\end{lemma}
\begin{proof}
    Let $i'' = \min(i,i')$ and $j''=\max(j,j')$.
    Consider alignments $\B,\B'\in \Als(X,Y\fragmentco{i''}{j''})$
    obtained from $\A,\A'$ by inserting the missing characters of $Y$,
    that is, $Y\fragmentco{i''}{i}$ and $Y\fragmentco{j}{j''}$ for $\B$,
    and $Y\fragmentco{i''}{i'}$ and $Y\fragmentco{j'}{j''}$ for $\B'$.
    Observe that the alignments $\B,B'$ still do not contain any common diagonal edge (compared to $\A,\A'$, we added horizontal and vertical edges only).
    Thus, \cref{lem:compose-disj-alignments} implies
    \begin{align*}
        \selfed(X) &\le \ed_{\B}(X, Y\fragmentco{i''}{j''})+\ed_{\B'}(X, Y\fragmentco{i''}{j''})\\
        &= (i-i'') + \ed_{\A}(X, Y\fragmentco{i}{j})+(j''-j)+(i'-i'')+\ed_{\A'}(X, Y\fragmentco{i'}{j'})+(j''-j')\\
        &= |i-i'|+\ed_\A(X,Y\fragmentco{i}{j}) +\ed_{\A'}(X,Y\fragmentco{i'}{j'})+|j-j'|.\qedhere
    \end{align*}
\end{proof}

We conclude this subsection by discussing how to efficiently compute self-alignments in
the \modelname{} model.

\begin{lemma}\label{lem:selfed}
    There is an $\Oh(k^2)$-time \modelname algorithm that, given a string $X\in \Sigma^n$
    and an integer $k\in \Int_{\ge 0}$ determines whether $\selfed(X)\le k$ and, if so,
    retrieves (the breakpoint representation of) an optimal self-alignment $\A : X \onto X$.
\end{lemma}
\begin{proof}
    We use a simple modification of the classic Landau-Vishkin algorithm~\cite{LandauV88},
    which given strings $X$ and $Y$ tests whether $\ed(X, Y) \leq k$ in time $\Oh(n+k^2)$.
    We first describe this algorithm in the standard setting (i.e. not the \modelname{} implementation),
    then describe the modifications to make it work for self-edit distance,
    and finally discuss the \modelname{} implementation.

    The Landau-Vishkin algorithm uses dynamic programming. Phrased in terms of the alignment graph $\AG(X, Y)$,
    it maintains for each edit distance value $j \in \fragment{0}{k}$ and each diagonal
    $i \in \fragment{-k}{k}$ an entry $D[i,j]$ which stores
    \[
        D[i,j] = \max\{x \mid \text{there exists a } (0,0) \leadsto (x, x+i) \text{ path } P \text{ of cost at most } j\}.
    \]
    Intuitively, the entry $D[i,j]$ stores the furthest vertex on the $i$-th diagonal that can be reached from $(0,0)$
    with cost at most $j$.
    After preprocessing the strings in linear time, each entry $D[i, j]$ can be computed in constant time using \lceOpName{} queries.
    More precisely, suppose we computed the entries $D[i,j-1]$ for some fixed value $j$ and all $i \in \fragment{-k}{k}$.
    Then, we look at the path attaining $D[i,j-1]$ for each $i$ and try to extend it by taking one horizontal/vertical/diagonal edge
    (which has cost one, due to the maximality in the choice of $P$), and then doing an \lceOpName{} query to detect how far it can
    reach on each corresponding diagonal without extra cost. After this, we can determine the entries $D[i,j]$ for each $i \in \fragment{-k}{k}$
    by taking the furthest vertex that was reached on each diagonal.

    For self-edit distance, we use the observation that self-alignments correspond to $(0,0) \leadsto (|X|, |X|)$-paths in
    the alignment graph $\AG(X, X)$ with edges on the main diagonal removed.
    In particular, to compute an alignment of minimum cost, we only need to apply one modification to the Landau-Vishkin algorithm:
    we restrict the paths so that they avoid edges on the main diagonal.
    It can be seen that this change incurs no extra overhead (we just skip \lceOpName{} queries on the main diagonal), so we can compute an optimal self-alignment of cost at most $k$ in time $\Oh(n + k^2)$.

    \cite[Lemma 6.1]{Charalampopoulos20} provides an implementation of the Landau-Vishkin algorithm in the \modelname{} model.
    This algorithm determines if the edit distance of two strings is at most $k$ and, if so, returns the breakpoint representation of
    an optimal alignment witnessing that.
    The implementation of the modifications outlined above for the self-edit distance is a straightforward adaptation of their result.
\end{proof}

\subsection{Exploiting Self-Alignments: Faster Bounded Weighted Edit Distance for Strings with Small Self-Edit Distance}\label{sec:alg:sec:alg-periodic}

In this subsection, we employ the structural insights gained for self-alignments and the
self-edit distance of a string to obtain the following algorithmic result, which is a
major step toward proving~\cref{lem:wed}.

\begin{restatable*}{lemma}{lemalgsmallsed}\label{lem:alg-periodic}
    Given two strings $X,Y\in \Sigma^{\le n}$ and integers
    $1\le d \le k\le n$ such that $\selfed(X)\le k$ and $\big||X|-|Y|\big|\le
    2d$, as well as (oracle access to) a normalized weight function
    $w : \Sigma^2 \to \mathbb{R}_{\ge 0}$,
    there is an algorithm that computes the following values,
    each associated with the breakpoint representation of the underlying alignment:
    \begin{itemize}
        \item $\ed^w_{\le d}(X,Y)$,
        \item $\min_{0\le p\le |Y|} \ed^w_{\le d}(X,Y\fragmentco{p}{|Y|})$,
        \item $\min_{0 \le q \le |Y|} \ed^w_{\le d}(X,Y\fragmentco{0}{q})$,
        \item $\min_{0\le p \le q \le |Y|}\ed^w_{\le d}(X,Y\fragmentco{p}{q})$.
    \end{itemize}
    The algorithm admits an $\Oh(k^2 d\log n)$-time \modelname implementation
    and an $\Oh(n+k^2 +k\sqrt{nd\log n})$-time implementation in the standard setting.
\end{restatable*}

For technical reasons, the specific implementation of~\cref{lem:alg-periodic} for
the \modelname{} model and for the standard setting differ.
Thus, we prove the part
corresponding to the \modelname{} model in~\cref{sec:alg:sec:alg-periodic:sec:pillar} and
for the standard setting
in~\cref{sec:alg:sec:alg-periodic:sec:standard}

\subsubsection{A \modelname Algorithm Exploiting Self-Alignments}\label{sec:alg:sec:alg-periodic:sec:pillar}

A string $X$ and a non-decreasing sequence $0=x_0 \le x_1 \le \cdots \le x_{m}=|X|$
yield a \emph{decomposition} of $X$ into $m$ \emph{phrases}; we write
$X=X\fragmentco{x_0}{x_1}\cdot X\fragmentco{x_1}{x_2}\cdots
X\fragmentco{x_{m-1}}{x_m}=\bigodot_{i=0}^{m-1} X\fragmentco{x_i}{x_{i+1}}$.

In a first step, we show how to exploit a small self-edit distance of a string \(X\)
to (algorithmically) obtain a \emph{useful} decomposition of \(X\).

\begin{lemma}\label{lem:decomp_pillar}
    Given a string $X\in \Sigma^{n}$ and an $k\in \fragment{1}{n}$ such that $\selfed(X)\le k$,
    there is an $\Oh(k^2)$-time \modelname algorithm that
    builds a decomposition $X=\bigodot_{i=0}^{m-1} X\fragmentco{x_i}{x_{i+1}}$ such that:
    \begin{itemize}
        \item Each phrase $i\in \fragmentco{0}{m}$ has a length of
            $x_{i+1}-x_i\in \fragmentco{k}{2k}$.
        \item There is a set $F\subseteq \fragmentco{0}{m}$ of at most $|F|\le 2k$
            \emph{fresh phrases} such that
            $X\fragmentco{x_i}{x_{i+1}}=X\fragmentco{x_{i-1}}{x_{i}}$ holds for each $i\in
            \fragmentco{0}{m}\setminus F$.
    \end{itemize}
    The algorithm returns the set $F$ and the indices $x_{i},x_{i+1}$ for each $i\in F$.
\end{lemma}
\begin{proof}
    As a first step, we construct the breakpoint representation of a self-alignment $\A :
    X \onto X$ of cost at most $k$; by \cref{lem:selfed}, this takes $\Oh(k^2)$ time.
    Without loss of generality, we may assume that $\A$ only contains points on or above the main diagonal
    (the parts below the main diagonal can be mirrored along the main diagonal without affecting the cost of $\A$).

    Now, suppose that we have already partitioned a prefix $X\fragmentco{0}{x_i}$.
    We start with $i \coloneqq 0$ and $x_i \coloneqq 0$,
    and proceed as follows as long as $x_i < n$.

    The algorithm maintains a set $F\sub \fragmentco{0}{i}$ (satisfying the conditions above),
    the positions $x_{i-1}$ and $x_i$, and a position $x'_i$ such that
    $\A(X\fragmentco{0}{x_i})=X\fragmentco{0}{x'_i}$.
    \begin{enumerate}
        \item\label{pt:endone} If $i > 0$ and $n-x_{i-1}<2k$, we update $x_i \coloneqq n$
            and insert $i-1$ to $F$. In other words, we extend phrase $i-1$ to the end of $X$
            and classify it as a fresh phrase.
        \item\label{pt:endtwo} If $n-x_{i}<2k$, we set $x_{i+1} \coloneqq n$,
            update $x_{i} \coloneqq \floor{({x_{i-1}+x_{i+1}})/{2}}$,
            and insert both $i-1$ and $i$ to $F$.
            In other words, we append a new phrase, balance its length with the length of
            phrase $i-1$, and classify both as fresh phrases.
        \item\label{pt:long} If $\A$ does not match $X\fragmentco{x_{i}}{x_i+2k-1}$
            perfectly, we set $x_{i+1} \coloneqq x_i+2k-1$ and insert $i$ to $F$. In other words,
            we append a new phrase of maximum length and classify it as fresh.
        \item\label{pt:power} Otherwise, $\A$ matches $X\fragmentco{x_{i}}{x_i+2k-1}$
            perfectly to $X\fragmentco{x'_{i}}{x'_i+2k-1}$.
            We set $p \coloneqq (x_i-x'_i)\cdot \ceil{{k}/({x_i-x'_i})}$
            and determine the maximum
            integer $r$ such that
            $X\fragmentco{x_i}{x_i+pr}\aonto{\A} X\fragmentco{x'_i}{x'_{i}+pr}$ and
            $X\fragmentco{x_i}{x_i+pr} = X\fragmentco{x'_i}{x'_{i}+pr}$.
            We set $x_{i+q} \coloneqq x_i+pq$ for $q\in
            \fragment{1}{r}$ and insert $i$ to $F$.
            In other words, we append $r$ new phrases of length $p$ and classify phrase
            $i$ as a fresh phrase.
    \end{enumerate}

    The following claim proves the correctness of the algorithm.
    \begin{claim}\label{clm:lengths_pillar}
        At all times, the decomposition of $X\fragmentco{0}{x_i}$ consists of phrases of
        length in $\fragmentco{k}{2k}$.
        Moreover, $X\fragmentco{x_i}{x_{i+1}}=X\fragmentco{x_{i-1}}{x_{i}}$ holds for each
        $i\in \fragmentco{0}{m}\setminus F$.
    \end{claim}
    \begin{claimproof}
        We proceed by induction so that it suffices to argue about the phrases that the
        algorithm alters in a single step.
        \begin{enumerate}
            \item In Case~\ref{pt:endone}, the algorithm extends phrase $i-1$, which is
                already of length at least $k$. Moreover, the condition $n-x_{i-1}<2k$
                guarantees that its length remains at most $2k-1$.
            \item In Case~\ref{pt:endtwo}, we have $n-x_{i-1}\ge 2k$ (otherwise, we would
                be in Case~\ref{it:endone}),
                so the total length of phrases $i-1$ and $i$ is at least $2k$.
                Moreover, since the original length of phrase $i-1$ was at most $2k-1$ and
                $n-x_i \le 2k-1$, the total length of phrases $i-1$ and $i$ is at most $4k-2$.
                As we balance the two phrases, their lengths are between $k$ and $2k-1$.
            \item In Case~\ref{pt:long}, phrase $i$ is (by construction) of length $2k-1$.
            \item In Case~\ref{pt:power}, since $x_i-x'_i\in \fragment{1}{k}$, phrases
                $i,i+1,\ldots,i+r-1$ are of length $p\in \fragmentco{k}{2k}$.
                Moreover, since $X\fragmentco{x_i}{x_i+pr}= X\fragmentco{x'_i}{x'_i+pr}$
                and $p$ is a multiple of $x_i-x'_i >0$, the phrase length $p$ is a period
                of $X\fragmentco{x_i}{x_i+pr}$. Hence, each of the phrases
                $i+1,\ldots,i+r-1$ matches its predecessor.
        \end{enumerate}
        In total, this completes the proof of the claim.
    \end{claimproof}

    We say that phrase $i$ is \emph{perfect} if $\A(X\fragmentco{x_i}{x_{i+1}})$ matches
    $X\fragmentco{x_i}{x_{i+1}}$ perfectly; otherwise, phrase $i$ is imperfect.
    Observe that the number of imperfect phrases is at most $k$. Thus, the following claim
    yields a bound of $2k$ on $|F|$.
    \begin{claim}\label{clm:fresh_pillar}
        If $i\in F$, then one of the phrases $i$ or $i+1$ is imperfect.
    \end{claim}
    \begin{claimproof}
        As in the proof of \cref{clm:lengths}, we analyze the phrases that the algorithm
        alters in a single step.
        \begin{enumerate}
            \item In Case~\ref{it:endone}, phrase $i-1$ is a suffix of $X$; thus, it is
                imperfect because $\A$ inserts the last character of $X$.
            \item In Case~\ref{it:endtwo}, phrase $i$ is a suffix of $X$; thus, it is
                imperfect because $\A$ inserts the last character of $X$. Moreover, phrase
                $i-1$ has phrase $i$ as an imperfect neighbor.
            \item In Case~\ref{it:long}, phrase $i$ is imperfect because $\A$ does not
                match $X\fragmentco{x_i}{x_i+2k-1}$ perfectly to
                $X\fragmentco{x'_i}{x'_i+2k-1}$.
            \item In Case~\ref{it:power}, there is nothing to prove unless phrases $i-1$
                and $i$ are both perfect.
                We then have
                $X\fragmentco{x_{i-1}}{x_i+2k-1}\aonto{\A} X\fragmentco{x'_{i-1}}{x'_i+2k-1}$
                and $X\fragmentco{x_{i-1}}{x_i+2k-1} = X\fragmentco{x'_{i-1}}{x'_i+2k-1}$.
                Thus, phrase $i-1$ must have been created in Case~\ref{it:power} (with
                $x'_{i-1}-x_{i-1}=x'_{i}-x_{i}$).
                In particular, the previous iteration of the algorithm shared the same phrase
                length $p$, so
                $X\fragmentco{x_i}{x_i+p}\aonto{\A} X\fragmentco{x'_i}{x'_i+p}$
                and
                $X\fragmentco{x_i}{x_i+p} = X\fragmentco{x'_i}{x'_i+p}$
                contradicts the maximality of $r$ in that
                iteration.
        \end{enumerate}
        In total, this completes the proof of the claim.
    \end{claimproof}

    Now we analyze the running time of the algorithm.

    The time to construct $\A$ is $\Oh(k^2)$ by~\cref{lem:selfed}.
    Next, we analyze the time of each of the Cases~1~to~4.
    Note that Cases~1 and 2 can be easily done in time $\Oh(1)$.
    For Cases~3 and 4, we can check whether $\A$ matches $X\fragmentco{x_i}{x_i + 2k-1}$
    perfectly using the breakpoint representation of the alignment in time $\Oh(1)$.
    Moreover, for Case~4 we can determine the value $r$ using one call to
    $\lceOp{X\fragmentco{x_i}{|X|}}{X\fragmentco{x_{i'}}{|X|}}$.

    Finally, observe that during each case we add one fresh phrase.
    As observed earlier, \cref{clm:fresh_pillar} implies that $|F| \leq 2k$.
    Hence, we spend $\Oh(k)$ to build the decomposition, and therefore the overall running
    time is bounded by $\Oh(k^2)$.
\end{proof}

\lemalgsmallsed
\begin{proof}[Proof for the \modelname model.]
    Suppose that we have $\wed(X,Y\fragmentco{p}{q})\le d$ for some $0\le p \le q \le
    |Y|$. As \(w\) is normalized, we also have $\ed(X, Y\fragmentco{p}{q})\le d$
    and thus \[q-p=|Y\fragmentco{p}{q}|\ge |X|-d \ge |Y|-3d.\]

    Now, since we can transform $X$ into $Y$ by first inserting $Y\fragmentco{0}{p}$ and
    $Y\fragmentco{q}{|Y|}$, and then
    transforming $X$ into $Y\fragment{p}{q}$, it follows that \[
        \ed(X,Y) \le p+\ed(X,Y\fragmentco{p}{q}) + |Y|-q \le 4d.
    \]
    Thus, we can assume without loss of generality that $\ed(X,Y)\le 4d$; otherwise, it
    suffices to return $\infty$ for all queries.

    Next, we say that a vertex $(x,y)$ of the alignment graph $\AG^w(X,Y)$ is
    \emph{relevant} if $y-x\in \fragment{-d}{3d}$.
    Moreover, we define $G$ as the subgraph of $\AG^w(X,Y)$ induced by the relevant vertices.
    The following claim shows that we can work with $G$ instead of the entire alignment graph.
    \begin{claim}\label{clm:relevant_pillar}
        If $\wed(X,Y\fragmentco{p}{q})\le d$ holds for some fragment $Y\fragmentco{p}{q}$,
        then $\dist_{G}((0,p),(|X|,q))=\wed(X,Y\fragmentco{p}{q})$
        and the optimal alignments correspond to the shortest $(0,p)\leadsto (|X|,q)$ paths in $G$.
    \end{claim}
    \begin{claimproof}
        Observe that $\wed(X\fragmentco{0}{x},Y\fragmentco{p}{y})\le d$ implies
        $(x-0)-(y-p)\le d$, that is, we have \(y-x\ge -d.\)
        Moreover, $\wed(X\fragmentco{x}{|X|},Y\fragmentco{y}{q})\le d$ implies
        $(|X|-x)-(q-y)\le d$, that is, we have \[y-x \le q-|X|+d \le (|Y|-|X|)+d \le 3d.\]
        Finally, since $\wed(X,Y\fragmentco{p}{q})\le d$, we obtain that
        all pairs $(x,y)$ of the underlying alignment satisfy \(
            y-x\in \fragment{-d}{3d}.
        \)
        Hence, this alignment is a $(0,p)\leadsto (|X|,q)$ path in $G$. Conversely, every
        $(0,p)\leadsto (|X|,q)$ path in $G$ yields an alignment $\A\in
        \Als(X,Y\fragmentco{p}{q})$.
    \end{claimproof}

    We run the algorithm of \cref{lem:decomp_pillar} for $X$, arriving at a decomposition
    $X=\bigodot_{i=0}^{m} X\fragmentco{x_{i}}{x_{i+1}}$ with fresh phrases $F$.
    For each phrase $X\fragmentco{x_i}{x_{i+1}}$ consider the subgraph $G_i$ of $G$
    induced by $\fragment{x_i}{x_{i+1}}\times \fragment{0}{|Y|}$.

    Now, recall that $G$ contains only relevant vertices $(x, y)$, which by definition satisfy
    $y-x \in \fragment{-d}{3d}$.
    Thus, observe that $G_i$ is induced by \[
        \fragment{x_i}{x_{i+1}} \times \fragment{\max\{0, x_i - d\}}{\min\{|Y|, x_{i+1} + 3d\}}.
    \]
    The following claim shows that there are only $\Oh(k)$ graphs $G_i$ which are not
    isomorphic to $G_{i-1}$.
    \begin{claim}
        There is a set $F'\sub \fragmentco{0}{m}$ of size $\Oh(k)$ such that, for each
        $i\in \fragmentco{0}{m}\sm F'$,
        the graph $G_i$ is isomorphic to $G_{i-1}$. Moreover, such a set can be computed
        in $\Oh(k)$ time,
        along with the positions $x_i,x_{i+1}$ for each $i\in F'$.
    \end{claim}
    \begin{claimproof}
        We define $F'$ so that $i\in \fragmentoo{0}{m}\sm F'$ holds if and only if
        \begin{itemize}
            \item $d \le x_{i-1} \le x_{i+1}\le |Y|-3d$,
            \item $X\fragmentco{x_{i-1}}{x_i}=X\fragmentco{x_{i}}{x_{i+1}}$, and
            \item $Y\fragmentco{x_{i-1}-d}{x_i+3d} = Y\fragmentco{x_i-d}{x_{i+1}+3d}$.
        \end{itemize}
        Observe that this implies the desired property, that is,
        if $i \in \fragmentco{0}{m} \sm F'$, then $G_i$ is isomorphic to $G_{i-1}$.

        Let us first prove that $|F'|\le 100k$. Fix an optimal unweighted alignment $\A :
        X\onto Y$; as argued earlier \(\A\) has a cost of at most $4d$.

        Consider $i\in \fragment{7}{m-8}$ such that
        \[
            X\fragmentco{x_{i-7}}{x_{i-6}}
            =X\fragmentco{x_{i-5}}{x_{i-4}} = \cdots = X\fragmentco{x_{i+7}}{x_{i+8}},
        \]
        and $\A$ matches $X\fragmentco{x_{i-6}}{x_{i+8}}$
        perfectly to $Y\fragmentco{y_{i-6}}{y_{i+8}}$.
        Then, $Y\fragmentco{y_{i-6}}{y_{i+8}}=X\fragmentco{x_{i-7}}{x_{i+7}}$
        has period $p \coloneqq x_i-x_{i-1}$.

        Since $\wed(X,Y)\le 4d \le 4p$, we have
        \[y_{i-6} \le x_{i-6}+4p = x_{i-2} \le x_{i-1}-d
            \quad\text{and}\quad y_{i+8} \ge x_{i+8}-4p = x_{i+4}\ge x_{i+1}+3d.
        \]
        Hence, $Y\fragmentco{x_{i-1}-d}{x_{i+1}+3d}$ has period $p=x_i-x_{i-1}$.
        Consequently, we have $i\in F'$.

        Now, there are at most $12$ indices $i\in \fragmentco{0}{m}\sm \fragment{7}{m-8}$,
        at most $15|F|$ indices violating
        $X\fragmentco{x_{i-7}}{x_{i-6}}=\cdots=X\fragmentco{x_{i+7}}{x_{i+8}}$,
        and at most $14\cdot \ed(X,Y)$ indices for which $\A$ does not match
        $X\fragmentco{x_{i-6}}{x_{i+8}}$ perfectly.
        Overall, since $|F| \leq 2k$ and $\ed(X, Y) \leq 4d$ we conclude that $|F'|\le
        12+15\cdot 2k+14\cdot 4d \le 100k$.

        Next, we provide an algorithm constructing $F'$ in time $\Oh(|F|+|F'|)=\Oh(k)$.
        First, we post-process $F$ to make sure that it does not contain indices $i$ such
        that $X\fragmentco{x_{i-1}}{x_i}=X\fragmentco{x_{i}}{x_{i+1}}$.
        For this, we scan $F$ from left to right and, for any two subsequent entries
        $i,i'\in F$, we remove $i'$ whenever
        $X\fragmentco{x_{i'}}{x_{i'+1}}=X\fragmentco{x_{i}}{x_{i+1}}$. The latter
        condition can be tested in $\Oh(1)$ time using an \lceOpName query.

        Suppose that we have already computed $F'\cap \fragmentco{0}{j}$ for some $j\in F'$.
        If $j\in F$, then $j\in F'$ (due to the post-processing above), so we add $j$ to
        $F'$ and proceed to $j \coloneqq j+1$.

        Thus, suppose that $i<j<i'$, where $i,i'$ are two subsequent entries of
        $F\cup\{m\}$. Observe  that $x_i$ and $x_{i+1}$ let us retrieve $p=x_{i+1}-x_i$
        as well as $x_{j'}= x_i + p\cdot(j'-i)$ for $j'\in \fragment{i}{i'}$.

        We check if \[
            d \le x_{j-1} < x_{j+1} \le |Y|-3d
        \] and use an \lceOpName query to check if \[
        Y\fragmentco{x_{i-1}-d}{x_{i-1}+3d} = Y\fragmentco{x_i-d}{x_{i}+3d}.
        \]
        If either condition fails, then $j\in F'$, so we add $j$ to $F'$ and proceed to
        $j \coloneqq j+1$.

        Otherwise, we use another \lceOpName query to determine the length $\ell$ of the
        longest common prefix of $Y\fragmentco{x_{i-1}+3d}{|Y|}$ and
        $Y\fragmentco{x_{i}+3d}{|Y|}$ and set $r \coloneqq \min(i'-j,\floor{{\ell}/{d}})$.
        Observe that we have \[
            \fragment{j}{j+r}\cap F' = \{j+r\}.
        \] Thus, we add $j+r$ to $F'$ and
        proceed to $j\coloneqq j+r+1$.

        The algorithm performs $\Oh(|F'|)=\Oh(k)$ iterations, spending $\Oh(1)$ for each of them.
    \end{claimproof}

    For each $i\in \fragment{0}{m}$, denote by $V_i$ the vertices on the column $x_i$ of
    $G$, that is, $G = \{(x,y)\in V(G) \mid x=x_i\}$.
    Moreover, for $i,j\in \fragment{0}{m}$ with $i\le j$, let $D_{i,j}$ denote
    the matrix of pairwise distances from $V_i$ to $V_j$.
    The matrix \(D\) is a Monge matrix by \cref{fct:monge}.

    We now describe our algorithm to compute $D_{0,j}$ for all $j\in \bar{F}$, where
    $\bar{F}=F'\cup \{m\}$.
    Suppose that we have already computed $D_{0,i}$ for some $i\in F'$ and let $j$ denote
    the next index in $\bar{F}$.
    We preprocess the graph $G_i$ using \cref{thm:klein} and retrieve the matrix $D_{i,i+1}$.
    Then, we compute $D_{i,j}$ as the $(j-i)$-th power of $D_{i,i+1}$ with respect to the
    min-plus product.
    To this end, we iterate from $t \coloneqq \ceil{\log(j-i)}$ down to $t \coloneqq 0$,
    computing the $\floor{({j-i})/{2^t}}$-th and {$\ceil{({j-i})/{2^t}}$-th} min-plus
    powers of $D_{i,i+1}$. Finally, we compute $D_{0,j}$ as the min-plus product of $D_{0,i}$ and
    $D_{i,j}$.

    Observe that after computing $D_{0,m}$ we are essentially done:
    By \cref{clm:relevant_pillar}, for every $p,q$, the distance $\wed(X,Y\fragmentco{p}{q})\le d$
    corresponds to the entries of $D_{0,m}$ not exceeding $d$.
    These are precisely the values promised by the lemma statement.

    Next, we analyze the running time.
    Each application of \cref{thm:klein} costs $\Oh(|V(G_i)|\log |V(G_i)|)=\Oh(kd\log k)$
    time since $|V(G_i)|\le (3d+1)\cdot (2k+1)$.
    To retrieve the matrix $D_{i,i+1}$ we use \cref{thm:klein} for each of the
    $|V_i|\cdot|V_i| = \Oh(d^2)$ pairs of vertices, taking total time $\Oh(d^2 \log k)$.
    The min-plus products of $\Oh(d)\times \Oh(d)$ Monge matrices can be computed in
    $\Oh(d^2)$ each using the SMAWK algorithm (see~\cref{thm:smawk}).
    Hence, the construction of $D_{0,m}$ costs $\Oh(|F'|\cdot (kd\log k + d^2 \log
    n))=\Oh(k^2 d \log n)$ time in total.

    Finally, we need to show how to recover the shortest paths $v_0\leadsto v_m$ path $P$
    for some $v_0\in V_0$ and $v_m\in V_m$.
    In the first step, we recover, for each $j\in \bar{F}$, a vertex $v_j\in P\cap V_j$.
    Let $i$ denote the predecessor of $j$ in $\bar{F}$. We iterate over $v_i\in V_i$ and
    use the matrices $D_{0,i}$ and $D_{i,j}$
    to find a vertex which minimizes $\dist_{G}(v_0,v_i)+\dist_G(v_i,v_j)$.

    Next, for any two subsequent indices in $i,j\in \bar{F}$, we recover the shortest path
    from $v_i$ to $v_j$.
    For this, we use a recursive procedure based on the powers of $D_{i,i+1}$ computed
    while the $(j-i)$th power.
    At depth $t$ of the recursion, we have an interval $\fragmentco{i'}{j'}\sub
    \fragmentco{i}{j}$ of length $\floor{({j-i})/{2^t}}$ or $\ceil{({j-i})/{2^t}}$ and
    vertices $v_{i'}\in P\cap V_{i'}, v_{j'}\in P\cap V_{j'}$.

    If $\dist_G(v_{i'},v_{j'})=0$, the path simply follows the diagonal without any breakpoints.

    If $j'=i'+1$, we use the data structure of \cref{thm:klein} to retrieve the shortest
    path explicitly.

    In the remaining case, we consider an intermediate value
    $k'\coloneqq \floor{({i'+j'})/{2}}$, find a vertex $v_{k'}\in V_{k'}$ minimizing
    $\dist_G(v_{i'},v_{k'})+\dist_G(v_{k'},v_{j'})$ (based on the precomputed powers of
    $D_{i,i+1}$), and recurse on $\fragmentco{i'}{k'}$ as well as $\fragmentco{k'}{j'}$.

    The intermediate vertices $v_j$ for $j\in \bar{F}$ are recovered in $\Oh(d)$ time
    each, for a total of $\Oh(kd)$.
    The (breakpoint representation of the) path $v_i\leadsto v_j$ (for subsequent indices
    $i,j\in \bar{F}$)
    involves $\Oh(c\log n)$ recursive calls, there $c$ is the cost of the path.
    Each call takes either $\Oh(k)$ time (if \cref{thm:klein} is used) or $\Oh(d)$ time
    (if we retrieve the midpoint only).
    Since the total cost of the paths $v_i\leadsto v_j$ is at most $d$, the overall
    running time of the recover phase is $\Oh(kd\log n)$.
\end{proof}

\subsubsection{An Algorithm exploiting Self-Alignments in the Standard Setting}
\label{sec:alg:sec:alg-periodic:sec:standard}

Similarly to the algorithm for the \modelname{} model, we first show how to exploit a
small self-edit distance of a string \(X\) to obtain a useful decomposition of \(X\).
In particular, compare the following result with \cref{lem:decomp_pillar}.

\begin{lemma}\label{lem:decomp}
    Given a string $X\in \Sigma^{n}$ and integers $\ell,k\in \fragment{1}{n}$ such that
    $\selfed(X)\le k$,
    there is an $\Oh(n+k^2)$-time algorithm that
    builds a decomposition $X=X\fragmentco{x_0}{x_1}\cdot X\fragmentco{x_1}{x_2}\cdots
    X\fragmentco{x_{m-1}}{x_m}$ such that:
    \begin{itemize}
        \item Each phrase $i\in \fragmentco{0}{m}$ has length $x_{i+1}-x_i\in
            \fragmentco{\ell}{2\ell}$.
        \item There is a set $F\subseteq \fragmentco{0}{m}$ of at most $|F|\le 3k$
            \emph{fresh phrases}
            such that every phrase $i\in \fragmentco{0}{m}\setminus F$ has a \emph{source
            phrase} $i'\in \fragmentco{0}{i}$
            with $X\fragmentco{x_i}{x_{i+1}}=X\fragmentco{x_{i'}}{x_{i'+1}}$    and $x_i-x_{i'}
            \le \max(2\ell-1,k)$.
    \end{itemize}
\end{lemma}
\begin{proof}
    If $|X|<2\ell$, it suffices to return a single-phrase decomposition. Thus, we may
    assume $n\ge 2\ell$.

    As a first step, we construct an optimal self-alignment $\A : X \onto X$ of cost at
    most $k$; by \cref{lem:selfed}, this takes $\Oh(n+k^2)$ time.
    Without loss of generality, we may assume that $\A$ only contains points on or above the main diagonal
    (the parts below the main diagonal can be mirrored along the main diagonal without affecting the cost of $\A$).
    Next, we construct the decomposition from left to right.

    Suppose we have already decomposed a prefix $X\fragmentco{0}{x_i}$ into $i$ phrases.
    We start with $i \coloneqq 0$ and $x_i \coloneqq 0$, and proceed as follows as long as
    $x_i < n$ based on $x'_i$ such that $\A(X\fragmentco{0}{x_i})=X\fragmentco{0}{x'_i}$.
    \begin{enumerate}
        \item\label{it:endone} If $i>0$ and $n-x_{i-1} < 2\ell$, we update $x_{i}
            \coloneqq n$. In other words, we extend phrase $i-1$ to the end of $X$.
        \item\label{it:endtwo} If $n-x_{i} < 2\ell$, we set $x_{i+1} \coloneqq n$ and update
            $x_i \coloneqq  \lfloor {(x_{i-1}+x_{i+1})}/{2} \rfloor$.
            In other words, we append a new phrase and balance its length with the length of
            phrase $i-1$.
        \item\label{it:long} If $\A$ does not match $X\fragmentco{x_{i}}{x_i+2k-1}$
        perfectly, we set $x_{i+1} \coloneqq x_i+2\ell-1$. In other
            words, we append a new phrase of maximum length.
        \item\label{it:power} If $\A$ matches $X\fragmentco{x_{i}}{x_i+2k-1}$
        perfectly to $X\fragmentco{x'_{i}}{x'_i+2k-1}$ with $x_i-x'_i < 2\ell$,
            then we set \[
                x_{i+1} \coloneqq x_i + (x_i-x'_i)\cdot \lceil{\ell}/({x_i-x'_i})\rceil.
            \]
            In other words, we append a new phrase of the minimum possible length that is a
            multiple of $x_i-x'_i$.
        \item\label{it:far} If $\A$ matches $X\fragmentco{x_{i}}{x_i+2k-1}$
        perfectly to $X\fragmentco{x'_{i}}{x'_i+2k-1}$ with $x_i-x'_i \ge 2\ell$,
            then we retrieve $i'\in \fragmentco{0}{i}$ such that
            $x'_i \in \fragmentco{x_{i'}}{x_{i'+1}}$.
            \begin{enumerate}
                \item\label{it:far:copy} If $x_{i'} = x'_i$, we set $x_{i+1} \coloneqq x_i +
                    x_{i'+1}-x_{i'}$. In other words, we append a new phrase which is a
                    copy phrase $i'$.
                \item\label{it:far:one} If $x_{i'+1} +x_i-x'_i - x_{i-1} < 2\ell$, we
                    update $x_i \coloneqq x_{i'+1}+x_i-x'_i$. In other words, we extend phrase
                    $i-1$.
                \item\label{it:far:two} Otherwise, we set $x_{i+1} \coloneqq x_{i'+1}+x_i-x'_i$
                    and update $x_i \coloneqq \lfloor ({x_{i-1}+x_{i+1}})/{2} \rfloor$.
                    In other words, we append a new phrase and balance its length with the
                    length of phrase $i-1$.
            \end{enumerate}
    \end{enumerate}

    As in the proof for the \modelname{} model, we start with the correctness.

    \begin{claim}\label{clm:lengths}
        At all times, the decomposition of $X\fragmentco{0}{x_i}$ consists of phrases of
        length in $\fragmentco{\ell}{2\ell}$.
    \end{claim}
    \begin{claimproof}
        We proceed by induction so that it suffices to argue about the phrases that the
        algorithm alters in a single step.
        \begin{enumerate}
            \item In Case~\ref{it:endone}, the algorithm extends phrase $i-1$, which is
                already of length at least $\ell$. Moreover, the condition
                $n-x_{i-1}<2\ell$ guarantees that its length remains at most $2\ell-1$.
            \item In Case~\ref{it:endtwo}, we have $n-x_{i-1}\ge 2\ell$ (otherwise, we
                would be in Case~\ref{it:endone}),
                so the total length of phrases $i-1$ and $i$ is at least $2\ell$.
                Moreover, since the original length of phrase $i-1$ was at most $2\ell-1$
                and $n-x_i \le 2\ell-1$, the total length of phrases $i-1$ and $i$ is at
                most $4\ell-2$.
                As we balance the two phrases, their lengths are between $\ell$ and
                $2\ell-1$.
            \item In Case~\ref{it:long}, phrase $i$ is (by construction) of length
                $2\ell-1$.
            \item In Case~\ref{it:power}, phrase $i$ is (by construction) of length
                $x_i-x'_i\cdot \lceil {\ell}/({x_i-x'_i})\rceil$, which is between
                $\ell$ and $2\ell-1$ because $x_i-x'_i \le 2\ell-1$.
            \item In Case~\ref{it:far}, we rely on $x_{i'+1}-x_{i'} \in
                \fragmentco{\ell}{2\ell}$, which follows by induction.
                This further implies $x_{i'+1}-(x_i-x_i-x'_i)<2\ell$ and $x_{i'+1}<x_i$
                (due to $x_i-x'_i \ge 2\ell$),  which means that $i'+1<i$.
                \begin{enumerate}
                    \item In Case~\ref{it:far:copy}, phrase $i$ is (by construction) of
                        length $x_{i'+1}-x_{i'} \in  \fragmentco{\ell}{2\ell}$.
                    \item In Case~\ref{it:far:one}, the algorithm extends phrase $i-1$,
                        which is already of length at least $\ell$. Moreover, the
                        condition $x_{i'+1}+x_i-x'_i - x_{i-1}<2\ell$ guarantees that its
                        length remains at most $2\ell-1$.
                    \item In Case~\ref{it:far:two}, we have $x_{i'+1}+x_i-x'_i-x_{i-1}\ge
                        2\ell$ (otherwise, we would be in Case~\ref{it:far:one}), so the
                        total length of phrases $i-1$ and $i$ is at least $2\ell$.
                        Moreover, since the original length of phrase $i-1$ was at most
                        $2\ell-1$ and $x_{i'+1}+x_i-x'_i-x_i \le 2\ell-1$, the total
                        length of phrases $i-1$ and $i$ is at most $4\ell-2$.
                        As we balance the two phrases, their lengths are between $\ell$
                        and $2\ell-1$.
                \end{enumerate}
        \end{enumerate}
        In total, this completes the proof of the claim.
    \end{claimproof}

    We say that phrase $i$ is \emph{perfect} if $\A(X\fragmentco{x_i}{x_{i+1}})$ matches
    $X\fragmentco{x_i}{x_{i+1}}$ perfectly; otherwise, phrase $i$ is imperfect.
    Observe that the number of imperfect phrases is at most $k$. Thus, the following claim
    yields a bound of $3k$ on the number of fresh phrases.
    \begin{claim}\label{clm:fresh}
        Every fresh phrase is either imperfect itself or has an imperfect neighbor.
    \end{claim}
    \begin{claimproof}
        As in the proof of \cref{clm:lengths}, we analyze the phrases that the algorithm
        alters in a single step.
        \begin{enumerate}
            \item In Case~\ref{it:endone}, phrase $i-1$ is a suffix of $X$; thus, it is
                imperfect because $\A$ inserts the last character of $X$.
            \item In Case~\ref{it:endtwo}, phrase $i$ is a suffix of $X$; thus, it is
                imperfect because $\A$ inserts the last character of $X$. Moreover, phrase
                $i-1$ has phrase $i$ as an imperfect neighbor.
            \item In Case~\ref{it:long}, phrase $i$ is imperfect because $\A$ does not
                match $X\fragmentco{x_i}{x_i+2\ell-1}$ perfectly to
                $X\fragmentco{x'_i}{x'i+2\ell-1}$.
            \item In Case~\ref{it:power}, there is nothing to prove unless phrases $i-1$
                and $i$ are both perfect.
            We then have
            $X\fragmentco{x_{i-1}}{x_i+2\ell-1} \aonto{\A} X\fragmentco{x'_{i-1}}{x'_i+2\ell-1}$
            and $X\fragmentco{x_{i-1}}{x_i+2\ell-1} = X\fragmentco{x'_{i-1}}{x'_i+2\ell-1}$
            with $x_i-x'_i=x_{i-1}-x'_{i-1}$. Thus,
            phrase $i-1$ must have been created in Case~\ref{it:power} (and have the same
            length).
            In particular, $X\fragmentco{x_{i-1}}{x_i}=X\fragmentco{x_{i}}{x_{i+1}}$
            because $X\fragmentco{x_i-2\ell+1}{x_i+2\ell-1}$ and both phrases have the
            same length. Consequently, phrase $i-1$ is the source of phrase $i$.
            \item In Case~\ref{it:far}, we rely on the fact that if the algorithm is in
                Case~\ref{it:far} twice in a row with the same shift $x_i-x'_i$, then the
                second time must be in Case~\ref{it:far:copy}.
            \begin{enumerate}
                \item In Case~\ref{it:far:copy}, phrase $i'$ is the source of phrase $i$.
                \item In Case~\ref{it:far:one}, there is nothing to prove unless phrase
                    $i-1$ is perfect.
                We then have
                $X\fragmentco{x_{i-1}}{x_i+2\ell-1}\aonto{\A}
                X\fragmentco{x'_{i-1}-}{x'_i+2\ell-1}$ and
                $X\fragmentco{x_{i-1}}{x_i+2\ell-1} =
                X\fragmentco{x'_{i-1}-}{x'_i+2\ell-1}$.
                Thus, phrase $i-1$ must have been originally created in Case~\ref{it:far}
                (with the shift). This guarantees the existence of $i'\in
                \fragmentco{0}{i}$ such that $x_{i'}=x'_i-$,
                which is a contradiction because we are not in Case~\ref{it:far:copy}.
                \item In Case~\ref{it:far:two}, there is nothing to prove unless phrases
                    $i-1$ and $i$ are perfect.
                We then have
                $X\fragmentco{x_{i-1}}{x_i+2\ell-1}\aonto{\A}
                X\fragmentco{x'_{i-1}}{x'_i+2\ell-1}$
                and $X\fragmentco{x_{i-1}}{x_i+2\ell-1} =
                X\fragmentco{x'_{i-1}}{x'_i+2\ell-1}$.
                Thus, phrase $i-1$ must have been
                originally created in Case~\ref{it:far} (with the same shift). This
                guarantees the existence of $i'\in \fragmentco{0}{i}$ such that
                $x_{i'}=x'_i$,
                which is a contradiction because we are not in Case~\ref{it:far:copy}.
            \end{enumerate}
        \end{enumerate}
        In total, this completes the proof of the claim.
    \end{claimproof}

    The running time is $\Oh(n+k^2)$ for the construction of $\A$ plus $\Oh(n)$ for the
    decomposition. In total, this completes the proof.
\end{proof}

Finally, we are ready to prove \cref{lem:alg-periodic} in the standard setting.

\begin{proof}[Proof of \cref{lem:alg-periodic} for the standard setting.]
    If $\ed(X,Y\fragmentco{p}{q})\le d$, then $\selfed(Y\fragmentco{p}{q})\le 2k+2d$ holds
    by \cref{fct:selfed-properties} (triangle inequality).
    Moreover, $q - p = |Y\fragmentco{p}{q}|\ge |X|-d \ge |Y|-3d$, and thus
    \[\selfed(Y\fragmentco{0}{p})+\selfed(Y\fragmentco{q}{|Y|}) \le 2p+2(|Y|-q)\le 6d.\]
    Consequently, using the sub-additivity of \cref{fct:selfed-properties} we have that
    \[
        \selfed(Y)\le \selfed(Y\fragmentco{0}{p})+\selfed(Y\fragmentco{p}{q})+\selfed(Y\fragmentco{q}{|Y|}) \le 2k+8d \le 10k.
    \]
    Thus, we can assume without loss of generality that $\selfed(Y)\le 10k$; otherwise, it
    suffices to return $\infty$ for all queries.

    We set $\ell \coloneqq \min(\lceil\sqrt{nd}/(k\sqrt{\log n})\rceil,d)$
    and run the algorithm of \cref{lem:decomp} for both $X$~and~$Y$,
    arriving at decompositions \[
        X=\bigodot_{i\in \fragmentco{0}{m_X}} X\fragmentco{x_{i}}{x_{i+1}}
        \quad\text{and}\quad Y=\bigodot_{j\in \fragmentco{0}{m_Y}} Y\fragmentco{y_{j}}{y_{j+1}}.
    \]
    Next, we define a \emph{box decomposition} of the alignment graph $\AG^w(X,Y)$
    as the indexed family of boxes $(B_{i,j})_{(i,j)\in \fragmentco{0}{m_X}\times
    \fragmentco{0}{m_Y}}$, where box $B_{i,j}$ is the subgraph of $\AG^w(X,Y)$
    induced by $\fragment{x_i}{x_{i+1}}\times \fragment{y_j}{y_{j+1}}$.

    Recall that in the first part of the proof of \cref{lem:alg-periodic}, we considered a
    vertex $(x,y)$ \emph{relevant} if $y-x\in \fragment{-d}{3d}$.
    Now, we say that a box $B_{i,j}$ is \emph{relevant} if it contains at least one
    relevant vertex, that is, $\fragment{y_j-x_{i+1}}{y_{j+1}-x_i}\cap
    \fragment{-d}{3d}\ne \emptyset$.
    We consider the subgraph $G$ of $\AG^w(X,Y)$ induced by all vertices contained in
    relevant boxes.
    In \cref{clm:relevant_pillar}, we considered a subgraph induced by relevant vertices
    only. Including more vertices in the induced subgraph can only decrease the distances,
    so the following claim remains valid.

    \begin{claim}\label{clm:relevant}
        If $\wed(X,Y\fragmentco{p}{q})\le d$ holds for some fragment $Y\fragmentco{p}{q}$,
        then $\dist_{G}((0,p),(|X|,q))=\wed(X,Y\fragmentco{p}{q})$
        and the optimal alignments correspond to the shortest $(0,p)\leadsto (|X|,q)$ paths in $G$.
    \end{claim}
    \begin{claimproof}
        Observe that $\wed(X\fragmentco{0}{x},Y\fragmentco{p}{y})\le d$ implies
        $(x-0)-(y-p)\le d$, that is, we have \(
            y-x\ge -d.
        \)
        Moreover, $\wed(X\fragmentco{x}{|X|},Y\fragmentco{y}{q})\le d$ implies
        $(|X|-x)-(q-y)\le d$, that is, we have\[
            y-x \le q-|X|+d \le (|Y|-|X|)+d \le 3d.
        \]
        Finally, using that $\wed(X,Y\fragmentco{p}{q})\le d$, we obtain that all pairs $(x,y)$
        of the underlying alignment satisfy $y-x\in \fragment{-d}{3d}$.
        Hence, this alignment is a $(0,p)\leadsto (|X|,q)$ path in $G$. Conversely, every
        $(0,p)\leadsto (|X|,q)$ path in $G$ yields an alignment $\A\in
        \Als(X,Y\fragmentco{p}{q})$.
    \end{claimproof}

    Our algorithm lists all relevant boxes (without constructing them explicitly)
    and partitions them into equivalence classes of isomorphic boxes;
    $B_{i,j}$ and $B_{i',j'}$ are isomorphic if and only if \[
        X\fragmentco{x_i}{x_{i+1}}=X\fragmentco{x_{i'}}{x_{i'+1}}
        \quad\text{and}\quad
        Y\fragmentco{y_j}{y_{j+1}}=Y\fragmentco{y_{j'}}{y_{j'+1}}.
    \]
    We preprocess every non-isomorphic box using the data structure of~\cref{thm:klein}
    to compute the matrix of all distances from \emph{input} vertices to \emph{output} vertices,
    where a vertex $(x,y)$ of a box $B_{i,j}$ is an input vertex if $x=x_i$ or $y=y_j$,
    and an output vertex if $x=x_{i+1}$ or $y=y_{j+1}$.

    \begin{claim}\label{clm:few}
        There are $\Oh({nd}/{\ell^2})$ relevant boxes, out of which
        $\Oh({k^2}/{\ell})$ can be non-isomorphic.
    \end{claim}
    \begin{claimproof}
        As for the number of relevant boxes, consider an $i\in \fragmentco{0}{m_X}$. The
        box $B_{i,j}$ is relevant
        if and only if $\fragment{y_{j}}{y_{j+1}}\cap \fragment{x_i-d}{x_{i+1}+3d}\ne \emptyset$.

        Since each phrase length is in $\fragmentco{\ell}{2\ell}$, the interval
        $\fragment{x_i-d}{x_{i+1}+3d}$
        intersects at most \[
            2+{|\fragment{x_i-d}{x_{i+1}+3d}|}/{\ell}
            \le 2 +{(2\ell+4d)}/{\ell}
            =\Oh({d}/{\ell}) \quad\text{phrases.}
        \]
        Due to $m_X \le {n}/{\ell}$,
        this means that the total number of relevant boxes is $\Oh({nd}/{\ell^2})$.

        We say that a box $B_{i,j}$ is
        \begin{itemize}
            \item \emph{semi-relevant} if
                $\fragment{y_j-x_{i+1}}{y_{j+1}-x_i}\cap \fragment{-2k}{10k}\ne \emptyset$ and
            \item \emph{fresh} if it is semi-relevant and $X\fragmentco{x_i}{x_{i+1}}$ or
        $Y\fragmentco{y_j}{y_{j+1}}$ is a fresh phrase.
        \end{itemize}
        A box $B_{i,j}$ is semi-relevant if and only if \[
        \fragment{y_{j}}{y_{j+1}}\cap \fragment{x_i-2k}{x_{i+1}+10k}\ne \emptyset.\]
        Since each phrase length is in $\fragmentco{\ell}{2\ell}$, the interval
        $\fragment{x_i-2k}{x_{i+1}+10k}$ intersects
        at most \[
            2+({2\ell+11k})/{\ell}=\Oh({k}/{\ell}) \quad\text{phrases.}
        \]
        Thus, every fresh phrase $X\fragmentco{x_i}{x_{i+1}}$ gives rise to
        $\Oh({k}/{\ell})$ fresh boxes.
        A symmetric argument shows that every fresh phrase $Y\fragmentco{y_j}{y_{j+1}}$
        gives rises to $\Oh({k}/{\ell})$ fresh boxes,
        for a total of $\Oh({k^2}/{\ell})$ fresh boxes.

        We claim that every semi-relevant box is isomorphic to some fresh box. For this,
        we proceed by induction on $i+j$.
        If $X\fragmentco{x_i}{x_{i+1}}$ or $Y\fragmentco{y_j}{y_{j+1}}$ is a fresh phrase,
        then the box $B_{i,j}$ is fresh by definition.

        Otherwise, there are $i'\in \fragmentco{0}{i}$ and $j'\in \fragmentco{0}{j}$
        such that \[
            X\fragmentco{x_i}{x_{i+1}}=X\fragmentco{x_{i'}}{x_{i'+1}}
            \quad\text{and}\quad
            Y\fragmentco{y_j}{y_{j+1}}=Y\fragmentco{y_{j'}}{y_{j'+1}}
        \] with \[
            x_{i}-x_{i'}\le \max(2\ell-1,k)\le k \quad\text{and}\quad
            y_{j}-y_{j'}\le \max(2\ell-1,10k)\le 10k.
        \]

        If $\fragment{y_j-x_{i+1}}{y_{j+1}-x_i}\cap \fragment{-2k}{8k}\ne \emptyset$,
        then $\fragment{y_j-x_{i'+1}}{y_{j+1}-x_{i'}}\cap \fragment{-2k}{10k}\ne
        \emptyset$ and,
        by the inductive assumption, $B_{i',j}$ is isomorphic to a fresh box.
        By a symmetric argument, if $\fragment{y_j-x_{i+1}}{y_{j+1}-x_i}\cap \fragment{8k}{10k}\ne\emptyset$,
        then \(\fragment{y_{j'}-x_{i+1}}{y_{j'+1}-x_{i}}\cap \fragment{-2k}{10k}\ne\emptyset\) and,
        by the inductive assumption, $B_{i,j'}$ is isomorphic to a fresh box.

        Consequently, every relevant box is isomorphic to a fresh box.
    \end{claimproof}

    Now, set $S \coloneqq \{(x,y)\in V(G) \mid x = 0\}$ (or $S \coloneqq \{(0,0)\}$ if only $p=0$
    is allowed)
    and $T \coloneqq \{(x,y)\in V(G) \mid x=|X|\}$ (or $T \coloneqq \{(|X|,|Y|$)\}
    if only $q=|Y|$ is allowed).
    Our goal is to compute the shortest $S\leadsto T$ path in $G$.
    By \cref{clm:relevant}, if the sought edit distance is at most $d$, then this path
    gives the underlying optimal alignment.
    Otherwise, the cost exceeds $d$, and we can return $\infty$.

    To determine $\dist_G(S,T)$, we process relevant boxes $B_{i,j}$ in the lexicographic order.
    For each such box, we already have the distances $\dist_{G}(S,u)$ for every input
    vertex $u$ (either computed trivially or from boxes $B_{i-1,j}$ and $B_{i,j-1}$) and a
    distance matrix $\dist_{G}(u,v)$ with $u$ spanning across the input vertices and $v$
    spanning across the output vertices
    (observe that $\dist_{G}(u,v)=\dist_{B_{i,j}}(u,v)$).

    Since $B_{i,j}$ is a planar graph, the distance matrix satisfies the Monge property.
    Consequently, we can use \cref{thm:smawk}
    to determine \[
        \dist_{G}(S,v)=\min_u \left(\dist_{G}(S,u)+\dist_{G}(u,v)\right)\]
    for all output vertices $v$.

    In order to recover the shortest $S\leadsto T$ path in $G$, we backtrack along that
    path. First, we iterate over the vertices to $T$ to identify
    $v\in T$ satisfying $\dist_{G}(S,v)=\dist_{G}(S,T)$.
    Suppose that the remaining goal is to report the shortest $S\leadsto v$ path in $G$
    for an output vertex $v$ of a relevant box $B_{i,j}$.
    We iterate over the input vertices $u$ of the box $B_{i,j}$ to find one satisfying
    $\dist_{G}(S,u)+\dist_{G}(u,v)=\dist_{G}(S,v)$.
    The values $\dist_{G}(S,u)$ have already been computed, whereas the values
    $\dist_{G}(u,v)$ are stored in a matrix (for a box identical to $B_{i,j}$).
    Once we know $u$, we recover the shortest $u\leadsto v$ path using \cref{thm:klein}.

    We conclude with the running time analysis.
    Checking whether $\selfed(Y)\le 10k$ takes
    $\Oh(n+k^2)$ time.
    The decompositions of $X$ and $Y$ are constructed in $\Oh(n+k^2)$ time using \cref{lem:decomp}.
    Relevant boxes are listed in $\Oh(1)$ time each, for a total of $\Oh({nd}/{\ell^2})$.
    We can partition the relevant boxes into equivalence classes of isomorphisms in
    $\Oh(n+ {nd}/{\ell^2})$ (this step involves linear-time lexicographic sorting of
    phrases). The application of \cref{thm:klein} costs $\Oh(\ell^2 \log \ell)$ time per
    non-isomorphic box, for a total of $\Oh(k^2 \ell\log \ell)$ time.
    Processing each relevant box $B_{i,j}$ using \cref{thm:smawk} costs $\Oh(\ell)$ time,
    for a total of $\Oh({nd}/{\ell})$.
    Finally, the recovery of the optimal alignment costs $\Oh(\ell)$ time per box visited.
    Since any path visits at most $\Oh(m_X+m_Y)=\Oh({n}/{\ell})$ boxes, this final
    step takes $\Oh(n)$ time.
    Thus, the overall running time of our algorithm is $\Oh(n+k^2 + {nd}/{\ell}+k^2\ell\log
    \ell)$, which is  $\Oh(n+k^2 +k\sqrt{nd\log n})$ due to our choice of $\ell = \min(\lceil\sqrt{nd}/(k\sqrt{\log n})\rceil,d)$.
\end{proof}

\subsection{Faster Algorithms for Bounded Weighted Edit Distance}\label{sec:alg:sec:alg-main}

\SetKwFunction{Split}{Split}

\begin{algorithm}[t]
    \caption{The splitting algorithm from \cref{lem:alg-splitting}.
        The procedure $\protect\Split(X, Y, d, k)$ computes
        a partition of $X = X_1 X_2$ and $Y = Y_1 Y_2$ such that $\wed(X, Y) = \wed(X_1,
        Y_1)+\wed(X_2, Y_2)$, provided that $\wed(X, Y) \leq d$.}\label{alg:splitting}
    \Split{$X, Y, d, k$}\Begin{
        $m \gets \lfloor |X| / 2 \rfloor$\;
        $\ell_1 \gets \min\{i \in \fragment{0}{m} \mid \selfed(X\fragmentco{i}{m}) \leq 11k\}$\;
        $\ell_2 \gets \max\{i \in \fragment{m}{|X|} \mid \selfed(X\fragmentco{m}{i}) \leq 11k\}$\;
        $X^*_1 \gets X\fragmentco{\ell_1}{m},\; X^*_2 \gets X\fragmentco{m}{\ell_2}$\;
        $Y^* \gets Y\fragmentco{\max(0,\ell_1-d)}{\min(|Y|,\ell_2+d)}$\;
        \If{$\ell_1 = 0$ \KwSty{and} $\ell_2 = |X|$}{\label{alg:splitting:line:corner-case1}
            Compute $c \gets \wed_{\leq d}(X, Y)$ using~\cref{lem:alg-periodic}\;\label{alg:splitting:line:periodic}
        }\ElseIf{$\ell_2 = |X|$} {
            Compute $c \gets \min_i\wed_{\leq d}(X^*_1 \cdot X^*_2, Y^*\fragmentco{i}{|Y^*|})$ using~\cref{lem:alg-periodic}\label{alg:splitting:line:corner-case2}\;
        }\ElseIf{$\ell_1 = 0$} {
            Compute $c \gets \min_j\wed_{\leq d}(X^*_1 \cdot X^*_2, Y^*\fragmentco{0}{j})$ using~\cref{lem:alg-periodic}\label{alg:splitting:line:corner-case3}\;
        }\Else{
            Compute $c \gets \min_{i,j}\wed_{\leq d}(X^*_1 \cdot X^*_2, Y^*\fragmentco{i}{j})$ using~\cref{lem:alg-periodic}\label{alg:splitting:line:alignment}\;
        }
        \lIf{$c = \infty$}{\Return{\textsc{Fail}}}
        Write $\A$ for the alignment corresponding to $c$,
        and pick $(m, m') \in \A$\;\label{alg:splitting:line:retrievealn}
        $X_1 \gets X\fragmentco{0}{m},\; X_2 \gets X\fragmentco{m}{|X|}$\;
        $Y_1 \gets Y\fragmentco{0}{m'+\max(0, \ell_1-d)},\; Y_2 \gets Y\fragmentco{m'+\max(0, \ell_1-d)}{|Y|}$\;
        \Return{$(X_1, Y_1), (X_2, Y_2)$}\;
        }
\end{algorithm}

In the algorithm $\WeightedED(X, Y, k)$ of~\cref{lem:wed},
we intend to split the two strings $X, Y$ into roughly equal parts and
recursively computes their weighted edit distance.
As a final step before proving \cref{lem:wed}, we hence prove the following.

\begin{lemma}\label{lem:alg-splitting}
    Let $X, Y \in \Sigma^{\le n}$ denote strings and let $0 \leq d \leq k \le n$ denote
    integer parameters. There
    is an algorithm $\Split(X, Y, d, k)$ with the following properties:
    \begin{itemize}
        \item If $\wed(X,Y) \le k$, the algorithm either returns a partitioning $X = X_1
            \cdot X_2, Y = Y_1 \cdot Y_2$ such that
            $\wed(X, Y) = \wed(X_1, Y_1) + \wed(X_2, Y_2)$ and $|X_1| = \lfloor
            |X|/2\rfloor$, or it returns \textsc{Fail}.
        \item If $\wed(X,Y) \leq d$, then the algorithm does not return \textsc{Fail}.
        \item The algorithm runs in $\Oh(n + k^2 + k\sqrt{nd\log n})$-time in the standard
            setting and $\Oh(k^2d\log n)$-time in the \modelname model.
    \end{itemize}
\end{lemma}
\begin{proof}

    We proceed as follows.
    First, we split $X$ in two parts $X = X\fragmentco{0}{m}\cdot X\fragmentco{m}{|X|}$,
    where $m \coloneqq \floor{|X|/2}$.
    Now, let $X_1^* \coloneqq X\fragmentco{\ell_1}{m}$ denote the longest suffix of the
    first part with $\selfed(\cdot) \leq 11k$,
    and let $X_2^* \coloneqq \fragmentco{m}{\ell_2}$ denote the longest prefix of the
    second part with $\selfed(\cdot) \leq 11k$.
    Based on $\ell_1$ and $\ell_2$, we define $Y^* \coloneqq  Y\fragmentco{\max(0,\ell_1-d)}{\min(|Y|,\ell_2+d)}$.
    Next, assuming that $\ell_1 > 0$ and $\ell_2 < |X|$, we compute
    \[c \coloneqq \min_{i,j}\wed_{\leq d}(X^*_1 \cdot X^*_2, Y^*\fragmentco{i}{j})\]
    and the breakpoint representation of its corresponding optimal alignment $\A$ using~\cref{lem:alg-periodic} with
    parameters $d$ and $22k$
    (observe that this is valid since $\selfed(X^*_1\cdot X^*_2) \leq 22k$
    by~\cref{fct:selfed-properties}); see~\cref{alg:splitting:line:alignment}.
    If either $\ell_1 = 0$ or $\ell_2 = |X|$, we fix one of the endpoints of $Y^*$.

    If $c = \infty$, we return \textsc{Fail}. Otherwise, we return the partitioning \[
        X_1 \coloneqq X\fragmentco{0}{m},\quad X_2 \coloneqq X\fragmentco{m}{|X|}
    \]
    and \[
        Y_1 \coloneqq Y\fragmentco{0}{m'+\max(0,\ell_1-d)},\quad Y_2
    \coloneqq Y\fragmentco{m'+\max(0,\ell_1-d)}{|Y|},\] where $(m,m') \in \A$.
    The last detail is that we only have the breakpoint representation of $\A$.
    To recover $m'$ from it, we let $(x^*, y^*)$ be the breakpoint with largest $x^*$ such that $x^* \leq m$.
    Then it holds that $m' = y^* + m - x^*$, since for any two consecutive breakpoints $(x', y')$ and $(x, y)$,
    it holds that $(x-\delta, y-\delta) \in \A$ for $\delta \in \fragmentco{0}{\max(x - x', y-y')}$.
    Consult~\cref{alg:splitting} for the pseudocode.

    First, we analyze the running time.
    The computation of $X_1^*$ and $X_2^*$ using \cref{lem:selfed} takes $\Oh(k^2)$ time
    in the \modelname model.
    The call to~\cref{lem:alg-periodic} takes $\Oh(n + k^2 + k\sqrt{nd\log n})$-time in
    the standard setting and $\Oh(k^2d\log n)$-time in the \modelname model.

    Next, we analyze the correctness.
    If $\wed(X, Y) \leq d$, then we claim that $c \leq d$. Indeed, take any optimal
    alignment $\tilde\A : X \onto Y$. The cost of the induced alignment of $\tilde\A$ on
    $X^*_1 \cdot X^*_2$ is at most $d$.
    Moreover, since $\wed(X,Y) \leq d$, we have that \[
        \big||X^*_1 \cdot X^*_2|-|\A(X^*_1 \cdot X^*_2)|\big| \leq d \]
    and therefore by definition of $Y^*$ it follows that $\tilde\A(X^*_1 \cdot X^*_2)$ is
    contained within $Y^*$. This implies
    that $c \leq d$, (regardless of which of the cases
    in~\crefrange{alg:splitting:line:corner-case1}{alg:splitting:line:alignment}
    compute it), as claimed. Thus, we conclude that if $\wed(X, Y) \leq d$, then the
    algorithm does not return \textsc{Fail}, as stated.

    Now we argue about the correctness of the returned partitioning.
    We claim that every (global) alignment $\tilde\A \in X \onto Y$ attaining
    $\wed_{\tilde\A}(X, Y)\le k$
    \emph{meets} $\A \in \Als(X^*, Y^*\fragmentco{i}{j})$ inside $X^*_1$.
    Let $\A_1 : X^*_1 \onto Y\fragmentco{i}{i+\delta}$ denote the induced alignment of
    $\A$ on $X^*_1$ and let
    $\tilde\A_1 : X^*_1 \onto Y\fragmentco{i'}{i'+\delta'}$ denote the induced alignment
    of $\tilde\A$ on $X^*_1$.
    Formally, we claim that $\A_1\cap \tilde\A_1 \ne \emptyset$.
    For a proof by contradiction, suppose that this is not true.
    Then, in particular, $\A_1$ and $\tilde\A_1$ do not contain any common diagonal edge.
    Hence, we can use~\cref{prop:disjoint-alignments-bound-sed} to conclude that
    \begin{equation}\label{eqn:selfed-bound}
      \selfed(X^*_1) \leq |i-i'|+\ed_{\A_1}(X^*_1,Y\fragmentco{i}{i+\delta}) +
      \ed_{\tilde\A_1}(X^*_1,Y\fragmentco{i'}{i'+\delta'}) +|(i+\delta)-(i'+\delta')|.
    \end{equation}
    To bound this expression, recall that $X_1^* = X\fragmentco{\ell_1}{m}$ and $\A :
    X\fragmentco{\ell_1}{m} \onto Y\fragmentco{i}{i+\delta}$.
    \begin{claim}
        It holds that $|i - \ell_1|, |(i+\delta)-m| \leq 3d$.
    \end{claim}
    \begin{claimproof}
        If $c$ is computed in either \cref{alg:splitting:line:periodic} or \cref{alg:splitting:line:corner-case3},
        then it holds that $\ell_1 = i = 0$. Moreover, since $c \leq d$ it follows that $|(i+\delta)-m| = |\delta - m| \leq d$.
        Thus, in these cases we are done.

        Recall that $X^*_1 = X\fragmentco{\ell_1}{m}, X^*_2 = X\fragmentco{m}{\ell_2}$ and
        $Y^* = Y\fragmentco{\max\{0, \ell_1-d\}}{\min\{|Y|, \ell_2+d\}}$.

        Note that
        \begin{align*}
            c
                &= \min_{i,j} \wed(X^*_1 X^*_2, Y^*\fragmentco{i}{j}) \\
                &= \min_{i',j'} \wed(X\fragmentco{\ell_1}{\ell_2}, Y\fragmentco{\max\{0, \ell_1-d\}+i'}{\min\{|Y|, \ell_2+d\}-j'}).
        \end{align*}
        Fix the $i', j'$ minimizing this expression. We aim to bound $|i - \ell_1| = |i' + \max\{0, \ell_1-d\} - \ell_1|$.

        In particular, since $c \leq d$, it holds that the difference in length between both strings is at most $d$,
        i.e.
        \begin{equation*}\label{proof:alg-splitting:eq:1}
            |\min\{|Y|, \ell_2 + d\} - j' - \max\{0, \ell_1 - d\} - i' - (\ell_2 - \ell_1)| \leq d.
        \end{equation*}

        From this, we infer that
        \begin{equation}\label{proof:alg-splitting:eq:2}
          i' + j' \leq  \min\{|Y|, \ell_2 + d\}  - \max\{0, \ell_1 - d\} - \ell_2 + \ell_1 + d \leq 3d.
        \end{equation}
        The last inequality can be verified by considering the 4 possibilities of which arguments attain the maximum and the minimum.
        For example, if $\ell_2 + d < |Y|$ and $\ell_1 - d < 0$, then we obtain $\ell_2 + d - \ell_2 + \ell_1 + d < 3d$.

        Now, to bound $|i' + \max\{0, \ell_1-d\} - \ell_1|$ consider two cases:
        \begin{itemize}
            \item If $\ell_1 - d \geq 0$, then $|i' + \max\{0, \ell_1-d\} - \ell_1| = |i' - d|$.
            Then note that $d - i' \leq d$. Moreover, using~\eqref{proof:alg-splitting:eq:2} we have that $i' - d \leq 2d$.
            Therefore, $|i' + \max\{0, \ell_1-d\} - \ell_1| \leq 2d$.
            \item If $\ell_1 - d < 0$, then $|i' + \max\{0, \ell_1-d\} - \ell_1| = |i' - \ell_1|$.
            Using~\eqref{proof:alg-splitting:eq:2}, we observe that $i' - \ell_1 \leq i' \leq 3d$.
            Moreover, using the case assumption it holds that $\ell_1 - i' \leq \ell_1 < d$.
            Thus, $|i' + \max\{0, \ell_1-d\} - \ell_1| \leq 3d$.
        \end{itemize}
        We conclude that $|i' + \max\{0, \ell_1-d\} - \ell_1| \leq 3d$, as claimed.

        We can bound $|m - i| = |m - (\max\{0, \ell_1 - d\} + i' + \delta)| \leq 3d$ analogously.
    \end{claimproof}

    For the global alignment, we have that \(\ed_{\tilde\A_1}(X^*_1, Y\fragmentco{i'}{i'+\delta'}) \leq k\)
    and thus $|i'-\ell_1|, |(i'+\delta')-m|\leq k$.
    Combining this with~\eqref{eqn:selfed-bound} we obtain that
    \[
        \selfed(X^*_1) \leq  |i-\ell_1|+|\ell_1-i'|+d+k+|(i+\delta)-m|+|m-(i'+\delta')| \leq 7d + 3k \leq 10k.
    \]
    This implies $\ell_1=0$, as otherwise $\selfed(X\fragmentco{\ell_1-1}{m})\le 10k+1 \le
    11k$ would contradict the definition of $\ell_1$.
    In that case, however, $\tilde\A$ and $\A$ meet at $(0,0)$ because
    \[c=\min_j
    \wed(X\fragmentco{0}{\ell_2},Y\fragmentco{0}{j}),
    \]
    see~\cref{alg:splitting:line:corner-case3,alg:splitting:line:periodic}.

    This contradiction completes the proof that $\tilde\A$ and $\A$ meet inside $X_1^*$.
    A symmetric argument yields that $\tilde\A$ and $\A$ meet inside $X_2^*$.
    By the (local) optimality of $\A$, we can adapt $\tilde\A$ to behave as $\A$ between
    the meeting points without increasing the cost. In particular, this implies that there
    is an optimal alignment $\tilde\A$
    such that $(m,m') \in \tilde\A$. The correctness of the partitioning thus follows
    from~\cref{fct:split-alignment}.
\end{proof}

\begin{algorithm}[t]
    \caption{The algorithm from~\cref{lem:wed}.
        Given strings $X, Y$ and an integer $k$, the algorithm computes $\wed_{\le
    k}(X, Y)$.}\label{alg:wed}

    \WeightedED{$X, Y, k$}\Begin{
        $d \gets \lceil{2k^2}/{n}\rceil$\; \label{alg:wed:line:initialize}
        \If{$\wed(X, Y) \leq d$}{ \label{alg:wed:line:check-base}
            Compute and \Return $\wed(X, Y)$ using~\cref{prop:baseline-wed}\; \label{alg:wed:line:base-case}
        }
        \While{$\Split(X, Y, d, k) = \textsc{Fail}$}{\label{alg:wed:line:while}
            \lIf{$d = k$}{\Return $\infty$}
            $d \gets \min(k, 2 \cdot d)$\;
        }
        $(X_1,Y_1), (X_2,Y_2) \gets \Split(X, Y, d, k)$\; \label{alg:wed:line:split}
        $c_1 \gets \WeightedED(X_1, Y_1, k)$\;
        $c_2 \gets \WeightedED(X_2, Y_2, k)$\;
        \Return $c_1 + c_2$\;
    }
\end{algorithm}

We now put the pieces together to prove~\cref{lem:wed}.

\lemwed
\begin{proof}
    To begin, the algorithm sets $d \coloneqq \ceil{2k^2/n}$
    and checks if $\wed(X, Y) \leq d$ using~\cref{prop:baseline-wed}.
    If this condition is met, the algorithm returns
    the computed value. If not, the algorithm applies~\cref{lem:alg-splitting} to
    partition $X$ and $Y$ using an
    increasing parameter $d, 2d, 4d, \dots$, until the algorithm finds
    a partitioning of $X = X_1 \cdot X_2$ and $Y = Y_1 \cdot Y_2$.
    The process then continues recursively on $(X_1, Y_1)$ and $(X_2, Y_2)$.
    Refer to~\cref{alg:wed} for the pseudocode.

    We begin by analyzing the correctness of the algorithm.
    If the algorithm solves a subproblem directly in~\cref{alg:wed:line:base-case},
    then the correctness follows from~\cref{prop:baseline-wed}.
    Otherwise, since by assumption we have that $\wed(X, Y) \leq k$,
    \cref{lem:alg-splitting} guarantees that the call to \Split in
    line~\cref{alg:wed:line:while} with
    $d/2 < \wed(X, Y) \leq d$ does not fail.
    In particular, in~\cref{alg:wed:line:split}, we obtain a partition
    $(X_1, Y_1), (X_2, Y_2)$ such that $\wed(X, Y) = \wed(X_1, Y_1) + \wed(X_2, Y_2)$. Therefore,
    we can solve each subproblem independently via recursion.

    We now we analyze the running time in the standard setting. First focus on a single
    execution of~\cref{alg:wed}, ignoring the time
    spent in recursive calls. Write $n \coloneqq |X| + |Y|$
    and $c \coloneqq \wed(X, Y)$. If $c < 2k^2/n$, then the algorithm solves the subproblem
    directly in~\cref{alg:wed:line:base-case} in time
    $\Oh(n(c+1)) = \Oh(n + k\sqrt{nc})$ (the additive $n$
    comes from the case $c = 0$).
    Otherwise, consider the while-loop in~\cref{alg:wed:line:while}.
    The $i$-th iteration takes time $\Oh(n + k^2 + k\sqrt{nd \log n})$ where $d =
    \Theta(2^i k^2/n)$ (because of the initialization in~\cref{alg:wed:line:initialize}).
    Thus, we can simplify this expression to $\Oh(n + k\sqrt{nd \log n})$.
    \cref{lem:alg-splitting} guarantees that we
    exit the loop when $d \leq 2c$, which takes at most $\Oh(\log c) = \Oh(\log n)$ iterations.
    Thus, the overall time spent in the while-loop is $\Oh(n\log n + k\sqrt{nc\log n})$.
    Observe that we do not get a log-factor in the second term because the running time
    forms a geometric series dominated by its last term.
    Therefore, we conclude that a single execution of~\cref{alg:wed} takes time $\Oh(n\log
    n + k\sqrt{nc\log n})$.

    Now, let us consider the overall running time of the algorithm, taking into account
    the recursive calls.
    By \cref{lem:alg-splitting}, we recurse on subproblems $(X_i, Y_i)$ where $\wed(X_i,
    Y_i) = c_i$, and
    $|X_i| = |X|/2$ for $i = 1, 2$, and $c_1 + c_2 = c$. Thus, at any level $\ell$ of the
    recursion we have at most $2^\ell$ subproblems
    $(n_i, c_i)$ where $\sum_i n_i \leq n$ and $\sum_i c_i \leq c$. As we argued earlier,
    the time spent on each subproblem is $\Oh(n_i\log n + k\sqrt{n_i c_i\log n_i})$.
    Therefore, we can upper bound the total time spent at any level of the recursion by
    \[
        \Oh\left(\sum_i (n_i\log n + k\sqrt{n_i c_i\log n_i})\right)
        \leq \Oh\left(n\log n  + k\sqrt{\log n}\cdot \sum_i \sqrt{n_i c_i}\right)
        \leq \Oh\left(n\log n + k\sqrt{n c\log n}\right),
    \]
    where the last step follows from the Cauchy-Schwarz Inequality.

    Finally, we bound the depth of the recursion. As stated in \cref{lem:alg-splitting},
    the length of $X$ is halved at each recursive call. Although the length of $Y$ is not
    necessarily halved, we have $\big||X|-|Y|\big| \leq k$ because $\wed(X, Y) \leq k$.
    Hence, after $\Oh(\log n)$ levels of recursion,
    we obtain subproblems where $n = |X| + |Y| \leq 2k$, which are solved directly
    in~\cref{alg:wed:line:base-case}
    (since at this point, $d = \ceil{2k^2/n} \geq k$). Therefore, the depth of the
    recursion is $\Oh(\log n)$, and
    the overall running time of the algorithm is $\Oh(n\log^2 n  + k\sqrt{nc}
    \log^{1.5}n)$, as claimed.

    Next, we follow the same steps to bound the running time of the \modelname
    implementation of the algorithm.
    If $c < 2k^2/n$, then the subproblem is solved directly
    in~\cref{alg:wed:line:base-case} in time $\Oh(nc) = \Oh(k^2)$.
    (Strictly speaking, in the \modelname{} implementation we first
    check if $c = 0$ using one \lceOpName{} query to avoid spending time $\Oh(n)$.)
    Otherwise, the calls to~\cref{lem:alg-splitting} inside the while-loop
    in~\cref{alg:wed:line:while} take total time
    $\Oh(k^2 c\log n)$. Therefore, a single execution of~\cref{alg:wed} takes time
    $\Oh(k^2 c \log n)$.

    Now, let us consider the overall running time of the algorithm, taking into account
    the recursive calls.
    At any level $\ell$ of the recursion we have at most $2^\ell$ subproblems
    $(n_i, c_i)$, where $\sum_i n_i \leq n$ and $\sum_i c_i \leq c$. As we argued earlier,
    the time spent on each
    subproblem is $\Oh(k^2 c_i \log n_i)$. Therefore, we can upper bound
    the total time spent at any level of
    the recursion by $\Oh(k^2 c \log n)$.
    Since the depth of the recursion is $\Oh(\log n)$, the overall running time of the
    algorithm is $\Oh(k^2 c \log^2 n)$, as claimed; completing the proof.
\end{proof}

%% file: sections/lower_bounds.tex
\section{Fine-Grained Lower Bounds and the All-Pairs Shortest Paths Hypothesis}\label{sec:finegrained}

Conditional hardness of and reductions between problems have been integral parts of
Theoretical Computer Science from its very beginnings. Starting from fundamental questions
whether a problem is decidable, over time, researchers developed new and refined existing
techniques to obtain increasingly fine-grained lower bounds for problems.

While classical {\sf NP}-hardness rules out polynomial-time algorithms
conditioned on \({\sf P} \ne {\sf NP}\), such a lower bound, for instance,
does not rule out algorithms running in time \(\Oh(n^{\log n})\)---and such an algorithm
may even be practical.
To rule out such a running time, we need stronger hypotheses.
Hence, the Exponential Time Hypothesis ({\sc eth}) was
introduced~\cite{Impagliazzo2001,Impagliazzo2001a},
(intuitively) conjecturing that the Satisfiability problem has an exponential running
time.
Employing {\sc eth}, researchers
could now prove strong lower bounds, even for parameterized versions of problems. For
instance in~\cite{Chen2005}, the authors prove---among other results---that there is no
\(n^{o(k)}\) algorithm for finding a \(k\) pairwise connected vertices in a graph (a
\(k\)-clique).

For polynomial-time algorithms, a lower bound of the form \(n^{o(k)}\) is still too
weak, though---we wish to determine an \emph{exact} exponent in the lower bound,
ideally matching the running time exactly. To that end, researchers strengthened {\sc
eth}---the Strong Exponential Time Hypothesis ({\sc seth}) (introduced without a name in
\cite{Impagliazzo2001a})---which has since
become a useful hypothesis to obtain fine-grained lower bounds. A prominent example of
such a fine-grained lower bound is the celebrated
result that the (unweighted) edit distance of two strings cannot be computed in
subquadratic time~\cite{bi18}.
Over time, researchers introduced more hypotheses to find convincing arguments why
certain problems might not have faster-than-known algorithms.

For our purposes, it is instructive to work with such an alternative hypothesis. In
particular, we base our hardness results on a hardness hypothesis for the All-Pairs Shortest Paths
(\apsp) problem.

\defproblemalg{All-Pairs Shortest Paths ({\sc apsp})}%
{Weighted, undirected graph \(G = (V, E, \wg)\)}%
{For each pair of vertices \(u, v \in V\), compute the length of the shortest
\(u\)-\(v\)-path.}

The classic textbook algorithm of Floyd and Warshall from the \oldstylenums{1960}s
solves \apsp in cubic time.
As it turns out, there are many graph and matrix problems that are even \emph{equivalent} to
\apsp (for instance computing the min-plus product of two matrices)~\cite{Williams18}, meaning
that any (polynomial) improvement for one problem yields a similar improvement for the
other and vice versa.
However, for none of them, a subcubic algorithm is known; which
motivates the following hypothesis.

\apspconj

In this paper, we give a reduction from \apsp to (variants of) the Bounded Weighted Edit
Distance problem. Assuming~\cref{apsp-conj}, this reduction justifies calling our improved
algorithms optimal.

For our reduction, we in fact start from the following equivalent problem~\cite{Williams18} which
is easier to handle.

\defproblemalg{Negative Triangle}%
{Weighted, undirected graph \(G = (V, E, \wg)\)}%
{Check if there are three vertices \(a, b, c \in V\) that form a triangle of negative
    weight, that is, \(\{a, b\}, \{b, c\}, \{c, a\} \in E\) and \(\wg(a,b) + \wg(b,
c) + \wg(c, a) < 0\).
}

\begin{fact}[\cite{Williams18}]\label{fc:4-1-1}
    For any \(\smallconst > 0\),
    there is no algorithm that solves the Negative Triangle
    problem on graphs \(G = (V, E,\delta)\) with \(|V| \le N\) in time
    \(\Oh(N^{3 - \smallconst})\), unless the {\sc apsp} Hypothesis fails.
\end{fact}

As it turns out, we can simplify our proofs by starting from the following strengthened
version of~\cref{fc:4-1-1}, which seems to be generally known and can be obtained
from~\cref{fc:4-1-1} by standard techniques. Hence, we include only a brief proof sketch.

\cornegtriangle
\begin{proof}[Proof sketch]
    Starting from an instance of Negative Triangle \(G = (V, E)\),
    first  create three copies of the vertices \(V_1, V_2, V_3\),
    and connect each pair of two copies using (copies) of the edges from \(E\) (with one
    endpoint in copy \(V_i\) and the other endpoint in copy \(V_j\)).\footnote{That is, we
    take the (graph) tensor product of the input graph and a triangle.}
    Clearly, any negative triangle in the new instance appeared in the original instance
    and vice versa.

    Next, we partition the sets \(V_1, V_2, V_3\) into (disjoint) subsets of the desired
    sizes and consider the subgraphs induced by three vertex subsets. Clearly, we may
    solve the original instance by solving all of the partitioned instances; this then
    yields the desired lower bound.
\end{proof}

\section{From \texorpdfstring{{\large APSP}}{APSP} to Batched Weighted Edit Distance} \label{sec:batched-wed-lowerbound}

In this section, we obtain our first \apsp-based lowed bound for weighted edit distance
related problems; in particular, we obtain the promised lower bound for Batched Weighted
Edit Distance.

We proceed in two steps. First, we construct and analyze a gadget (of a batch of strings
\(X\), a string~\(Y\), and a weight function) that can be used to ``multiply''
two matrices (using the min-plus product).

In a second step, we then combine multiple copies of said gadget to be able to compute the
min-plus product of three matrices. Now, as we are interested only in the weight of the
minimum weight triangle in a graph (in fact, we are just interested in whether said weight
is negative), we do not have to compute the full matrix product; just the minimum entry of
the resulting matrix suffices.

In total, we prove the following result.
\thmbatched*

\subsection{Computing the Min-Plus-Product of Two Matrices as a Weighted Edit Distance of
Strings}
\label{ssect:lb1}

For
a matrix \(A \in \fragment{-E}{E}^{p \times q}\) and
a matrix \(B \in \fragment{-E}{E}^{q \times r}\)
we write \(\sab\) for an alphabet of
\((2p + r) + (2p + q)\) distinct characters; we also write \[
    \sab = \sabx \cup \sy \coloneqq
    \{x_0,\dots,x_{p-1},\quad x_{p}^{(0)}, \dots, x_{p}^{(r - 1)}, \quad x_{p + 1},
    \dots, x_{2p} \} \;\cup\; \{ y_0, \dots, y_{2p + q - 1}\}.
\]
Observe that the alphabet \(\sy\) is independent of \(A\) and \(B\).
Further, for each \(\ell \in \fragmentco0r\),
we set \(X_{\ell} \coloneqq x_0\cdots x_{p-1} x_{p}^{(\ell)} x_{p +
1}\cdots x_{2p}\) and \(Y \coloneqq y_0\cdots y_{2p + q - 1}\).
Observe that we can obtain any \(X_{\ell}\) from any other \(X_{\ell'}\) by substituting a
single character of \(X_{\ell'}\) (namely \(x_{p}^{(\ell')} \onto x_{p}^{(\ell)}\)).
Hence, we have \(\hd(X_{\ell}, X_{\ell + 1}) \le 1\) for every \(\ell \in
\fragmentco{0}{r-1}\).

Next, for positive \(E \ll D \ll F\)
(where we use \(a \ll b\) to denote \(a (|X| + |Y|) < b\)),
we define the following weight function \(\wabn: (\sab\cup\{\emptystring\})^2 \to \mathbb{Z}_{\ge 0}\):
\begin{alignat*}{3}
    &\wab{\sigma}{\emptystring} &&\coloneqq F, &\qquad& \forall \sigma\in \sab\\[1.5ex]
    &\wab{\emptystring}{\sigma} &&\coloneqq F, && \forall \sigma \in \sab \\[1.5ex]
    &\wab{x_i}{y_{i + j}} &&\coloneqq F + A_{i, j} - A_{i + 1,j} + (q - j) D,
    && \forall i \in \fragmentco{0}{p}, \forall j \in \fragmentco{0}{q}\\
    &\wab{x^{(\ell)}_{p}}{y_{p + j}} &&\coloneqq F + B_{j, \ell},
    && \forall j \in \fragmentco{0}{q}\\
    &\wab{x_{p + i}}{y_{p + i + j}} &&\coloneqq F + (j + 1) D,
    && \forall i \in \fragment{1}{p}, \forall j \in \fragmentco{0}{q}\\[1.5ex]
    &\wab{\sigma}{\rho} &&\coloneqq 2F,
    && \text{otherwise};
\end{alignat*}
where we set \(A_{p,j} \coloneqq 0\) to simplify our exposition.

\begin{lemma}\label{lm:2-14-1}
    For any \(\ell \in \fragmentco0{r}\) and any \(i \in \fragmentco0{p}\), we have \[
        \ed^{\wabn}(X_{\ell}\fragmentco{i}{|X_{\ell}|-i}, Y)
        = |Y|\,F + (p - i)(q + 1)D + \min_{j}\{A_{i,j} + B_{j,\ell}\}.
    \]
    Further, the optimal alignment deletes no characters of \(X_{\ell}\).
\end{lemma}
\begin{proof}
    Fix an \(i \in \fragmentco0{p}\) and an \(\ell \in \fragmentco0{r}\).
    For a \(j \in \fragmentco0{q}\), consider an
    alignment \(\aijl : X_{\ell}\fragmentco{i}{|X_{\ell}| - i} \onto Y\)
    defined by
    \begin{alignat*}{3}
        & \emptystring
        &&\onto Y\fragmentco{0}{i + j}
        &\quad& \text{of cost \((i + j)F\)},\\
        &X_{\ell}\fragmentco{i}{p}    &&\onto Y\fragmentco{i + j}{p + j}
        &\quad& \text{of cost \(\textstyle\sum_{\alpha = i}^{p - 1} (F +  A_{\alpha,j} -
        A_{\alpha + 1, j} + (q - j)D)\)},\\
        &X_{\ell}\position{p}     &&\onto Y\position{p + j}
        &\quad& \text{of cost \(F + B_{j, \ell}\)},\\
        &X_{\ell}\fragmentoc{p}{2p - i}
        &&\onto Y\fragmentoc{p + j}{2p - i + j}
        &\quad& \text{of cost \((p-i)F + (p - i)(j + 1)D \)},\\
        & \emptystring
        &&\onto Y\fragmentoo{2p - i + j}{2p + q}
        &\quad& \text{of cost \((q - j + i - 1)F\)}.
    \end{alignat*}

    \begin{claim}\label{cl:2-14-1}
        The alignment \(\aijl\) has a cost of \(|Y|\, F + {(p-i)(q + 1) D} + A_{i, j} + B_{j, \ell}\).
    \end{claim}
    \begin{claimproof}
        Adding up the costs of the different parts of \(\aijl\), we obtain
        \begin{align*}
            \ed^{\wabn}_{\aijl}(X_{\ell}\fragmentco{i}{|X_{\ell}|-i}, Y)
            &= {\color{highcol} (i + j)F}
            + \sum_{\alpha = i}^{p - 1} ({\color{highcol} F} +  A_{\alpha,j} - A_{\alpha + 1, j} +
            {\color{lowcol} (q - j)D})\\
            &\qquad
            + {\color{highcol} F} + B_{j, \ell}\\
            &\qquad + {\color{highcol} (p-i)F} + {\color{lowcol} (p - i)(j + 1)D}
            + {\color{highcol} (q - j + i - 1)F}\\
            &= {\color{highcol}(i + j + (p - i) + (p - i) + (q - j + i)) F}\\
            &\qquad + {\color{lowcol} ( (p - i)(q - j) + (p-i)(j + 1))D}\\
            &\qquad + A_{i, j} + B_{j, \ell}\\
            &= {\color{highcol}(2p + q) F} + {\color{lowcol} (p-i)(q + 1) D}
            + A_{i, j} + B_{j, \ell}
            \\
            &= {\color{highcol}|Y|\, F} + {\color{lowcol} (p-i)(q + 1) D} + A_{i, j} +
            B_{j, \ell},
        \end{align*} completing the proof of the claim.
    \end{claimproof}

    Observe that \cref{cl:2-14-1} directly implies an upper bound on the (weighted) edit
    distance of \(X_{\ell}\fragmentco{i}{|X_{\ell}| - i}\) and \(Y\); we have\[
        \ed^{\wabn}(X_{\ell}\fragmentco{i}{|X_{\ell}|-i}, Y)
        \quad \le \quad |Y|\,F + (p - i)(q + 1)D + \min_{j}\{A_{i,j} + B_{j,\ell}\}.
    \]
    In the remainder of the proof, we hence show that there are no better alignments
    between \(X_{\ell}\fragmentco{i}{|X_{\ell} - i|}\) and \(Y\).

    First, observe that by choice of \(F\) and \(D\), we may consider the cost of any
    alignment \(X\fragmentco{a}{b} \onto Y\) separately in terms of multiples of \(D\),
    \(F\), and \(1\). As a direct consequence, we see that deleting a character of
    \(X_{\ell}\) can never be optimal.
    \begin{claim}\label{cl:2-14-2}
        Any alignment \(\alm: X_{\ell}\fragmentco{a}{b} \onto Y\) that deletes at least one
        character of \(X_{\ell}\) has a cost of at least \((|Y| + 1) F + R(D, 1)\), where
        \(-F \ll R(D, 1) \ll F\).
    \end{claim}
    \begin{claimproof}
        Fix an alignment \(\alm: X_{\ell}\fragmentco{a}{b} \onto Y\) and write
        \(\hat{X}_{\ell}\) for the string obtained from \( X_{\ell}\fragmentco{a}{b} \)
        after applying all deletions of \(\alm\).

        Observe that any alignment between strings has to insert or delete characters
        according to the length difference of the involved strings.
        In particular, observe that as \(Y\) is already longer than \(X_{\ell}\), we have
        \(y_{b,a} \coloneqq |Y| - |\hat{X}_{\ell}| \ge |Y| - (b - a) > 0\),
        that is, after fixing the length of \(\hat{X}_{\ell}\), the alignment
        \(\alm\) has to insert exactly \(y_{b, a}\) characters (which costs \(y_{b, a} F\))
        and \(\alm\) has to substitute the remaining \(|\hat{X}_{\ell}|\) characters
        (which costs \(|\hat{X}_{\ell}|\, F + R(D, 1)\) for some \(-F \ll R(D,1) \ll F\)).

        In total, an alignment \(\alm\) thus costs\[
            (|X_{\ell}| - |\hat{X}_{\ell}|)F + (|Y|-|\hat{X}_{\ell}|) F
            + |\hat{X}_{\ell}|\, F + R(D, 1) = (|Y| + |X_{\ell}| - |\hat{X}_{\ell}|)F +
            R(D,1),
        \] which yields the claim whenever \(\alm\) deletes at least one character of
        \(X_{\ell}\).
    \end{claimproof}

    By \cref{cl:2-14-2}, any alignment that deletes at least one character of \(X_{\ell}\)
    costs more than any of the alignments \(\aijl\). Similarly, one can show that also any
    alignment using a substitution of cost \(2F\) costs more than any of the alignments
    \(\aijl\). Hence, in the following, we can limit our attention to alignments that
    neither delete a character of \(X_{\ell}\) nor use a substitution of cost \(2F\); we
    call such an alignment a \emph{decent} alignment. Observe that, naturally, the
    alignments \(\aijl\) are decent.

    As any decent alignment has the same cost in terms of \(F\) (namely \(|Y|\,F\)), we
    can focus on the cost of any such alignments in terms of \(D\) and \(1\) (that is,
    modulo \(F\)).

    \begin{claim}\label{cl:2-14-3}
        Any decent alignment \(\alm : X_{\ell}\fragmentco{i}{|X| - i} \onto Y\) that is
        different from \(\aijl\) for all \(j \in \fragmentco{0}{q}\), has a cost of at least
        \(|Y|F + ((p - i)(q + 1) + 1)D + R(1)\), where \(-D \ll R(1) \ll D\).
    \end{claim}
    \begin{claimproof}
        Fix a decent alignment \(\alm : X_{\ell}\fragmentco{i}{|X| - i} \onto Y\) that is
        different from \(\aijl\) for all \(j \in \fragmentco{0}{q}\).

        First, observe that in any alignment \(\aijl\), we can pair up substitutions of
        costs \((q - j)D + R(1)\) with substitutions of cost \((j + 1)D + R(1)\); in total
        there are \((p-i)\) such pairs, yielding a cost of \((p - i)(j + 1)D + R(1)\).

        Now, observe that somewhere, \(\alm\) inserts a character between two substitutions
        (otherwise it would be equal to some \(\aijl\)). In particular, \(\alm\)
        substitutes some character with cost \((q - j)D + R(1)\) and some character with cost
        at least \((j + 2)D + R(1)\), that is, one of the \((p-i)\) pairs has a cost of at
        least \((q + 2)D + R(1)\). However, all other pairs still have a cost of at least
        \((q + 1)D + R(1)\). The claim follows.
    \end{claimproof}

    Observe that by \cref{cl:2-14-3}, the alignments \(\aijl\) have the smallest cost (in
    terms of \(D\)) among the decent alignments. In total, this completes the proof.
\end{proof}

It is useful to discuss alignments \(X_{\ell}\fragmentco{a}{|X_{\ell} - b} \onto Y\) for \(b \ne
a\).

\begin{lemma}\label{lm:2-14-2}
    For any \(a, b \in \fragmentco{0}{p}\), we have
    \[
        \ed^{\wabn}(X_{\ell}\fragmentco{a}{|X_{\ell}| - b}, Y)
        \ge |Y|\,F + (p - i)(q + 1)D + |b - a|\,D + R(1),
    \] where \(-D \ll R(1) \ll D\).
\end{lemma}
\begin{proof}
    Fix \(a, b \in \fragmentco{0}{p}\). Again, by \cref{cl:2-14-2}, we may restrict our
    attention to decent alignments. Further, suppose that \(a < b\), the case \(b > a\)
    can be dealt with similarly, and the case \(a = b\) is already covered by
    \cref{lm:2-14-1}.

    Now, observe that \cref{cl:2-14-3} applies for the fragment
    \(X_{\ell}\fragmentco{b}{|X_{\ell}| - b}\); the substitution of characters
    from \(X_{\ell}\fragmentco{a}{b} \onto Y\fragmentco{a'}{b'}\) is cheapest if
    \(a' = a + q - 1\) and \(b' = b + q - 1\), which has a cost of \((b-a)D\); yielding
    the claim.
\end{proof}

\subsection{Lower Bounds for Batched Weighted Edit Distance}

In a next step, we wish to extend the construction from \cref{ssect:lb1} to compute the minimum weight
triangle of a complete tripartite graph \(G = (P \cup Q \cup R, A \cup B \cup C)\) (or equivalently the minimum entry of a matrix obtained as
the min-plus-product of three matrices). To this end, we intend to first extend the
construction from \cref{ssect:lb1} to compute the minimum weight triangle \emph{that uses
a specific vertex \(w \in R\)}  (or equivalently, we fix a column in \(B\) and a row in
\(C\)). This yields a first string for the batch that we wish to construct.
Then, we iterate over all vertices in \(R\); updating the last string of the batch
accordingly and add the result as a new string to the batch.

To implement the aforementioned plan, we also need to take care of the technical detail
that the construction from \cref{ssect:lb1} requires us to select in advance a vertex \(u
\in P\)---we achieve this by adding a gadget that allows only alignments corresponding to
a specific vertex in \(u\). Then, we again construct a batch of strings as before
(where we now also  iterate over all possible
vertices of \(P\)). In total, this approach results in \(m = p \cdot r\) strings in the
batch.
Choosing \(p, q = \Theta(n)\), this yields the following lower bound.

\begin{restatable}{thought}{dyn-lb-simple}\label{thm:dyn-lb-simple}
    For any \(\smallconst > 0\), and any integers \(0 < r \le n\),
    there are strings \(X_0, \dots, X_{nr}, Y \in \Sigma^{*}\)
    with \(|X_{\star}|, |Y| = \Theta(n)\), and
    a weight function
    \(w : (\Sigma \cup \emptystring)^2 \to \mathbb{Z}_{\ge 0}\),
    on which Batched Weighted Edit Distance takes time
    \(\Omega((n^2 r)^{1-\smallconst})\),
    unless the \apsp hypothesis fails.
\end{restatable}

Observe that \(rn\) strings, each of length \(\Theta(n)\), already yield an input size
of \(\Theta(n^2 r)\), rendering the lower bound of \cref{thm:dyn-lb-simple} meaningless.
Hence, we still need to modify our construction in a substantial way, which we
describe next.

Recall that in the construction in \cref{ssect:lb1}, we encodes the whole matrix of the
edge weights \(A\). As it turns out, we can instead split \(A\) into \(\tau\) parts
\(A_{1},\dots, A_{\tau}\) and use
independent copies of the construction of \cref{ssect:lb1} for each \(A_{i}\) (all using
separate copies of the same column of \(B\)). After slightly adapting our previous idea,
we see that we can then compute the minimum weight triangle of a tripartite graph by
using a batch of \(rp/ \tau\) strings (where consecutive strings now have a Hamming
distance of \(\Oh(\tau)\) instead of just \(\Oh(1)\)).
Using additional (minor) modifications and observations, we then obtain
\cref{thm:dyn-lb-advanced} (which would match \cref{thm:dyn-lb-simple} at \(\scb = 2\) and
again becomes meaningless for \(\scb \ge 2- 2 \smallconst\)).

\thmbatched*

Now, for a formal proof of \cref{thm:dyn-lb-advanced},
we start by formalizing the construction that
allows us to fix a row of the matrix \(C\).
To that end, for
a matrix \(A \in \fragment{-E}{E}^{p \times q}\),
a matrix \(B \in \fragment{-E}{E}^{q \times r}\),
a matrix \(C \in \fragment{-E}{E}^{r \times p}\), and
an integer \(\tau \in \fragment{1}{p}\)
we set \(\ptau \coloneqq  \lceil p / \tau \rceil\)
and we write \(\sabct\) for an alphabet of distinct characters, where
\begin{align*}
    \sabct &= \sabctx \cup \sabcty\\
    \sabctx &\coloneqq \saibx{0} \cup \dots \cup \saibx{\tau - 1}\\
           &\quad\cup
    \{ \S_{0}^{(\ell,0)},\dots, \S_{0}^{(\ell,\ptau - 1)}, \quad \$_{1},\dots,\$_{\tau-1},
    \quad \S_{\tau}^{(0)}, \dots, \S_{\tau}^{(\ptau - 1)} \mid \ell \in
    \fragmentco{0}{r}\} \\
    \sabcty &\coloneqq \sy \cup \{ y_{-(\tau - 1)(2\ptau + 2) - \ptau}, \dots, y_{-1},
    \quad y_{0}, \dots, y_{(\tau - 1)(2\ptau + 2) + \ptau - 1}  \},
\end{align*}
where \(\saibx{\star}\) denote disjoint copies of the alphabets \(\sabx\) from
\cref{ssect:lb1}; in particular, \(\saibx{i}\) is defined on the rows
\(A\fragmentco{i\cdot \ptau}{(i+1) \cdot \ptau}\) of \(A\) and the matrix
\(B\).\footnote{For notational convenience, we define \(\saibx{\tau-1}\) also on \(\ptau\)
    rows, where we set the potential extra rows to be copies of the last
row of \(A\).}
Further, for each \(\ell \in \fragmentco0r\) and \(i \in \fragmentco{0}{\ptau}\)
we set \[
    X_{\ell, i} \coloneqq \S_{0}^{(\ell, i)} X^{(0)}_{\ell} \$_{1}
    X^{(1)}_{\ell} \$_{2} \cdots \$_{\tau - 2}
    X^{(\tau-2)}_{\ell} \$_{\tau-1} X^{(\tau-1)}_{\ell} \S_{\tau}^{(i)},
\] where the strings \(X^{(\star)}_{\ell}\) are defined as in \cref{ssect:lb1} (using the
corresponding rows of \(A\)).
To simplify some of the proofs later on, observe that we have for every \(\alpha, i \in
\fragment{0}{\ptau}\) and \(\ell \in \fragmentco{0}{r}\)
\begin{align*}
    |X^{(\alpha)}_{\ell}| &= 2\ptau + 1,\nonumber\\
    X^{(\alpha)}_{\ell} &= X_{\ell, i}\fragmentco{1 + \alpha (2 \ptau + 2)}{(\alpha + 1) (2
    \ptau + 2)},
    \\
        |X_{\ell, i}| &= 1 + \tau (2 \ptau + 2), \nonumber\\
        \$_{\alpha} &= X_{\ell, i}\position{\alpha(2\ptau + 2)}.
\end{align*}
Further, if \(\ell\) and \(i\) are understood from the context,
we may refer to the characters \(\S_{0}^{(\star)}\) and \(\S_{\tau}^{(\star)}\)
as \(\$_0\) and \(\$_{\tau}\), respectively.

Observe that we can obtain any \(X_{\ell, i}\) from any other \(X_{\ell', i}\) by
substituting \((\tau + 1)\) characters in total: a
single character in each \(X^{(\star)}_{\ell'}\) (namely \(x_{p}^{(\ell',\star)} \onto
x_{p}^{(\ell,\star)}\)) and \(\S_{0}^{\ell',i} \onto \S_{0}^{\ell, i}\).
Further,
observe that we can obtain any \(X_{\ell, i}\) from any other \(X_{\ell, i'}\) by
substituting exactly 2 characters (namely
\(\S_{0}^{(\ell, i')} \onto \S_{0}^{(\ell, i)}\)
and
\(\S_{\tau}^{(i')} \onto \S_{\tau}^{(i)}\)).

Next, we define a string
\[Y \coloneqq  y_{-(\tau - 1)(2\ptau + 2) - \ptau} \dots y_{-1} Y^{(\bullet)}
y_{0} \dots y_{(\tau - 1)(2\ptau + 2) + \ptau - 1},\]
where \(Y^{(\bullet)} \coloneqq Y^{(0)} =
\dots = Y^{(\ptau)}\) is the string from the construction in \cref{ssect:lb1}.
It is convenient to set \(Y\fragmentco{0}{|Y^{(\bullet)}|} \coloneqq Y^{(\bullet)}\), that
is, we index everything in \(Y\)  before \(Y^{(\bullet)}\) using negative indices.

\def\Xli{\ensuremath X_{\ell,i}}
\def\Yb{\ensuremath Y^{(\bullet)}}

Further, it turns out to be useful to name certain fragments of the strings \(\Xli\) and
\(Y\).
For an \(\ell \in \fragmentco{0}{r}\), an \(i \in \fragmentco0{\ptau}\), and an
\(\alpha \in \fragmentco{0}{\tau}\),
we write \(\Xli = \S^{(\ell,i)}_{0} \hat{X}^{(\alpha)}_{\ell, i}
\bar{X}^{(\alpha)}_{\ell, i} \check{X}^{(\alpha)}_{\ell,i} \S^{(i)}_{\tau}\) for a
decomposition of \(X_{\ell, i}\), where
\begin{align*}
    \hat{X}^{(\alpha)}_{\ell, i} &\coloneqq
    \Xli\fragmentoc{0}{\alpha(2\ptau + 2)}\;
    \Xli\fragmentoc{\alpha(2\ptau + 2)}{\alpha(2\ptau + 2) + i}\\
                                 &\,= \Xli\fragmentoo{0}{\alpha(2\ptau +
                                 2)}\;\$_{\alpha}\; X^{(\alpha)}_{\ell}\fragmentco{0}{i},\\
    \bar{X}^{(\alpha)}_{\ell, i} &\coloneqq
    X^{(\alpha)}_{\ell}\fragmentco{i}{|X_{\ell}^{(\alpha)}| - i},\\
    \check{X}^{(\alpha)}_{\ell, i} &\coloneqq
    \Xli\fragmentco{(\alpha+1)(2\ptau + 2) - i}{(\alpha + 1)(2\ptau + 2)}\;
    \Xli\fragmentco{(\alpha+1)(2\ptau + 2)}{\tau(2\ptau + 2)}\\
                                   &\,=
                                   X^{(\alpha)}_{\ell}\fragmentco{|X_{\ell}^{(\alpha)}| -
                                   i}{|X^{(\alpha)}_{\ell}|}\;\$_{\alpha+1}\;
                                   \Xli\fragmentoo{(\alpha+1)(2\ptau + 2)}{\tau(2\ptau + 2)}.
\end{align*}
Similarly, we decompose
\(Y =
\hat{Y}'_{\alpha,i} |_0^{(\alpha, i)}  \hat{Y}_{\alpha,i} \bar{Y}_{\alpha,i}
\check{Y}_{\alpha,i} |_{\tau}^{(\alpha, i)} \check{Y}'_{\alpha,i}
\), where
\begin{align*}
    \hat{Y}'_{\alpha,i} &\coloneqq Y\fragmentco{-(\tau-1)(2\ptau + 2) -
    \ptau}{-(\alpha(2\ptau + 2) + i + 1)},\\
        |_0^{(\alpha, i)}  &\coloneqq Y\position{-(\alpha(2\ptau + 2) + i + 1)}
        = y_{-(\alpha(2\ptau + 2) + i + 1)},\\
        \hat{Y}_{\alpha,i} &\coloneqq Y\fragmentoc{-(\alpha(2\ptau + 2) + i + 1)}{-1}, \\
        \bar{Y}_{\alpha,i} &\coloneqq Y\fragmentco{0}{|\Yb|} = \Yb,\\
        \check{Y}_{\alpha,i} &\coloneqq \fragmentco{|\Yb|}{|\Yb| +
        (\tau-\alpha-1)(2\ptau + 2) + i},\\
            |_{\tau}^{(\alpha, i)} &\coloneqq Y\position{|\Yb| + (\tau-\alpha-1)(2\ptau + 2) + i}
            = y_{(\tau-\alpha-1)(2\ptau + 2) + i},\\
            \check{Y}'_{\alpha,i}&\coloneqq Y\fragmentoo{|\Yb| + (\tau-\alpha-1)(2\ptau + 2) +
            i}
            {|\Yb| + (\tau-1)(2\ptau + 2) + \ptau}.
\end{align*}
Observe that for any  \(\ell \in \fragmentco{0}{r}\), an \(i \in \fragmentco0{\ptau}\), and an
\(\alpha \in \fragmentco{0}{\tau}\), we have
\begin{align*}
    |\hat{X}^{(\alpha)}_{\ell, i}| &= |\hat{Y}_{\alpha,i}|
    = \alpha(2\ptau + 2) + i,\\
    |\check{X}^{(\alpha)}_{\ell, i}| &= |\check{Y}_{\alpha,i}|
    = (\tau - \alpha - 1)(2\ptau + 2) + i;
\end{align*}
and observe that \(\hat{Y}_{\alpha,i} = \emptystring\) for \(\alpha = \tau - 1\) and
\(i = \ptau - 1\) and \(\check{Y}_{\alpha,i} = \emptystring\) for \(\alpha = 0\) and
\(i = \ptau - 1\).

Next, for positive \(E \ll D \ll I \ll F\)
(where we use \(a \ll b\) to denote \(a (|X| + |Y|) < b\)),
we define the following weight function \(\wabctn: (\sabct \cup \{\emptystring\})^2 \to \mathbb{Z}_{\ge 0}\):
\begin{alignat*}{3}
    &\wabct{\sigma}{\emptystring} &&\coloneqq F, &\qquad& \forall \sigma\\[1.5ex]
    &\wabct{\emptystring}{\sigma} &&\coloneqq F, && \forall \sigma \\[2ex]
    &\wabct{x}{y} &&\coloneqq \wabn^{(\ell)}(x \to y)
    && \forall x \in \saibx{\ell}\quad \forall y \in \sy\\[1.5ex]
    &\wabct{x}{y} &&\coloneqq F
    \qquad\qquad\qquad\forall y \notin \Yb
    &&
    \forall x \notin \{ \S_{0}^{(\ell, i)}, \S_{
        \tau}^{(i)} \mid i \in \fragmentco{0}{\ptau}, \ell \in
    \fragmentco0r\}\\
    &\wabct{\S_{ 0}^{(\ell, i)}}{|_{0}^{(\alpha, i)}} &&\coloneqq
    F + C_{\ell, i +  \alpha \ptau} + (\tau - \alpha) \, I + i(q + 1)D\hspace{-20em}\\
    & && && \forall i \in \fragmentco{0}{\ptau} \quad \forall \alpha \in \fragmentco{0}{\tau}\\
    &\wabct{\S_{ \tau}^{(i)}}{|_{\tau}^{(\alpha, i)}} &&\coloneqq F + (\alpha + 1)\, I
    && \forall i \in \fragmentco{0}{\ptau} \quad \forall \alpha \in
    \fragmentco{0}{\tau}\\[1.5ex]
    &\wabct{\sigma}{\rho} &&\coloneqq 2F,
    && \text{otherwise}.
\end{alignat*}

To gain some intuition for \(\wabctn\), let us first focus on the case \(\tau = 1\).
Intuitively, a substitution \(\S_{0}^{\ell, i} \onto y\) has either
a cost of \(F + C_{\ell, i}\) plus some auxiliary terms (if \(y = |_{0}^{(0,i)}\))
or a cost of \(2F\) (otherwise).
In other words, for a fixed \(\ell\), each of the different
characters \(\S_{0}^{\ell, \star}\) can be substituted to exactly one of the characters
\(y\)
thereby ``selecting'' a specific diagonal to follow
(recall that a cost of \(2F\) is virtually the same as \(\infty\)).
Similarly, the costs of the substitutions \(\S_{\tau}^{\star} \onto y\)
select a single diagonal where a substitution is possible.

In total (and by assigning appropriate auxiliary costs), we force the optimal alignment of
\(\Xli \onto Y\) to have to following shape: the alignment first has to insert the
prefix \(\hat{Y}'_{0,i}\); then the alignment
substitutes \(\S_{0}^{(\ell,i)}\hat{X}^{(0)}_{\ell,i} \onto |_{0}^{(0,i)} \hat{Y}_{0,i}\).
Next, we follow the alignment of
\cref{lm:2-14-1}, which aligns
\[
    \bar{X}_{\ell,i}^{(0)} =
    X^{(0)}_{\ell}\fragmentco{i}{|X^{(0)}_{\ell}|-i} \onto
    \Yb.\]
Next, we substitute the remaining suffix of \(X_{\ell,i}\), namely
\(
    \check{X}_{\ell,i}^{(0)} \S_{\tau}^{(i)} \onto
    \check{Y}_{0,i} |_{\tau}^{(0,i)}\).
Finally, we insert the remaining characters from \(Y\).

For general \(\tau\), multiple substitutions of \(\S_{\star}^{\star}\) are possible for
a given \(i\); namely, exactly one for each of the \(\tau\) strings
\(X^{(\star)}_{\ell}\). For technical reasons, we need an additional buffer character
between the strings \(X^{(\star)}_{\ell}\), which causes slight index shifts for the
characters in \(Y\) that we pick as the target of a substitution.

Formally, we obtain the following result.

\begin{lemma}\label{lm:3-21-1}
    For any \(\ell \in \fragmentco{0}{r}\), any \(i \in \fragmentco0{\ptau}\) we have \[
        \ed^{\wabctn}(X_{\ell, i}, Y)
        =
        |Y|\,F + (\tau + 1) I + \ptau(q + 1)D + \min_{\alpha \in
            \fragmentco{0}{\tau}, j}\{A_{i + \alpha\ptau,j} + B_{j,\ell} +
        C_{\ell, i + \alpha\ptau}\}.
    \]
    Further, the optimal alignment deletes no characters of \(X_{\ell, i}\).
\end{lemma}
\begin{proof}

    Next, for an \(\ell \in \fragmentco{0}{r}\), an \(i \in \fragmentco0{\ptau}\), and an
    \(\alpha \in \fragmentco{0}{\tau}\),
    write \(\aial : \bar{X}_{\ell,i}^{(\alpha)} \onto \Yb\)
    for the optimal alignment from \cref{lm:2-14-1}. Recall that by
    \cref{lm:2-14-1}, we have
    \[
        \ed^{\wabctn}_{\aial}(\hat{X}_{\ell,i}^{(\alpha)}, \Yb)
        = |\Yb|\,F + (\ptau - i)(q + 1)D + \min_{j}\{A_{i + \alpha\ptau,j} + B_{j,\ell}\}.
    \]

    We extend \(\aijl\) into an alignment \(\aijla : X_{\ell, i} \onto Y\) by aligning
    \begin{alignat*}{3}
        &\emptystring &&\onto \hat{Y}'_{\alpha,i}
                &&\qquad \text{of cost \(|\hat{Y}'_{\alpha,i}|\, F\)},\\
        &X_{\ell, i}\position{0} = \S_{ 0}^{\ell, i }
        &&\onto |_{0}^{(\alpha, i)}
        &&\qquad \text{of cost \( F + C_{\ell, i + \alpha \ptau} + (\tau - \alpha) \cdot I
            + i(q + 1)\,D
        \)},\\
        &\hat{X}^{(\alpha)}_{\ell, i}
        &&\onto \hat{Y}_{\alpha,i}
        &&\qquad \text{of cost \(|\hat{Y}_{\alpha,i}|\, F\)},\\
        &
        \bar{X}_{\ell,i}^{(\alpha)}
        &&\onto Y^{(\bullet)}
        &&\qquad \text{of cost \(|\Yb|\,F + (\ptau - i)(q + 1)D + \min_{j}\{A_{i +
        \alpha\ptau,j} + B_{j,\ell}\} \)},\\
        &\check{X}_{\ell, i}^{(\alpha)}
        &&\onto \check{Y}_{\alpha,i}
        &&\qquad \text{of cost \(|\check{Y}_{\alpha,i}|\, F\)},\\
        &X_{\ell, i}\position{|X_{\ell, i}| - 1} = \S_{ \tau}^{i}
        &&\onto |_{\tau}^{(\alpha, i)}
        &&\qquad \text{of cost \( F + (\alpha + 1) \cdot I \)},\\
        &\emptystring &&\onto \check{Y}'_{\alpha,i}
        &&\qquad \text{of cost \(|\check{Y}'_{\alpha,i}|\, F\)}.
    \end{alignat*}

    Next, we argue that (some of) the alignments \(\aijla\) attain the claimed
    optimum weighted edit distance between \(\Xli\) and \(Y\) and that any optimum
    alignment indeed is one of the alignments \(\aijla\).
    To that end, fix an \(\ell \in \fragmentco{0}{r}\), an \(i \in \fragmentco0{\ptau}\), and an
    \(\alpha \in \fragmentco{0}{\tau}\).

    \begin{claim}\label{cl:3-14-1}
        The alignment \(\aijla\) has a cost of \[
            |Y|\,F + (\tau + 1) I + \ptau(q + 1)D + \min_{
            j}\{A_{i + \alpha\ptau,j} + B_{j,\ell}\} +
            C_{\ell, i + \alpha\ptau}.
        \]
        Further, the alignment \(\aijla\) deletes no characters of \(X_{\ell, i}\).
    \end{claim}
    \begin{claimproof}
        By choice of \(E\), \(D\), \(I\), and \(F\), we can bound the cost of \(\aijla\)
        separately in terms of each of \(F\) and (together) \(E\), \(D\), \(I\), and \(1\).

        For the cost in terms of \(F\), we observe that the alignment \(\aijla\) never
        uses a deletion (which is easily verified from the definition and the fact that
        the optimal alignment from \cref{lm:2-14-1} does not delete characters from
        \(X_{\ell, i}\)),
        and that all of its insertions and deletions cost \(F\) (plus
        some lower order terms). As \(Y\) is longer than \(X_{\ell, i}\), this gives a bound of
        \(|Y| F\). Indeed, this is verified by adding up the corresponding terms of the
        costs of the parts of the alignment \(\aijla\).

        For the other terms of the cost, summation yields
        \begin{align*}
        &C_{\ell, i + \alpha \ptau} + {\color{highcol}(\tau - \alpha) \cdot I} +
        {\color{lowcol}
        i(q + 1)D}
        + {\color{lowcol} (\ptau - i)(q + 1)D} +
        \min_{j}\{A_{i + \alpha\ptau,j} + B_{j,\ell}\}
        + {\color{highcol} (\alpha + 1) \cdot I}\\
        &= (\tau + 1)\,I + \ptau(q + 1)\,D +
        \min_{j}\{A_{i + \alpha\ptau,j} + B_{j,\ell}\} + C_{\ell, i + \alpha \ptau}.
        \end{align*}
        Taken together, this yields the claim.
    \end{claimproof}

    Observe that \cref{cl:3-14-1} immediately yields an upper bound on
    \(\ed^{\wabctn}(X_{\ell, i}, Y)\) of the desired value. For the lower bound, we argue
    that any other alignments have a higher cost already in terms of one of \(D\), \(I\),
    or \(F\).

    As in the proof of \cref{lm:2-14-1}, we first limit our attention to decent
    alignments, that is, alignments \(X_{\ell, i} \onto Y\) that neither delete characters from
    \(X\), nor use a substitution of cost \(2F\).
    Similarly to \cref{cl:2-14-2}, it can be easily verified that
    \begin{itemize}
        \item any decent alignment has a cost of \(|Y|F\) (plus some lower-order terms);
            and
        \item any non-decent alignment has a cost of at least \((|Y| + 1)F\).
    \end{itemize}
    As all decent alignments do not differ in their cost in terms of \(F\), we henceforth
    focus on the costs of alignments modulo \(F\).

    \begin{claim}\label{cl:3-15-1}
        Any decent alignment \(\alm: \Xli \onto Y\) has all of the following properties.
        \begin{enumerate}[\sf (\ref{cl:3-15-1}-1)]
            \item \(\alm\) substitutes
                \(\S_{0}^{(\ell, i)} \onto |_{0}^{(\alpha, i)}\)
                for some \(\alpha \in \fragmentco{0}{\tau}\).
            \item \(\alm\) substitutes
                \(\S_{\tau}^{(i)} \onto |_{\tau}^{(\beta, i)}\)
                for some \(0 \le \beta \le \alpha\).
            \item If \(\beta < \alpha\), then \(\alm\) has a cost (modulo \(F\)) of at
                least \((\tau + 2)\, I\).
            \item If \(\beta = \alpha\), we say that \(\alm\) is a \emph{becoming} alignment.
                Any becoming alignment has a cost of \((\tau + 1)\,I\) (plus some
                lower-order terms).
        \end{enumerate}
    \end{claim}
    \begin{claimproof}
        The first claim is immediate from the construction: all other substitutions for
        \(\S_0^{(\ell, i)}\) have a cost~of~\(2F\).

        For the second claim, suppose that \(\alm\) substitutes
        \(\S_{0}^{(\ell, i)} \onto |_{0}^{(\alpha,i)}\)
        for some \(\alpha \in \fragmentco{0}{\tau}\).
        We wish to compute the first character \(|_{\tau}^{(\beta,i)}\) that may be a
        target of a substitution in \(\alm\).

        To that end, first recall that as \(\alm\) is decent, \(\alm\) may not delete
        characters of \(\Xli\).
        Hence, suppose that \(\alm\) continues to substitute all characters from
        \(X_{\ell, i}\) until (including) \(\S_{\tau}^{(i)}\).

        Now, first consider the case \(q = 1\). In this case, we have \(|\Yb| = 2\ptau + 1
        = |X_{\ell}^{(\alpha)}|\), so \(\alm\) substitutes
        \(\check{X}_{\ell,i}^{(\alpha)}
        \onto \check{Y}^{(\alpha, i)})\) and thus also \(\S_{\tau}^{(i)} \onto
        |_{\tau}^{(\alpha,i)}\); yielding the claim.

        Next, consider the case \(q > 1\). Now,
        \(\alm\) substitutes \(\$_{\alpha + 1}\) (or \(\S_{\tau}^{(i)}\), if \(\alpha =
        \tau-1\)) for a character in
        \(Y^{(\bullet)}\)---which costs \(2F\) and is hence not possible:
        we have
        \begin{align*}
                |\S_{0}^{(\ell, i)}\cdots \$_{\alpha + 1}|
            &=
            |\Xli\fragment{0}{(\alpha + 1)(2\ptau + 2)}| =
            1 + (\alpha + 1)(2\ptau + 2)\nonumber\\
            &> 1 + \alpha(2\ptau + 2) + i = |\,|_{0}^{(\alpha,i)} \hat{Y}_{\alpha,i}|
            \intertext{and}
                |\S_{0}^{(\ell, i)}\cdots \$_{\alpha + 1}|
            &=
            |\Xli\fragment{0}{(\alpha + 1)(2\ptau + 2)}| =
            1 + (\alpha + 1)(2\ptau + 2)\nonumber\\
            &< 1 + \alpha(2\ptau + 2) + i + 2\ptau + q = |\,|_{0}^{(\alpha,i)} \hat{Y}_{\alpha,i} \Yb|
        \end{align*}
        In fact, the previous argument shows that \(\alm\) has to insert at least
        \((q-1)\) characters of \(\,|_{0}^{(\alpha,i)} \hat{Y}_{\alpha,i}\Yb\) into
        \(\Xli\) to be decent.
        However, inserting just \((q-1)\) characters aligns
        \(\check{X}_{\ell,i}^{(\alpha)}
        \onto \check{Y}^{(\alpha, i)})\) and thus also \(\S_{\tau}^{(i)} \onto
        |_{\tau}^{(\alpha,i)}\); yielding the claim.

        The last two claims follow immediately from adding the corresponding costs of the
        substitutions
        \(\S_{0}^{(\ell, i)} \onto  y_{-\alpha(2\ptau + 2) - i - 1}\)
        and \(\S_{\tau}^{(i)} \onto  y_{i + (\tau - \beta - 1)(2\ptau + 2)}\) and
        observing that only substitutions \(\S^{(\star)}_{\star} \onto y_{\star}\) may incur costs
        in terms of \(I\).
    \end{claimproof}

    By \cref{cl:3-15-1}, we may focus on becoming alignments; as all becoming alignments
    do not differ in their cost in terms of \(I\), we henceforth focus on the costs of
    alignments modulo \((I\,F)\).

    \begin{claim}\label{cl:3-15-2}
        Write \(\alm: X_{\ell, i} \onto Y\) for a becoming alignment that for some
        \(\alpha \in \fragment{0}{\tau}\) substitutes
        \(\S_{0}^{(\ell, i)} \onto |_{0}^{(\alpha,i)}\)
        and  \(\S_{\tau}^{(i)} \onto |_{\tau}^{(\alpha,i)}\).
        Then, \(\alm\) has all of the following properties.
        \begin{enumerate}[\sf (\ref{cl:3-15-2}-1)]
            \item \(\alm\) aligns \(X_{\ell, i} \onto
                {Y'} \coloneqq |_{0}^{(\alpha,i)} \hat{Y}_{\alpha,i} \Yb
                \check{Y}_{\alpha,i} |_{\tau}^{(\alpha,i)}\).
                In particular,
                \(\alm\) deletes \((q + 2i - 1)\) characters of \({Y'}\).
            \item If \(\alm\) deletes \((q + 2i - 1)\) characters from
                \(Y^{(\bullet)}\), then \(\alm\) aligns
                \(\bar{X}_{\ell, i}^{(\alpha)} \onto
                Y^{(\bullet)}\); we call such an alignment a \emph{good} alignment.
            \item If \(\alm\) does not delete \((q + 2i - 1)\) characters from
                \(Y^{(\bullet)}\), then \(\alm\) has a cost of at least
                    \(((\ptau - i)(q + 1) + 1)\,D\).
        \end{enumerate}
    \end{claim}
    \begin{claimproof}
        For the first claim, we use
        \begin{align*}
            |Y^{(\bullet)}| = 2\ptau + q,&\qquad
            |\bar{X}_{\ell,i}^{(\alpha)}| =
            |X_{\ell}^{(\alpha)}\fragmentco{i}{|X_{\ell}^{(\alpha)}| - i}
            = 2\ptau + 1 - 2i,\\
            |\hat{X}_{\ell,i}^{(\alpha)}| = |\hat{Y}_{\alpha,i}|,&\qquad
            |\check{X}_{\ell,i}^{(\alpha)}| = |\check{Y}_{\alpha,i}|,
        \end{align*}
            and calculate
        \begin{align*}
            |Y'| - |\Xli| &= |\Yb| - |\bar{X}_{\ell,i}^{(\alpha)}|
                          =  (2\ptau + q) - (2\ptau + 1- 2i)
                                 = q + 2i -1;
        \end{align*}
        yielding the first claim.

        For the second claim, if \(\alm\) deletes \((q + 2i - 1)\) characters from
        \(Y^{(\bullet)}\), then \(\alm\) deletes no characters from \(Y'\) outside of
        \(Y^{(\bullet)}\) (otherwise, \(\alm\) would have to delete a character from
        \(X_{\ell, i}\), which is not possible in a decent alignment).
        Hence, \(\alm\) substitutes \(\hat{X}_{\ell, i}^{(\alpha)} \onto
        \hat{Y}_{\alpha,i}\) and \(\check{X}_{\ell, i}^{(\alpha)} \onto
        \check{Y}_{\alpha, i}\), and thus also
        \(\bar{X}_{\ell, i}^{(\alpha)} \onto \Yb\), yielding the second claim.

        Finally, for the third claim, the second claim and \(\alm\) being decent imply
        that \(\alm\) deletes less than \((q + 2i - 1)\) characters from
        \(Y^{(\bullet)}\), and thus \(\alm\) aligns a fragment
        \(X' \coloneqq X^{(\alpha)}_{\ell}\fragmentco{a}{|X^{(\alpha)}_{\ell}| - b} \onto
        \Yb\) for \(a, b \in \fragment{0}{i}\) and \(a \ne i\) or \(b \ne i\).
        Now, \cref{lm:2-14-2} yields the claim.
    \end{claimproof}

    Recall that for any good alignment \(\alm\), we can apply \cref{lm:2-14-1} to see that
    the cost of \(\alm\) must be at least the cost of the alignment \(\aijla\) (for the
    corresponding \(\alpha\)). In total, this shows that the alignments \(\aijla\) are
    optimal among the becoming alignments that substitute
    \(\S_{0}^{(\ell, i)} \onto |_{0}^{(\alpha,i)}\).
    This yields the desired lower bound on the weighted edit distance between \(X_{\ell,
    i}\) and \(Y\), and hence the claimed result.
\end{proof}

In a next step, we modify the (nonzero) weights in \(\wabctn\) of \cref{lm:3-21-1} to lie between
\(\interval12\).
To that end, write \(K \coloneqq F\cdot(|\Xli| + |Y|)\) and
consider the modified weight function \(\wabcttn\) that is given
by
\begin{alignat*}{3}
    &\wabctt{x}{y} = \wabctt{y}{x}\\
                  &\quad\coloneqq
    \begin{cases}
        0 & \text{if \(x = y\)},\\
        1 & \text{if \(x \ne y\) and \(x, y \in \sabctx \cup  \{\emptystring\}\)},\\
        2 & \text{if \(x \ne y\) and \(x, y \in \sabcty \cup \{\emptystring\}\)},\\
        1 + \min\{\wabct{x}{y}, \wabct{y}{x}\} / K
              & \text{if \(x \in \sabctx\) and \(y \in \sabcty\).}
    \end{cases}
\end{alignat*}
Observe that we take the minimum \(\min\{\wabct{x}{y}, \wabct{y}{x}\}\) to make
\(\wabcttn\) symmetric, as (the up until now unused) values \(\wabct{y}{x}\) are \(2F\) for
\(y \in \Sigma_Y\) and \(x \in \Sigma_X\).

Next, we convince ourselves that \cref{lm:3-21-1} holds with \(\wabcttn\) as
well---intuitively, this is immediate, as any good alignment of \cref{lm:3-21-1} gets
scaled in the same way (as all of them have the same number of deletions of characters of
\(Y\)).

\begin{lemma}\label{lm:3-22-1}
    For a matrix \(A \in \fragment{-E}{E}^{p \times q}\),
    a matrix \(B \in \fragment{-E}{E}^{q \times r}\),
    a matrix \(C \in \fragment{-E}{E}^{r \times p}\),
    an integer \(\tau \in \fragment{1}{p}\),
    and the corresponding alphabet \(\sabct = \sabctx \cup \sabcty\),
    the weight function \(\wabcttn : (\sabct \cup\{\emptystring\} \to
    \{0,2\} \cup \interval{1}{1 + 1/(|X| + |Y|)}\) is a metric and satisfies
    \begin{align*}
        \wabctt{x}{y} = \wabctt{y}{x} &=
        \begin{cases}
            0 & \text{if \(x = y\)},\\
            1 & \text{if \(x \ne y\) and \(x, y \in \sabctx \cup \{\emptystring\}\)},\\
            2 & \text{if \(x \ne y\) and \(x, y \in \sabcty \cup \{\emptystring\}\)},\\
        \end{cases}\\
        \wabctt{x}{y} = \wabctt{y}{x} &\in
            \interval{1}{1 + 1/(|X| + |Y|)}
                \quad\text{if \(x \in \sabctx\) and \(y \in \sabcty\).}
    \end{align*}

    Further, any alignment \(\alm: \Xli \onto Y\) of cost \(\alpha F + c\) under \(\wabctn\)
    that inserts \(d_Y \ge |Y| - |\Xli|\) characters of \(Y\) and that deletes
    \(d_X = d_Y - (|Y| - |\Xli|)\) characters of \(\Xli\),
    has a cost of \(((\alpha - d_X - d_Y)F + c) / K + |Y| + d_X + d_Y\) under
    \(\wabcttn\).
\end{lemma}
\begin{proof}
    First, we easily convince ourselves that the values of \(\wabcttn\) take indeed the
    claimed values. Indeed, this follows from the construction and our choice of \(K
    \coloneqq F \cdot (|\Xli| + |Y|)\).

    Next, we observe that \(\wabcttn\) is indeed a metric: \(\wabcttn\) is symmetric (by
    construction), the distance of every character to itself is \(0\), and \(\wabcttn\)
    satisfies the triangle inequality (as adding any two numbers that are at least \(1\)
    yields a number that is at least \(2\), and hence larger than any possible value of
    \(\wabcttn\)).

    Finally, consider an alignment \(\alm: \Xli \onto Y\) of cost \(|Y|F + c\) under \(\wabctn\)
    that inserts \(d_Y \ge |Y| - |\Xli|\) characters of \(Y\) and that deletes
    \(d_X = d_Y - (|Y| - |\Xli|)\) characters of \(\Xli\).
    Recall that insertions and deletions have a cost of \(F\) under \(\wabctn\).

    Now, observe that apart from deletions of characters in \(\Xli\) or insertions of
    characters in \(Y\), the alignment \(\alm\) performs
    only substitution of characters in \(\sabctx\) to
    characters in \(\sabcty\). Any such substitution with a cost \(\alpha' F + c'\) under
    \(\wabctn\) has a cost of \(1 + (\alpha' F + c') / K\) under \(\wabcttn\).
    Hence, the total cost of all (of the \(|\Xli| - d_X = |Y| - d_Y\))
    substitutions under \(\wabctn\) is
    \((\alpha - d_X - d_Y)F + c\), which transforms into a cost of
    \(((\alpha - d_X - d_Y)F + c) / K + (|Y| - d_Y)\).
    Now for the \(d_X\) deletions and \(d_Y\) insertions, their total cost of \(d_X F +
    d_Y F\) under \(\wabctn\) transforms into a cost of \(d_X + 2d_Y\) under \(\wabcttn\);
    in total yielding the claim.
\end{proof}
\begin{corollary}\label{cl:3-22-2}
    For any \(\ell \in \fragmentco{0}{r}\), any \(i \in \fragmentco0{\ptau}\),
    the optimal alignment \(\alm: \Xli \onto Y\) under \(\wabctn\) is also an optimal
    alignment under \(\wabcttn\).

    In particular, the optimal alignment \(\alm^*\) from \cref{lm:3-21-1} has (under
    \(\wabcttn\)) a cost of
    \begin{align*}
        &\ed^{\wabcttn}_{\alm^*}(\Xli, Y)\nonumber\\
        &\,=
        \left(|\Xli|\,F + (\tau + 1) I + \ptau(q + 1)D + \min_{\alpha \in
            \fragmentco{0}{\tau}, j}\{A_{i + \alpha\ptau,j} + B_{j,\ell} +
        C_{\ell, i + \alpha\ptau}\}\right) / K + 2 |Y| - |\Xli|.
    \end{align*}
\end{corollary}
\begin{proof}
    First, recall that the optimal alignment \(\alm^*\) from \cref{lm:3-21-1} deletes no characters
    from \(\Xli\). Hence, \cref{lm:3-22-1} yields that \(\alm^*\) under \(\wabcttn\) has a
    cost of
    \begin{align*}
        &\ed^{\wabcttn}_{\alm^*}(\Xli, Y)\nonumber\\
        &\,=
        \left(|\Xli|\,F + (\tau + 1) I + \ptau(q + 1)D + \min_{\alpha \in
            \fragmentco{0}{\tau}, j}\{A_{i + \alpha\ptau,j} + B_{j,\ell} +
        C_{\ell, i + \alpha\ptau}\}\right) / K + 2 |Y| - |\Xli|.
    \end{align*}
    Observe that we have \[
        L \coloneqq \left(|\Xli|\,F + (\tau + 1) I + \ptau(q + 1)D + \min_{\alpha \in
            \fragmentco{0}{\tau}, j}\{A_{i + \alpha\ptau,j} + B_{j,\ell} +
    C_{\ell, i + \alpha\ptau}\}\right) / K < 1.\label{eq:3-22-3}
    \]

    Now, as \(\alm^*\) is optimal under \(\wabctn\), it is in particular optimal (under
    \(\wabctn\)) among all alignments that do not delete a character from \(\Xli\).
    By \cref{lm:3-22-1}, \(\alm^*\) is hence also optimal among all
    alignments that do not delete a character from \(\Xli\) under \(\wabcttn\).

    Next, observe that any alignment that deletes \(d_X\) characters from \(\Xli\) has to
    insert \(d_X + |Y| - |\Xli|\) characters from \(Y\). Now, even deleting a single
    character from \(d_X\) causes a cost increase of \(2\) (under \(\wabcttn\)), while
    yielding potential savings of at most \(L < 1\), and is hence worse that the alignment
    \(\alm^*\); completing the proof.
\end{proof}

We are finally ready to prove \cref{thm:dyn-main}.

\thmbatched
\begin{proof}
    Fix \(0 \le \scb \le 1\) and \(\smallconst > 0\), and suppose that there was an algorithm
    \(\alg\) that for every integer \(n \ge 1\) in total time
    \(\Oh(n^{2+\scb/2-\smallconst})\) solves
    every instance of the
    Batched Edit Distance problem on a batch
    of \(\Theta(n^{\scb})\) length-\(x\) strings
    and a string
    \(Y\) with \(x \le |Y| \le n\).

    We proceed to construct an instance of said form that can also be used to check if the
    minimum weight triangle in a complete tripartite graph has negative weight.
    To that end, write \(G =
    (P\cup Q\cup R, E, \wg)\) for a weighted, complete tripartite graph with
    \(p \coloneqq |P| = n\), \(q \coloneqq |Q| = n\),
    and \(r \coloneqq |R| = n^{\scb/2}\)
    (which we assume to be integer to simplify our exposition).
    Next, we construct matrices
    \(A \in \fragment{-E}{E}^{p\times q}\),
    \(B \in \fragment{-E}{E}^{q\times r}\) and
    \(C \in \fragment{-E}{E}^{r\times p}\) in the natural way by setting
    \begin{align*}
        A_{i,j} &\coloneqq \wg(P\position{i}, Q\position{j}),\\
        B_{j,\ell} &\coloneqq \wg(Q\position{j}, R\position{\ell}),\\
        C_{\ell,i} &\coloneqq \wg(R\position{\ell}, P\position{i}).
    \end{align*}

    Next, for \(\tau \coloneqq n^{1 - \scb/2}\) and writing \(\ptau \coloneqq
    \lceil p /\tau \rceil = \Theta(n^{\scb / 2})\), we define the alphabet
    \(\sabct\) as before, and observe
    \begin{align*}
        |\sabct| &= \Theta(p + q + r\tau + rp / \tau) = \Theta(n + n^{\scb}),\\
        {|\sabct|}^2 &= \Theta(n^{2}).
    \end{align*}
    Further, we also define the strings \(\Xli \in \sabct^*\) (for \(i \in
    \fragmentco{0}{\ptau}, \ell \in \fragment{0}{r}\)) and \(Y \in \sabct^*\), as
    well as \(\wabcttn\) as before. We observe that \(|\wabcttn| = {|\sabct|}^2 =
    \Theta(n^2)\).

    Next, we use \(\alg\) on the string \(Y\) and the batch
    \begin{align*}
        X_{0,0},\quad  X_{0,1}, \quad &\dots\quad X_{0, \ptau - 1},\\
        X_{1,0},\quad  X_{1,1}, \quad&\dots \quad X_{1, \ptau - 1},\\
        &\;\,\,\vdots\\
        X_{r-1,0},\quad X_{r-1,1},\quad&\dots \quad X_{r-1, \ptau - 1};
    \end{align*}
    where we set the threshold to\[
        k \coloneqq 2|Y| - |\Xli| + L,
    \] for \[
        L \coloneqq \left( |\Xli| F + (\tau + 1)I + \ptau(q + 1)D \right) / K.
    \]
    Observe that we have \(L \in \intervalco{0}{1}\) by choice of \(K\).

    \begin{claim}\label{cl:3-22-4}
        The constructed instance consists in \(\Theta(n^{\scb})\) strings.
        Hence, \(\alg\) solves the constructed instance in
        time \(\Oh(n^{2+\scb/2-\smallconst})\).

    \end{claim}
    \begin{claimproof}
        We
        compute
        \(
            r\ptau = \Theta(n^{\scb/2} \cdot n^{\scb/2}) = \Theta(n^{\scb}),
            \)
        yielding the claim.
    \end{claimproof}

    Next, we argue that the constructed instance indeed is able to check the sign of the minimum
    weight triangle in \(G\).

    \begin{claim}\label{cl:3-22-5}
        Write \(R_{\ell, i}\) for the weighted edit distance of the strings
        \(\Xli\) and \(Y\). The minimum weight triangle of \(G\) has a weight of\[
            K \cdot \min_{\ell, i}\{ R_{\ell, i} - L - 2|Y| + |\Xli| \}.
        \]
        In particular, we can obtain the weight of the minimum weight triangle from the
        minimum of all \(R_{\ell, i}\) (that is, the output of \(\alg\) on input the
        constructed instance),
        and we do not need to know each individual value \(R_{\ell, i}\).
    \end{claim}
    \begin{claimproof}
        By \cref{cl:3-22-2}, we have
        \begin{align*}
            R_{\ell, i} &= \ed^{\wabcttn}_{\alm^*}(\Xli, Y)\nonumber
            = L + \min_{\alpha \in
                    \fragmentco{0}{\tau}, j}\{A_{i + \alpha\ptau,j} + B_{j,\ell} +
            C_{\ell, i + \alpha\ptau}\} / K + 2 |Y| - |\Xli|.
        \end{align*}
        Hence, we have\[
            \min_{\ell, i}\{ (R_{\ell, i} - L - 2|Y| + |\Xli|)\cdot K \}
            = \min_{\ell, i, j}\{A_{i,j} + B_{j, \ell} + C_{\ell, i}\};
        \] which is indeed the weight of the minimum weight triangle in \(G\).
    \end{claimproof}

    Observe that \cref{cl:3-22-4,cl:3-22-5} together yield the main claim.
    We wrap up by discussing the claimed restrictions on the instance that we use to show
    hardness. To that end, observe that \(\wabcttn\) has already the desired properties.

    Next, we wish to bound the Hamming distance of two consecutive strings in the batch
    by \(\Oh(n^{1-\scb})\). Observe that currently, the only pairs of queries that
    violate this bound are of the form \(X_{\ell,\ptau - 1}, X_{\ell, 0}\),
    which have a Hamming distance of \(\tau + 1 = \Theta(n^{1 - \scb/2})\).
    However, we can easily circumvent this problem by inserting ``dummy strings'' into the
    batch that do not change the overall result.

    To that end, write \(\# \notin \sabct\) for a fresh character and modify \(\wabcttn\)
    so that every operation involving \(\#\) has a cost of \(2\).
    Let us revisit how we modify \(X_{\ell,\ptau - 1} \leadsto
    X_{\ell + 1, 0}\). In particular, we now proceed as follows
    \begin{align*}
        X_{\ell,\ptau - 1}\position{0}
        &\gets \#\\
        X_{\ell,\ptau - 1}\position{|X_{\ell,\ptau - 1}|-1}
        &\gets \#\\[1.5ex]
        X_{\ell,\ptau - 1}\position{1 + p + \alpha(2\ptau + 2)}
        &\gets X_{\ell,0}\position{1 + p + \alpha(2\ptau + 2)} = X_{\ell}^{(\alpha)}\position{\ptau}\\
        &\vdots\\
        X_{\ell,\ptau - 1}\position{1 + p + \alpha(2\ptau + 2)}
        &\gets X_{\ell,0}\position{1 + p + \alpha(2\ptau + 2)} = X_{\ell}^{(\alpha)}\position{\ptau}\\
        &\vdots\\
        X_{\ell,\ptau - 1}\position{1 + p + \alpha(2\ptau + 2)}
        &\gets X_{\ell,0}\position{1 + p + \alpha(2\ptau + 2)} =
        X_{\ell}^{(\alpha)}\position{\ptau}\\[1.5ex]
        X_{\ell,\ptau - 1}\position{0}
        &\gets X_{\ell,0}\position{0} = \S_{0}^{(\ell + 1,0)}\\
        X_{\ell,\ptau - 1}\position{|X_{\ell,\ptau - 1}|-1}
        &\gets X_{\ell,0}\position{|X_{\ell + 1, 0}| - 1} = \S_{\tau}^{(0)}
    \end{align*}
    Additionally, after every
    \(\Theta(n^{1-\scb})\) updates in the above sequence, we
    add a ``dummy string'' to the batch.

    As we changed the first and last character of \(\Xli\) (which might differ between different
    \(X\)'s anyway), we did not increase the sum of  Hamming distances of consecutive
    strings.
    In total, we add
    \(\Theta(r\tau /
    n^{1-\scb}h) = \Theta(n^{\scb/2 + (1-\scb/2) - (1-\scb)}) =
    \Theta(n^{\scb})\) strings to the batch (which now has the desired property on the
    Hamming distance of consecutive strings).

    As \(\alg\) returns just the minimum of weighted edit distance,
    we additionally need to argue that the weighted edit distance between a
    query string and \(Y\)
    is larger than the weighted edit distance of any non-dummy string an \(Y\).
    However, this is easy, as in
    any dummy string, we have to substitute two times the character~\(\#\)---which costs
    \(4\) in total (and is thus in total by at least \(1\) more expensive than substituting
    \(\S\)-characters); but any potential savings (compared to non-dummy strings) by
    choosing a potentially different alignment in between the \(\#\) characters amount to
    at most \(1\) (as a optimal alignment for every non-dummy string uses only operations
    that cost at most \(1 + 1/(|X| + |Y|)\) each).

    In total, this yields the desired result and hence, also completes the proof.
\end{proof}

%% file: sections/static_lower_bound.tex
\section{From Batched Weighted Edit Distance to Bounded Weighted Edit Distance}\label{sec:main-wed-lowerbound}

In a last step toward proving \cref{mthm:lb}, we show how to combine the strings of a
batch from the Batched Weighted Edit Distance problem into a single string.

To that end, consider an instance of the Batched Weighted Edit Distance Problem that
satisfies the conditions of \cref{thm:batched}
and write $h \coloneqq \max_{i=1}^{m-1}\hd(X_i,X_{i+1})$ and $r = (m-1)(h+4)+x+2y+1$.
Moreover, let $\Sigma_X$ consist of the characters that occur in any of the strings $X_i$,
and let $\Sigma_Y$ consist of the characters occurring in~$Y$.
Recall, that $\Sigma_X$ and $\Sigma_Y$ are disjoint (which is in fact also implied by
$\w{a}{\emptystring}=\w{\emptystring}{a}=1$ if $a\in \Sigma_X$
and $\w{a}{\emptystring}=\w{\emptystring}{a}=2$ if $a\in \Sigma_Y$).

We extend the existing alphabet by $2r+2$ distinct fresh symbols.
Formally, we set \[
    \hat{\Sigma} \coloneqq \Sigma_X \cup \Sigma_Y \cup \{u_1,\ldots,u_r,v_1,\ldots,v_r,\bot,\lozenge\}.
\] Further, we update the weight function
$\hat{w} : (\hat{\Sigma}\cup\{\emptystring\})\to \Real_{\ge 0}$ to
\[\hw{a}{b} \coloneqq \begin{cases}
    0 & \text{if }a = b,\\
    \w{a}{b} & \text{if }a,b\in \Sigma_X\cup\Sigma_Y \cup \{\emptystring\},\\
    1 & \text{otherwise}.
\end{cases}\]
Indeed, we \(\hat{w}\) is a useful weight function for our purposes.

\begin{lemma}\label{fct:hw:triangle}
    The weight function $\hat{w}$ is normalized, symmetric, and satisfies the triangle inequality.
    Moreover, all values of \(\hat{w}\) are rational numbers  in $\interval{0}{2}$ with
    a common $\Oh(\log n)$-bit denominator.
\end{lemma}
\begin{proof}
    By the conditions of \cref{thm:batched}, the weight function $w$ is normalized and
    symmetric, and all values of \(w\) are rational numbers in  $\interval{0}{2}$
    with a common $\Oh(\log n)$-bit denominator.

    Since all values $\hw{a}{b}$ are either equal to $0$ (if $a=b$), equal to $1$, or
    equal to $\w{a}{b}$,  the function $\hat{w}$ inherits all of the desired properties
    from \(w\).

    The triangle inequality holds trivially if the three characters $a,b,c$ are not distinct.
    If they are distinct, then $\hw{a}{b} \le 2 = 1 + 1 \le \hw{a}{c} + \hw{c}{b}$ holds
    by the construction of $\hat{w}$ and the conditions on $w$ listed in
    \cref{thm:batched}.
\end{proof}

As a next intermediate step, we define strings that lie in-between two consecutive strings
\(X_i\) and \(X_{i+1}\) in the batch. These in-between strings turn out to be useful as
the can be aligned to either \(X_i\) or \(X_{i+1}\) for roughly the same cost.

Formally, for $i\in \fragmentco{1}{m}$, we define $\Xb_i$ to be the string obtained from
$X_i$ by substituting with $\bot$ exactly $h$ characters
including all characters that contribute to $\hd(X_i,X_{i+1})$.
Similarly, we define the strings $\Xb_0$ and $\Xb_m$, which we obtain from $X_1$ and $X_m$, respectively,
by substituting with $\bot$ any $h$ characters.

Finally, also define the utility strings $U=u_1\cdots u_r$ and $V=v_1\ldots v_r$.
We are ready to combine the strings of the batch into a single pair of strings. We set
\[
    \hX \coloneqq \lozenge X_1 \lozenge \cdot \bigodot_{i=2}^{m} \left(U Y V \lozenge X_i \lozenge \right),\quad
    \hY \coloneqq \Xb_{0}\cdot \bigodot_{i=1}^{m} \left(U\lozenge Y \lozenge V
    \Xb_i\right),\quad\text{and}\quad \hk \coloneqq (m-1)(h+4)+ 2r + 2x + k.
\]

As a key lemma, we show that we can use the weighted edit distance of the
strings \(\hX\) and \(\hY\) to recover the minimum edit distance of any string \(X_i\) of
the batch to \(Y\).

\begin{lemma}\label{lem:equiv}
    We have $\min_{i=1}^m \wed(X_i,Y)\le k$ if and only if $\hwed(\hX,\hY) \le \hk$.
\end{lemma}
\begin{proof}
    We start with auxiliary claims that formalize our intuition about the strings
    \(\Xb_i\).

    \begin{claim}\label{clm:x}
        For every $i\in \fragment{1}{m}$, we have
        $\hwed(\lozenge X_i \lozenge,\Xb_{i-1})=\hwed(\lozenge X_i \lozenge,\Xb_{i})=h+2$.
    \end{claim}
    \begin{claimproof}
        By construction, we can obtain $\Xb_i$ by substituting with $\bot$ exactly $h$
        characters of $X_i$.
        Since $\hw{a}{\bot} = 1$ for every $a\in \Sigma_X$, this means that $\hwed(X_i,\Xb_i)\le h$.
        From $\hw{\lozenge}{\emptystring}=1$, we conclude that $\hwed(\lozenge X_i \lozenge, \Xb_i) \le h+2$.

        Next, observe that $\lozenge X_i \lozenge$ contains two occurrences of $\lozenge$
        and $x$ occurrences of characters in $\Sigma_X$, whereas $\Xb_i$ contains $h$
        occurrences of $\bot$ and $x-h$ occurrences of characters in $\Sigma_X$.
        Thus, at most $x-h$ characters of $\lozenge X_i \lozenge$ can be matched;
        the remaining $h+2$ characters much be deleted or substituted, which costs at
        least one unit each.
        Consequently, we have $\hwed(\lozenge X_i \lozenge, \Xb_i) \ge h+2$.

        The argument regarding $\hwed(X_i,\Xb_{i-1})$ is exactly the same.
    \end{claimproof}

    \begin{claim}\label{clm:y}
        For every $i\in \fragment{0}{m}$,
        we have $\hwed(Y, \lozenge Y\lozenge V \Xb_i U \lozenge Y \lozenge)=4+2r+x+2y$
        and $\hwed(Y, \lozenge Y \lozenge) = 2$.
    \end{claim}
    \begin{claimproof}
        By \cref{fct:hw:triangle,fct:substring}, we have $\hwed(Y,\lozenge Y \lozenge) = \hwed(\emptystring,\lozenge\lozenge)
        = 2\hw{\emptystring}{\lozenge} = 2$.
        The same argument yields
        \begin{align*}
            \hwed(Y, \lozenge Y\lozenge V \Xb_i U \lozenge Y \lozenge)
            &=\hwed(\emptystring,\lozenge\lozenge V \Xb_i U \lozenge Y \lozenge)\\
            &= 4\hw{\emptystring}{\lozenge}+\hwed(\emptystring,V)+\hwed(\emptystring,U)
            + \hwed(\emptystring,\Xb_i)+\hwed(\emptystring,Y)\\
            &= 4 + 2r + x + 2y.
        \end{align*}
        As for the last equality, we utilize the fact that $\hw{\emptystring}{a} = \w{\emptystring}{a} = 2$ if $a\in \Sigma_Y$
        and $\hw{\emptystring}{a}=1$ if $a\in\hat{\Sigma}\sm \Sigma_Y$; this follows from
        the definition of $\hat{w}$ and the properties of $w$ listed in
        \cref{thm:batched}; in particular, $\w{\emptystring}{a}=\w{a}{\emptystring}$ if
        $a\in \Sigma_X$.
    \end{claimproof}

    \begin{claim}\label{clm:kk}
        For every $i\in \fragment{1}{m}$, we have
        $\hwed(\lozenge X_i \lozenge, \Xb_{i-1} U\lozenge  Y \lozenge V \Xb_i) \le 2r+2x+k$ if and only if $\wed(X_i,Y)\le k$.
    \end{claim}
    \begin{claimproof}
        As for the $\Leftarrow$ implication, observe that  if $\wed(X_i,Y)\le k$, then
        \begin{align*}
            &\hwed(\lozenge X_i \lozenge, \Xb_{i-1} U\lozenge  Y \lozenge V \Xb_i)\\
            &\quad\le \hwed(\emptystring,\Xb_{i-1}U) + \hw{\lozenge}{\lozenge} +
            \hwed(X_i,Y)+\hw{\lozenge}{\lozenge}+\hwed(\emptystring,V\Xb_i)\\
            &\quad= x+r + \wed(X_i,Y) + r+x \le 2r+2x+k.
        \end{align*}
        As for intermediate equality, we utilize the fact that $\hw{\emptystring}{a}=1$ if
        $a\in\hat{\Sigma}\sm \Sigma_Y$; this follows from the definition of $\hat{w}$ and
        the properties of $w$ listed in \cref{thm:batched}; in particular,
        $\w{\emptystring}{a}=\w{a}{\emptystring}$ if $a\in \Sigma_X$.
        Moreover, $\hwed(X_i,Y)=\wed(X_i,Y)$ because $\hat{w}$ and $w$ coincide on
        $\Sigma_X\cup\Sigma_Y\cup\{\emptystring\}$.

        Now, suppose that $\hwed(\lozenge X_i \lozenge, \Xb_{i-1} U\lozenge  Y \lozenge V \Xb_i) \le 2r+2x+k$
        and consider an underlying optimal alignment $\A$.
        Suppose that \(\A\) inserts precisely $q$ characters of $Y$.
        The overall number of insertions is at least
        \[|\Xb_{i-1} U\lozenge  Y \lozenge V \Xb_i|-|\lozenge X_i \lozenge|
        = 2r + x + y.\]
        Moreover, since $\hw{\epsilon}{a}=2$ if $a\in \Sigma_Y$, the total cost of these insertions
        is at least $q+2r+x+y$.

        Since $\Sigma_Y$ is disjoint with $\Sigma_X\cup\{\lozenge\}$, the remaining $y-q$ characters of $Y$
        must be be involved in substitutions, each costing at least $1$.
        Overall, the total cost of all insertions and of substitutions involving characters of $Y$
        is at least $2r+x+2y$.

        However the cost of $\A$ is at most \[
            2r+2x+k < 2r+2x+2y-x+1 = 2r+x+2y+1.
        \]
        Thus, the alignment cannot make any edit other than those listed above.
        In particular, none of the characters of $\lozenge X_i \lozenge$
        can be deleted or substituted for a character other than a character of $Y$.
        In particular, the both $\lozenge$s in $\lozenge X_i \lozenge$ must be aligned against
        characters of $\lozenge Y \lozenge$,
        which means that all the characters of $\Xb_{i-1} U$ and $V \Xb_i$ need to be inserted.
        Consequently,
        \[
            \hwed_\A(\lozenge X_i \lozenge, \Xb_{i-1} U\lozenge  Y \lozenge V \Xb_i) \ge 2r+2x + \hwed(\lozenge X_i \lozenge,\lozenge  Y \lozenge).
        \]
        Due to the assumption on the cost of $\A$,
        we conclude that $\hwed(\lozenge X_i \lozenge,\lozenge  Y \lozenge)\le k$.
        Since \cref{fct:hw:triangle,cor:greedy} imply $\hwed(\lozenge X_i \lozenge,\lozenge  Y \lozenge)=\hwed(X_i,Y)$,
        this yields $\wed(X_i,Y)=\hwed(X_i,Y)\le k$.
    \end{claimproof}

    We are now ready to proceed with the proof of \cref{lem:equiv}.
    First, suppose that $\wed(X_i,Y)\le k$ for some $i\in \fragment{1}{m}$.
    In that case, we use \cref{clm:x,clm:y,clm:kk} as follows:
    \begin{align*}
        \hwed(\hX,\hY) &\le \sum_{j=1}^{i-1} \left(\hwed(\lozenge X_j \lozenge, \Xb_{j-1}) + \hwed(U Y V, U \lozenge Y \lozenge V)\right) + \hwed(\lozenge X_i \lozenge,\Xb_{i-1}U \lozenge Y \lozenge V \Xb_{i}) \\ &\pushright{+ \sum_{j=i+1}^m \left(\hwed(U Y V, U\lozenge Y \lozenge V) + \hwed(\lozenge X_j \lozenge, \Xb_j)\right)}\\
                       &\le \sum_{j=1}^{i-1} \left(\hwed(\lozenge X_j \lozenge, \Xb_{j-1}) + \hwed(Y, \lozenge Y \lozenge)\right) + 2r + 2x + k \\ &\pushright{+ \sum_{j=i+1}^m \left(\hwed(Y, \lozenge Y \lozenge) + \hwed(\lozenge X_j \lozenge, \Xb_j)\right)}\\
                       & = \sum_{j=1}^{i-1} \left(h+2 + 2\right) +  2r + 2x + k + \sum_{j=i+1}^m \left(2+h+2\right) = (m-1)(h+4) + 2r + 2x + k = \hk.
    \end{align*}

    As for the converse implication, suppose that $\hwed(\hX,\hY) \le \hk$ and consider
    the underlying optimal alignment.
    Observe that we have\[
        \hk = (m-1)(h+4)+ 2r + 2x + k < 2r + (m-1)(h+4) + 2x + 2y - x + 1 = 3r.
    \]
    Thus, by the pigeonhole principle, there is a $q\in \fragment{1}{r}$ such that at most
    two edits involve the occurrences of $u_q$ and $v_q$ in $\hY$.

    Observe that these characters form a subsequence of $\hY$ equal to $(u_q v_q)^{m}$,
    whereas the subsequence of $\hX$ consisting of the occurrences of $u_q$ and $v_q$ is
    equal to $(u_qv_q)^{m-1}$. Thus, all the occurrences of $u_q$ and $v_q$ in $\hX$ must
    be matched exactly. The only way to achieve this is by leaving out two \emph{adjacent}
    characters of $(u_q v_q)^{m}$. By \cref{cor:greedy,fct:hw:triangle}, whenever an
    occurrence of $u_q$ or $v_q$ in $\hX$ is matched exactly, we can assume without loss
    of generality that the entire occurrences of $U$ or $V$, respectively, is matched
    exactly.
    Thus, it suffices to consider the following two cases:

    \subparagraph*{Case 1.} For some $i\in \fragment{1}{m}$, all but the $i$-th leftmost
    occurrences of $U$ and $V$ in $\hY$ are matched exactly.
    In this case, we use \cref{clm:x,clm:y} as follows:
    \begin{align*}
        \hwed_\A(\hX,\hY) &\le \sum_{j=1}^{i-1} \left(\hwed(\lozenge X_j \lozenge, \Xb_{j-1}) + \hwed(Y, \lozenge Y \lozenge)\right) + \hwed(\lozenge X_i \lozenge,\Xb_{i-1}U \lozenge Y \lozenge V \Xb_{i}) \\ &\pushright{+ \sum_{j=i+1}^m \left(\hwed(Y,\lozenge Y \lozenge) + \hwed(\lozenge X_j \lozenge, \Xb_j)\right)}\\
                          & = \sum_{j=1}^{i-1} \left(h+2 + 2\right) + \hwed(\lozenge X_i \lozenge,\Xb_{i-1}U \lozenge Y \lozenge V \Xb_{i}) + \sum_{j=i+1}^m \left(2+h+2\right) \\
                          & = (m-1)(h+4) + \hwed(\lozenge X_i \lozenge,\Xb_{i-1}U \lozenge Y \lozenge V \Xb_{i}).
    \end{align*}
    Since $\hwed_\A(\hX,\hY)\le \hk = (m-1)(h+4) + 2r + 2x + k$, we conclude that \[
        \hwed(\lozenge X_i \lozenge,\Xb_{i-1}U \lozenge Y \lozenge V \Xb_{i})\le 2r+2x+k.
    \]By \cref{clm:kk}, this implies $\wed(X_i,Y)\le k$ just as claimed.

    \subparagraph*{Case 2.} For some $i\in \fragmentco{1}{m}$, all but the $i$th leftmost
    occurrences of $U$ and the $(i+1)$-th leftmost occurrence of $V$ in $\hY$ are matched
    exactly. In this case, we use \cref{clm:x,clm:y} as follows:
    \begin{align*}
        \!\hwed_\A(\hX,\hY) &\le \sum_{j=1}^{i-1} \left(\hwed(\lozenge X_j \lozenge, \Xb_{j-1})+\hwed(Y,\lozenge Y \lozenge)\right) + \hwed(\lozenge X_i \lozenge, \Xb_{i-1}) + \hwed(Y, \lozenge Y \lozenge V \Xb_i U \lozenge Y \lozenge)\\
                            & \pushright{ +\hwed(\lozenge X_{i+1} \lozenge, \Xb_{i+1})+\sum_{j=i+2}^{m} \left(\hwed(Y,\lozenge Y \lozenge)+\hwed(\lozenge X_j \lozenge, \Xb_j)\right)}\\
                            & = \sum_{j=1}^{i-1}\left(h+2+2\right) + (h+2) + (4+2r+x+2y) + (h+2) + \sum_{j=i+2}^m (2+h+2) \\
                            & = (i-1)(h+4) + (h+4) + (2r+x+2y) + (h+4) + (m-i-1)(h+4) \\
                            & = m(h+4) + 2r + x + 2y \\
                            & < (m-1)(h+4) + 2r + 2x + 2y - x + 1 < (m-1)(h+4) + 2r + 2x+k = \hk.
    \end{align*}
    This contradicts the assumption that $\hwed_\A(\hX,\hY)\le \hk$, which concludes the proof.
\end{proof}

Having completed the analysis of our construction, we proceed with the proof of our main
hardness result.

\mthmlb
\begin{proof}
    For a proof by contradiction, suppose that there is an algorithm violating the
    theorem statement.
    Our goal is to derive an algorithm for the Batched Weighted Edit Distance problem
    violating \cref{thm:batched}
    with $\beta = ({1}/{\kappa}) -1 \in \interval{0}{1}$. The latter algorithm works as follows.

    Given an instance $(\Sigma, w, X_1,\ldots,X_m,Y,k)$ of the Batched Weighted Edit
    Distance problem (satisfying the conditions of \cref{thm:batched}),
    we construct an instance $(\hat{\Sigma},\hat{w}, \hX,\hY,\hk)$ of the bounded edit
    distance problem as defined above
    and solve it using the hypothetical algorithm.
    By \cref{lem:equiv}, the answer is valid for the original Batched Weighted Edit
    Distance problem.

    Let us analyze the running time, assuming that $x\le y \le n$, $m \le n^{\beta}$, and
    $h\le n^{1-\beta}$.
    Observe that we have $r = (m-1)(h+4)+x+2y+1 = \Oh(n)$, and thus
    \begin{align*}
        \left|\hat{\Sigma}\right| &= |\Sigma_X|+|\Sigma_Y| + 2r+2 \le x + (m-1)h + y + \Oh(n) = \Oh(n),\\
        \hk &= (m-1)(h+4)+2r+2x+k < (m-1)(h+4)+2r+x+2y+1 = \Oh(n),\\
        |\hX| &= (1+x+1)+(m-1)(r+y+r+1+x+1) = \Oh(nm) = \Oh(n^{1+\beta}),\\
        |\hY| &= x + m(r+1+y+1+r+x) = \Oh(nm) = \Oh(n^{1+\beta}).
    \end{align*}
    Hence, there is $N=\Theta(n^{1+\beta})$ such that $|\hX|,|\hY|\le N$
    and $\hk,|\hat\Sigma| \le N^{1/(1+\beta)} = N^{\kappa}$.

    The hypothetical algorithm works in time \[\Oh\left(N^{0.5 + 1.5\kappa -
        \smallconst}\right)=
    \Oh\left(n^{(1+\beta)(0.5+1.5\kappa-\smallconst)}\right)=\Oh\left(n^{0.5+0.5\beta +
1.5-(1+\beta)\smallconst}\right) = \Oh\left(n^{2+0.5\beta - \smallconst}\right)\]
    This dominates the reduction cost, which is $\Oh(n^{1+\beta})$ time assuming (without
    loss of generality) that $\smallconst < 0.5 \le 1-0.5\beta$.

    It remains to argue about the restrictions in the theorem statement. We have already
    noted that $|\hat\Sigma|\le N^{\kappa}$.
    The properties of the weight function $\hat{w}$ follow from \cref{fct:hw:triangle};
    the common denominator of all the values of $\hat{w}$ is $\Oh(\log n)=\Oh(\log N)$
    because $N = \Theta(n^{1+\beta})$ for $\beta\ge 0$.
\end{proof}